\documentclass[11pt]{article}

\usepackage[numbers,sort&compress]{natbib}

\usepackage{tocloft}

\usepackage[title=article,pagegeom=wide,secfmt={},bibliographydefs=false]{phfnote}
\usepackage[smallproofs=false,thmset=empty]{phfthm}
\usepackage{phfqit}
\usepackage{phfparen}

\usepackage{mathptmx}

\usepackage[nameinlink,capitalize,noabbrev]{cleveref}

\phfLoadThmSet{[%
  parentcounter=section,%
  thmstyle=plain%
]}{[%
  parentcounter=section,%
  thmstyle=definition,%
  proofref=false%
]}{default}

\makeatletter

\def\Hfn{\Hfnbase{H}{}{}}
\def\Sfn{\Hfn} \def\HH{\Hbase{H}{}}
\def\Hzero{\Hbase{\hat{H}}{\mathrm{max}}}

\newcommand{\Dhyp}[1][\eta]{\Dbase{\DSym}_{\mathrm{h}}^{#1}}

\DefineTunedQitObject{DCohz}{\DCohbaseParse}{\DCohbaseRender}{{\hat\DSym}}{\let\DCohbaseStateSubscripts\DCohxStateSubscripts }
\DefineTunedQitObject{WProc}{\DCohbaseParse}{\DCohbaseRender}{{W}}{\let\DCohbaseStateSubscripts\DCohProcStateSubscripts \let\DCohbaseRenderSystems\fixDCohbaseRenderSystems}

\newcommand{\MM}{\Dbase{M}}

\DeclareMathOperator\YoungOp{Young}

\DeclareMathOperator\SHIFTop{\mathsf{SHIFT}}

\def\Young#1#2{\YoungOp(#1,#2)}

\def\mapinlambda#1[#2]{\;[\;{#1}\;]_{#2}\; }

\def\equationautorefname~#1\null{Eq.~(#1)\null }

\let\@oldunderscore\_
\def\_{\begingroup\fontfamily{txr}\selectfont\@oldunderscore\endgroup}

\makeatother 

\makeatletter
\appto\phfthm@hookproof@proof@startafterdisplay{%
  \parskip=0.5em%
  \parindent=0pt%
  \sloppy%
}

\renewcommand*{\eqref}[1]{%
  \hyperref[{#1}]{\textup{\tagform@{\ref*{#1}}}}%
}
\makeatother

\usepackage{ragged2e}
\DeclareCaptionJustification{myjust}{\justify}
\makeatletter
\AtBeginDocument{\immediate\write\@auxout{\string\citation{apsrev41Control}}}
\makeatother
\def\bibliographysetup{%
  \def\bibitemNoStop{}%
  \def\selectlanguage##1{}%
  \bibsep=2pt\relax
  \def\mybibliographyfont{}%
  \apptocmd{\thebibliography}{\def\ {\unskip\space}\mybibliographyfont}{}{}%
}

\numberwithin{equation}{section}

\begin{document}

\title{Thermodynamic Implementations of Quantum Processes}

\author{Philippe Faist,$^{1,2,3}$\quad Mario Berta,$^{4,5}$\quad
  and\quad Fernando G.\@ S.\@ L.\@ Brandao$^{1,5}$\\[1ex]
  \emph{\small $^1$%
    Institute for Quantum Information and Matter, Caltech, Pasadena, CA 91125, U.S.A.}\\
  \emph{\small $^2$%
    Institute for Theoretical Physics, ETH Zurich, 8093 Zurich, Switzerland}\\
  \emph{\small $^3$%
    Dahlem Center for Complex Quantum Systems, Freie Universit\"at Berlin, 14195 Berlin, Germany}\\
  \emph{\small $^4$%
    Department of Computing, Imperial College London, London SW7 2AZ,
    United Kingdom}\\
  \emph{\small $^5$%
    AWS  Center  for  Quantum  Computing,  Pasadena,  CA  91125,  U.S.A.}%
}

\date{July~21, 2021}

\maketitle

\begin{abstract}
  Recent understanding of the thermodynamics of small-scale systems have enabled
  the characterization of the thermodynamic requirements of implementing quantum
  processes for fixed input states. Here, we extend these results to construct
  optimal universal implementations of a given process, that is, implementations
  that are accurate for any possible input state even after many independent and
  identically distributed (i.i.d.\@) repetitions of the process.  We find that the
  optimal work cost rate of such an implementation is given by the thermodynamic
  capacity of the process, which is a single-letter and additive quantity
  defined as the maximal difference in relative entropy to the thermal state
  between the input and the output of the channel.
  As related results we find a new single-shot implementation of time-covariant
  processes and conditional erasure with nontrivial Hamiltonians,
  a new proof of the asymptotic
  equipartition property of the coherent relative entropy, and an optimal
  implementation of any i.i.d.\@ process with thermal operations for a fixed
  i.i.d.\@ input state.
  Beyond being a thermodynamic analogue of the reverse Shannon theorem for
  quantum channels, our results introduce a new notion of quantum typicality
  and present a thermodynamic application of convex-split methods.
\end{abstract}

\tableofcontents
\section{Introduction}

In the information-theoretic approach to thermodynamics, a careful analysis of the
resources required to perform thermodynamic tasks has allowed to consistently
and systematically describe the thermodynamic behaviour of quantum systems at
the nano-scale \cite{Goold2016JPA_review}. In particular, thermodynamics can be
phrased as a resource theory~\cite{Brandao2013_resource,Brandao2015PNAS_secondlaws,Chitambar2019RMP_resource}.
In a resource theory, one specifies which operations can be carried out at no
cost\,---\,the \emph{free operations}\,---\,and then one studies how much of
external resources (e.g., thermodynamic work) one needs to provide to carry out
operations that are not free. Two established resource theories for quantum
thermodynamics are \emph{thermal operations}~\cite{Brandao2013_resource,Brandao2015PNAS_secondlaws} and
\emph{Gibbs-preserving maps}~\cite{Janzing2000_cost,Faist2015NJP_Gibbs}.  In the
former, the free operations consist of energy-conserving interactions of the
system with a heat bath, while in the latter, the free operations are any
quantum operation that preserves the thermal state.  It is reasonable to assume
that thermal operations can be realized in an idealized setting, making them a
good choice of framework for constructing explicit protocols, whereas
Gibbs-preserving maps encompass a broader class of operations, allowing us to
derive stronger fundamental limits.

The resource theory approach to thermodynamics has revealed close connections with measures of information known from quantum information theory~\cite{Aberg2013_worklike,Horodecki2013_ThermoMaj}. Namely, single-shot thermodynamic and information-theoretic tasks are both quantified by relevant entropy measures~\cite{PhDRenner2005_SQKD,PhDTomamichel2012,BookTomamichel2016_Finite}.  Consequently, tools from quantum Shannon theory can be used to characterize tasks in thermodynamics, for instance to derive second-order asymptotics of the work cost of state transformations~\cite{Chubb2018Qu_beyond}. Recently, focus was shifted to understand the resource costs of quantum processes, rather than state transformations~\cite{Faist2015NatComm,Cirstoiu2017arXiv_gauge,Dana2017PRA_beyond,Faist2018PRX_workcost}. The information measure associated with quantum processes is the quantum capacity, along with its many variants~\cite{Smith10}. A natural question arises: What is the thermodynamic analogue of the quantum capacity?

Here, we ask how much work is required to implement a given quantum process, with
the requirement that the implementation is accurate for any possible input
state. In the single-instance regime, we find that the answer is a variation of
the results obtained in Ref.~\cite{Faist2018PRX_workcost}. However, in the
regime where we consider many independent and identically distributed (i.i.d.)
copies of the process, important differences arise due to typicality. We find
that the optimal work cost of such an implementation in the i.i.d.\@ regime is
given by the thermodynamic capacity, defined as the maximal difference between
the input and output free energy of the process over all possible input states. The fact that no implementation can perform better than the thermodynamic capacity follows fairly straightforwardly from the results of
Ref.~\cite{Faist2018PRX_workcost}. The technically challenging part of the
present paper is to show that there exist protocols that achieve this limit.

We provide three different constructions of such protocols, each valid
in different settings. In the first construction, we make the simplifying
assumption that Hamiltonian of the system is trivial as in
Ref.~\cite{Faist2015NatComm}. We then show that simple properties of one-shot
entropy measures, coupled with the post-selection
technique~\cite{Christandl2009PRL_Postselection}, provide an existence proof of
the required implementation. The implementation is given in terms of thermal
operations. In our second construction, we develop novel quantum typicality tools
which we use along with the post-selection technique to explicitly construct an
implementation in terms of Gibbs-preserving maps for any i.i.d.\@ process and
for any system Hamiltonian. In our third construction, we assume that the
i.i.d.\@ process is time-covariant, i.e., commutes with the time evolution. We then
use recent results on the convex-split lemma and position-based
decoding~\cite{Anshu2019IEEETIT_oneshot} to construct an implementation of a
time-covariant i.i.d.\@ process with thermal operations.

Our results imply that the thermodynamic resource theory of channels becomes
reversible in the i.i.d.\@ limit~\cite{OurShortPaperToAppear}. Namely, invoking the
results in Ref.~\cite{Navascues2015PRL_nonthermal}, we see that the work rate
that is required to implement a given i.i.d.\@ process is the same as what can
be extracted if the i.i.d.\@ process is provided to us as a black box. This provides a thermodynamic analogue of the reverse Shannon theorem from quantum information theory. This theorem states that the quantum mutual information of the channel uniquely characterizes the resources required to simulate the channel
with noiseless channel uses and shared entanglement, as well as to distil a
noiseless channel from many uses of the channel and shared entanglement~\cite{Bennett2014_reverse,Berta2011_reverse}. Indeed, our proof techniques are inspired by Refs.~\cite{PhDHarrow2005,Bennett2014_reverse,Haah2017IEEETIT_sampleoptimal,Noetzel2012arXiv_two}.

Finally, we provide some additional results that can be of independent interest:
A significantly simpler proof of the asymptotic equipartition property of the
quantity studied in Ref.~\cite{Faist2018PRX_workcost}, a new protocol for
Landauer erasure with side information for systems with a general
non-interacting Hamiltonian and in a time-covariant state, and a protocol for
implementing any general i.i.d.\@ process for a fixed i.i.d.\@ input state,
using thermal operations and a small source of coherence.

The remainder of this paper is structured as follows. Section \ref{sec:prelim} gives
the necessary preliminaries and fixes some
notation. Section \ref{sec:resource-theory-thermodynamics} introduces two resource
theories for thermodynamics, thermal operations and Gibbs-preserving maps. In
Section \ref{sec:thermodyn-capacity} we introduce the thermodynamic capacity and
present some elementary properties. In
Section \ref{sec:approach-using-quasi-convexity}, we provide our first construction for
a trivial Hamiltonian. In Section \ref{sec:proof-Weps-iid} we provide our second
construction, which is valid in the general setting and provides an
implementation in terms of Gibbs-preserving
maps. Section \ref{sec:optimal-universal-protocol-TO-timecovariant} provides our third
construction, valid for time-covariant i.i.d.\@ processes, and built with
thermal operations.
In \cref{sec:bonus-results}, we provide three related results that use
techniques developed in the above constructions.
Our conclusions are presented in Section \ref{sec:discussion}. Various more technical proof details are deferred to Appendices \ref{appx:first}--\ref{appx:last}.

\section{Preliminaries}\label{sec:prelim}

\subsection{Quantum states, quantum processes, and distance measures}

Each quantum system considered lives in a finite-dimensional Hilbert space. A
quantum state is a positive semi-definite operator $\rho$ satisfying
$\tr[\rho]=1$.  A sub-normalized quantum state is a positive semi-definite
operator $\rho$ satisfying $\tr[\rho]\leqslant 1$.  To each system $S$ is
associated a standard basis, usually denoted by $`{ \ket{k}_S }$.  For any two
systems $A,A'$, we denote by $A\simeq A'$ the fact that they are isometric. In
that case, we consider a representation in which the isometry maps the standard
basis onto the standard basis, i.e., $\IdentProc[A][A']{\proj{k}_A} = \proj{k}_{A'}$ for all $k$, where
$\IdentProc[A][A']{}$ denotes the identity process. For any two systems
$A\simeq A'$, we define the non-normalized maximally entangled reference ket
$\ket\Phi_{A:A'} = \sum_k\, \ket k_A\otimes \ket k_{A'}$.  Matrix inequalities
are with respect to the positive semi-definite cone: $A\leqslant B$ signifies
that $B-A$ is positive semi-definite. A completely positive map
$\mathcal{E}_{X\to X'}$ is a linear mapping that maps Hermitian operators on $X$
to Hermitian operators on $X'$ and that satisfies
$\mathcal{E}_{X\to X'}(\Phi_{X:R_X}) \geqslant 0$, where $R_X\simeq X$. The
adjoint $\mathcal{E}_{X\leftarrow X'}^\dagger$ of a completely positive map
$\mathcal{E}_{X\to X'}$ is the unique completely positive map $X'\to X$ that
satisfies $\tr[\mathcal{E}(Y)\,Z] = \tr[Y\mathcal{E}^\dagger(Z)]$ for all
operators $Y,Z$.  A completely positive map $\mathcal{E}_{X\to X'}$ is
trace-preserving if $\mathcal{E}^\dagger(\Ident_{X'}) = \Ident_X$ and trace
non-increasing if $\mathcal{E}^\dagger(\Ident_{X'}) \leqslant \Ident_X$.

Proximity of quantum states can be measured by the fidelity
$F(\rho,\sigma) = \norm{\sqrt\rho\sqrt\sigma}_1$, where the one-norm of an
operator is defined as $\norm{A}_1 = \tr\big[\sqrt{A^\dagger A}\big]$. The fidelity is
extended to sub-normalized states $\rho,\sigma$ as the generalized fidelity,
$\bar{F}(\rho,\sigma) = \norm{\sqrt\rho\sqrt\sigma}_1 +\sqrt{(1-\tr[\rho])(1-\tr[\sigma])}$, noting that
$F(\cdot,\cdot) = \bar{F}(\cdot,\cdot)$ whenever at least one of the states is
normalized. An associated metric can be defined for any sub-normalized states
as $P(\rho,\sigma) = \sqrt{1-\bar{F}^2(\rho,\sigma)}$, called the purified
distance~\cite{Tomamichel2010IEEE_Duality,PhDTomamichel2012,BookTomamichel2016_Finite}, or root infidelity, and is
closely related to the Bures distance and the quantum
angle~\cite{BookNielsenChuang2000}. The proximity of two sub-normalized quantum
states $\rho,\sigma$ may also be measured in the trace distance
$D(\rho,\sigma) = \frac12\norm{\rho-\sigma}_1$. We note that the one-norm of a
Hermitian operator $A$ can be expressed as
\begin{align}
\norm{A}_1 = \max_{\norm{Z}_\infty\leqslant 1} \tr[ZA]= \min_{\substack{\Delta_\pm\geqslant 0\\A = \Delta_+ - \Delta_-}}\tr[\Delta_+] + \tr[\Delta_-]\ ,\label{eq:one-norm-primal-dual}
\end{align}
where the first optimization ranges over Hermitian $Z$ operators and where the
second over positive semi-definite operators $\Delta_\pm$.  For any two states
$\rho,\sigma$ (one can even be sub-normalized), the purified distance and
the trace distance are related via
\begin{align}
D(\rho,\sigma) \leqslant P(\rho,\sigma) \leqslant \sqrt{2D(\rho,\sigma)}\ .
\end{align}
Similarly, we may define a distance measure for channels: For two completely
positive, trace non-increasing maps $\mathcal{T}_{X\to X'}$ and
$\mathcal{T}'_{X\to X'}$, the diamond norm distance is defined as
\begin{align}
\frac12\norm*{\mathcal{T}_{X\to X'} - \mathcal{T}'_{X\to X'}}_\diamond= \max_{\sigma_{XR}}D`\big(\mathcal{T}_{X\to X'}(\sigma_{XR}), \mathcal{T}'_{X\to X'}(\sigma_{XR}))\ ,
\end{align}
where the optimization ranges over all bipartite quantum states over $X$ and a
reference system $R\simeq X$.  The optimization may be restricted to pure states
without loss of generality.

\subsection{Entropy measures}

The von Neumann entropy of a quantum state $\rho$ is $\Sfn(\rho) = -\tr`[\rho\ln\rho]$. In this work, all entropies are defined in units of nats, using the natural logarithm $\ln(\cdot)$, instead of units of (qu)bits. A number of nats is equal to $\ln(2)$ times the corresponding number of qubits. The conditional von Neumann entropy of a bipartite state $\rho_{AB}$ is
given by
\begin{align}
  \HH[\rho]{A}[B] = \HH[\rho]{AB} - \HH[\rho]{B} = \Sfn(\rho_{AB}) - \Sfn(\rho_B)\ .
\end{align}
The quantum relative entropy is defined as
\begin{align}
  \DD{\rho}{\sigma} = \tr`\big[\rho`\big(\ln\rho - \ln\sigma)]\ ,
\end{align}
where $\rho$ is a quantum state and where $\sigma$ is any positive semi-definite
operator whose support contains the support of $\rho$.

One of our proofs relies on the hypothesis testing relative
entropy~\cite{Wang2012PRL_oneshot,Tomamichel2013_hierarchy,Matthews2014IEEETIT_blocklength,Buscemi2010IEEETIT_capacity,Brandao2011IEEETIT_oneshot}
in its form as presented in~\cite{Dupuis2013_DH}.  For any sub-normalized
quantum state $\rho$ and for any positive semi-definite operator $\sigma$ whose
support contains the support of $\rho$, we define it via the following
equivalent optimizations, which are semi-definite
programs~\cite{Watrous2009_sdps} in terms of the primal variable $Q\geqslant 0$
and the dual variables $\mu, X\geqslant 0$:
\begin{align}
  \exp`*{
  -\DHyp[\eta]{\rho}{\sigma}
  }
  \qquad
  =\qquad \mathrm{minimize:}\quad
  &
    \eta^{-1}\,\tr[Q\sigma]
    \\
  \mathrm{subject\ to:}\quad
  &
    Q \leqslant \Ident
    \nonumber\\
  &
    \tr[Q\rho] \geqslant \eta
    \nonumber
  \\[1ex]
  =\qquad \mathrm{maximize:}\quad
  &
    \mu - \eta^{-1}\tr[X]
    \label{eq:defn-hypothesis-testing-entropy--dual}
  \\
  \mathrm{subject\ to:}\quad
  &
    \mu\rho \leqslant \sigma + X\ .
    \nonumber
\end{align}
We also define for convenience the closely related quantity
\begin{align}
  \Dhyp[\eta]{\rho}{\sigma} =
  \DHyp[\eta]{\rho}{\sigma} - \ln(\eta)\ ,
\end{align}
which is sometimes also referred to as the hypothesis testing relative entropy.

\subsection{Schur-Weyl duality}

Consider a Hilbert space $\Hs_A$ and $n\in\mathbb{N}$. The group
$\GL(d_A)\times \SN(n)$ acts naturally on $\Hs_A^{\otimes n}$, where
$X\in\GL(d_A)$ acts as $X^{\otimes n}$ and where the permutation group permutes
the tensor factors.  We follow closely the notation of
Refs.~\cite{PhDHarrow2005,Haah2017IEEETIT_sampleoptimal}. Schur-Weyl tells us
that the full Hilbert space decomposes as
\begin{align}
  \Hs_A \simeq
  \bigoplus_\lambda\; \mathcal{V}_\lambda
  =
  \bigoplus_\lambda\; \mathcal{Q}_\lambda\otimes\mathcal{P}_\lambda\ ,
  \label{eq:Schur-Weyl-decomp-space}
\end{align}
where $\lambda\in\Young{n}{d}$ are Young diagrams with $n$ boxes and (at most)
$d$ rows, and where $\mathcal{Q}_\lambda$, $\mathcal{P}_\lambda$ are irreducible
representations of $\GL(d_A)$ and $\SN(n)$, respectively. The number of Young
diagrams in the decomposition above is at most $\poly(n)$, if $d_A$ is kept
constant.  We write $\poly(n)=O(\poly(n))$ in big O notation for terms whose
absolute value is upper bounded by some polynomial $n^c$ for $c\in\mathbb{N}$ in
the asymptotic limit $n\to\infty$.

We denote by $\Pi_{A^n}^\lambda$ the projector in $\Hs_A^{\otimes n}$
onto the term labelled by $\lambda$ in the decomposition above.  We denote by
$q_\lambda(X)$ a representing matrix of $X\in\GL(d_A)$ in the irreducible
representation labelled by $\lambda$; the operator $q_\lambda(X)$ lives in
$\mathcal{Q}_\lambda$.  We furthermore introduce the following notation, for any
$Y\in\mathcal{Q}_\lambda\otimes\mathcal{P}_\lambda$,
\begin{align}
  \mapinlambda{Y}[\lambda] =
  \Ident_{(\mathcal{Q}_\lambda\otimes\mathcal{P}_\lambda) \to A^n}
  \; Y \;
  \Ident_{(\mathcal{Q}_\lambda\otimes\mathcal{P}_\lambda) \leftarrow A^n}^\dagger\ ,
\end{align}
which represents the canonical embedding of an operator $Y$ on
$\mathcal{Q}_\lambda\otimes\mathcal{P}_\lambda$ into the space
$\Hs_A^{\otimes n}$, i.e., mapping $Y$ onto the corresponding block
in~\eqref{eq:Schur-Weyl-decomp-space}.  In particular,
\begin{align}
  \Pi^\lambda_{A^n} \mapinlambda{Y}[\lambda] \Pi^\lambda_{A^n}
  = \mapinlambda{Y}[\lambda]\ .
\end{align}
Any operator $X_{A^n}$ acting on the $n$ copies which commutes with all the permutations admits a decomposition of the
form
\begin{align}
  X_{A^n} = \sum_\lambda
  \mapinlambda{ X_\lambda \otimes \Ident_{\mathcal{P}_\lambda} }[\lambda]
  \label{eq:Schur-Weyl-perm-invar-operator-decomp}
\end{align}
for some set of operators $X_\lambda\in\mathcal{Q}_\lambda$.  In particular,
$[X_{A^n}, \Pi^\lambda_{A^n}] = 0$. We can make this more precise for i.i.d.\@{} states.  For any $X\in\GL(d_A)$, we
have that
\begin{gather}
  [\Pi_{A^n}^{\lambda}, X^{\otimes n}] = 0\\
  X^{\otimes n} =
  \sum_\lambda
  \mapinlambda{ q_\lambda(X)\otimes \Ident_{\mathcal{P}_\lambda} }[\lambda] \ .
  \label{eq:Schur-Weyl-iid-operator-decomp}
\end{gather}
For a given $\lambda\in\Young{n}{d}$, it is often useful to consider the
corresponding normalized probability distribution $\lambda/n = (\lambda_i/n)_i$.
The entropy of this distribution is given by
\begin{align}
  \bar{H}(\lambda) =
  H(\lambda/n) = -\sum_i \frac{\lambda_i}{n}\ln\frac{\lambda_i}{n}\ ,
\end{align}
where $\lambda_i$ is the number of boxes in the $i$-th row of the diagram.

If we have $n$ copies of a bipartite system $\Hs_A\otimes\Hs_B$, then we may
Schur-Weyl decompose $\Hs_A^{\otimes n}$, $\Hs_B^{\otimes n}$ and
$(\Hs_A\otimes\Hs_B)^{\otimes n}$ under the respective actions of
$\GL(d_A)\times\SN(n)$, $\GL(d_B)\times\SN(n)$ and $\GL(d_Ad_B)\times\SN(n)$.  A
useful property we will need here is that the projectors onto the respective
Schur-Weyl blocks commute between these decompositions.

\begin{lemma}
  \label{lemma:SchurWeylTensProdDecompCommute}
  Consider two spaces $\Hs_A,\Hs_B$ and let $\Pi^\lambda_{A^nB^n}$ and
  $\Pi^{\lambda'}_{A^n}$ be the projectors onto Schur-Weyl blocks of
  $\Hs_{AB}^{\otimes n}$ and $\Hs_{A}^{\otimes n}$, respectively, with
  $\lambda\in\Young{d_A d_B}{n}$ and $\lambda'\in\Young{d_A}{n}$. Then, we have
  \begin{align}
    [ \Pi^\lambda_{A^nB^n} , \Pi^{\lambda'}_{A^n} \otimes \Ident_{B^n} ] = 0\ .
  \end{align}
\end{lemma}

\begin{proof}[**lemma:SchurWeylTensProdDecompCommute]
  $\Pi^{\lambda'}_{A^n}\otimes\Ident_{B^n}$ is invariant under the action of
  $S_n$ permuting the copies of $A\otimes B$, and so it admits a decomposition
  of the form~\eqref{eq:Schur-Weyl-perm-invar-operator-decomp} and commutes with
  $\Pi_{A^nB^n}^\lambda$.
\end{proof}

The following is another lemma about how much overlap Schur-Weyl blocks have on
a bipartite system versus on one of the two systems. This lemma forms the basis
of our universal typical subspace.

\begin{lemma}\label{lemma:Schur-Weyl-bipartite-trAn-PiLambdaPiLambdaprime}
  Consider $n\in\mathbb{N}$ copies of a bipartite system $\Hs_A\otimes \Hs_B$. Then, for any
  $\lambda\in\Young{d_A d_B}{n}$ and $\lambda'\in\Young{d_B}{n}$, we have
  \begin{align}
    \Pi^{\lambda'}_{B^n}\,\tr_{A^n}`\big[\Pi^\lambda_{A^nB^n}]\,\Pi^{\lambda'}_{B^n}
    \leqslant \poly(n)\, \ee^{n(\bar{H}(\lambda)-\bar{H}(\lambda'))} \,
    \Pi^{\lambda'}_{B^n}
  \end{align}
  noting that $[\Ident_{A^n}\otimes \Pi^{\lambda'}_{B^n}, \Pi^\lambda_{A^nB^n}]=0$.
\end{lemma}

The proof is provided in \cref{app:missing}.

\subsection{Estimating entropy}

Measuring the Young diagram $\lambda$\,---\,that is, performing the projective
measurement with operators $\{ \Pi_{A^n}^\lambda \}_\lambda$\,---\,yields a good
estimation of the spectrum of a state $\rho_A$ when given
$\rho_A^{\otimes n}$~\cite{Haah2017IEEETIT_sampleoptimal}. An estimate for the
entropy of $\rho$ is thus obtained by calculating the entropy $H(\lambda/n)$
corresponding to the probability distribution $\lambda/n$.

\begin{proposition}[Spectrum and entropy estimation~\protect\cite{Bennett2014_reverse,Haah2017IEEETIT_sampleoptimal,PhDHarrow2005}]%
  \label{prop:entropy-measurement-POVM-n-systems}
  Consider $n\in\mathbb{N}$ copies of a system $\Hs_A$. Then, the family of projectors
  $`{ \Pi^\lambda_{A^n} }_\lambda$ given by Schur-Weyl duality forms a POVM
  obeying the following property: For any $\delta>0$, there exists an $\eta>0$
  such that for any state $\rho_A$, we have
  \begin{align}
    \tr`*[ `*(\sum_{\lambda:~\bar{H}(\lambda) \in [\Hfn(\rho)\pm\delta] }
    \Pi_{A^n}^\lambda ) \rho_A^{\otimes n} ]
    \geqslant 1 - \poly(n)\exp`*(- n\eta)\ .
    \label{eq:measure-lambda-povm-entropy-estimation-success-probability}
  \end{align}
\end{proposition}

The proof is provided in \cref{app:missing}.

\subsection{Estimating energy}

\begin{proposition}\label{prop:energy-measurement-POVM-n-systems}
  Consider any observable $H_A$ on $\Hs_A$ and write $\Gamma_A =
  \ee^{-H_A}$. Then, the set of projectors $`*{R_{A^n}^k}$ onto the eigenspaces
  of $\Gamma_A^{\otimes n}$ forms a POVM satisfying the following properties:
  \begin{enumerate}[label=(\roman*)]
  \item There are at most $\poly(n)$ POVM elements, with the label $k$ running
    over a set $k\in \mathcal{K}_{n}(H_A)\subset\mathbb{R}$;
  \item We have $[R_{A^n}^k, \Gamma_A^{\otimes n}] = 0$ and
    $\ee^{-nk}\, R_{A^n}^k = R_{A^n}^k\,\Gamma_A^{\otimes n}$;
  \item For any $\delta>0$ and for any state $\rho_A$,
    \begin{align}
      \tr`*[ R_{A^n}^{\approx_\delta \tr[\rho_A H_A]}
      \rho_A^{\otimes n} ] \geqslant 1 -
      2 \ee^{-n\eta}
      \label{eq:energy-measurement-POVM-n-systems-exp-success-on-iid-states}
    \end{align}
    with $\eta=\delta^2/(2\norm{H_A}_\infty^2)$ and
    where we defined for any $h\in\mathbb{R}$ that
    \begin{align}
      R_{A^n}^{\approx_\delta h} =
      \sum_{k\in\mathcal{K}_{n}(H_A) \;:\; \abs{k - h}\leqslant\delta}
      R_{A^n}^k\ ;
    \end{align}
  \item For any $h\in\mathbb{R}$, we have
  \begin{align}
    \ee^{-n(k+\delta)}R_{A^n}^{\approx_\delta h} \leqslant
    R_{A^n}^{\approx_\delta h} \, \Gamma_A^{\otimes n} \leqslant
    \ee^{-n(k-\delta)}R_{A^n}^{\approx_\delta h}\ .
    \end{align}
  \end{enumerate}
\end{proposition}

The proof is provided in \cref{app:missing}.

\subsection{Post-selection technique}

The post-selection technique is useful for bounding the diamond norm of a
candidate smoothed channel to a target ideal i.i.d.\@{} channel.

\begin{theorem}[Post-selection technique~\cite{Christandl2009PRL_Postselection}]
  \noproofref\label{x:post-selection-technique}
  Let $X,X'$ be quantum systems, $\mathcal{E}_{X\to X'}$ be a completely
  positive, trace-preserving map, and $\mathcal{T}_{X^n\to X^{\prime n}}$ be a
  completely positive, trace non-increasing map. Furthermore, let
  $\bar{R}\simeq X$,
  \begin{align}
    \zeta_{X^n} = \tr_{\bar{R}^n}`*[
    \int d\phi_{X\bar{R}}\,\proj\phi_{X\bar{R}}^{\otimes n} ]
    = \int d\sigma_X\,\sigma_X^{\otimes n}\ ,
  \end{align}
  where $d\phi_{X\bar{R}}$ denotes the Haar-induced measure on the pure states
  on $X\otimes\bar{R}$, and $d\sigma_X$ its induced measure on $X$ after partial
  trace, and let $\ket\zeta_{X^nR}$ be a purification of $\zeta_{X^n}$. Then, we have
  \begin{align}
    \frac12\norm{\mathcal{T} - \mathcal{E}^{\otimes n}}_{\diamond}
    \leqslant \poly(n)\,
    D`\big(\mathcal{T}(\zeta_{X^nR}), \mathcal{E}^{\otimes n}(\zeta_{X^nR}))\ .
  \end{align}
  Moreover, for all $n\in\mathbb{N}$ there exists a set $`*{ \ket{\phi_i}_{X\bar{R}} }$ of at most $\poly(n)$ states, and a probability distribution $`*{ p_i }$, providing a purification of $\zeta_{X^n}$ as
  \begin{align}
    \ket\zeta_{X^n\bar{R}^nR'}
    = \sum_i \sqrt{p_i}\, \ket{\phi_i}_{X\bar{R}}^{\otimes n}\otimes\ket{i}_{R'}
    \label{eq:de-Finetti-state-purified-with-poly-iid-states}
  \end{align}
  with a register $R'$ of size $\poly(n)$.
\end{theorem}

The first part of the theorem is \cite[Eq.~(4)]{Christandl2009PRL_Postselection} and the second part is, e.g., found as \cite[Cor.~D.6]{Berta2011_reverse}. The following proposition shows that a given channel is close to an i.i.d.\@{} channel, if it behaves as expected on all i.i.d.\@{} states with exponentially good accuracy.

\begin{proposition}
  \label{prop:iid-process-diamond-norm-with-any-W}
  For three systems $X,X',E$, let $V_{X\to X'E}$ be an isometry and
  $W_{X^n\to X'^nE^n}$ be an isometry which commutes with the permutations of
  the $n$ systems. Furthermore, assume that there exists $\eta>0$ independent of
  $n$ such that for all pure states $\proj\sigma_{XR_X}$ with a reference system
  $R_X\simeq X$, we have
  \begin{align}
    \Re`*{ \bra{\sigma}_{XR_X}^{\otimes n}
    (V^\dagger_{X\leftarrow X'E})^{\otimes n} \, W_{X^n\to X'^nE^n}
    \, \ket{\sigma}_{XR_X}^{\otimes n} }
    \geqslant 1 - \poly(n) \exp(-n\eta)\ .
  \end{align}
  For
  $\mathcal{E}_{X\to X'}`(\cdot) = \tr_E`\big[V_{X\to X'E}\,(\cdot)\,V^\dagger]$
  and
  $\mathcal{T}_{X^n\to X'^n}(\cdot)
  =\tr_{E^n}`\big[W_{X^n\to{}X'^nE^n}\,(\cdot)\,W^\dagger]$ we then have
  \begin{align}
    \frac12\norm[\big]{ \mathcal{T}_{X^n\to X'^n} - \mathcal{E}_{X\to X'}^{\otimes n} }_\diamond
    \leqslant \poly(n) \exp(-n\eta/2)\ .
  \end{align}
\end{proposition}

The proof is provided in \cref{app:missing}.

\section{Resource theory of thermodynamics} \label{sec:resource-theory-thermodynamics}

\subsection{Gibbs-preserving maps}

We consider the framework of Ref.~\cite{Faist2018PRX_workcost}, where for each
system $S$ considered a positive semi-definite operator $\Gamma_S\geqslant 0$ is
associated. A trace non-increasing, completely positive map $\Phi_{A\to B}$ is
allowed for free if it satisfies $\Phi_{A\to B}(\Gamma_A) \leqslant
\Gamma_B$. In the case of a system $S$ with Hamiltonian $H_S$, and in the
presence of a single heat bath at inverse temperature $\beta$, the relevant
thermodynamic framework is given by setting $\Gamma_S = \ee^{-\beta H_S}$.  In
the remainder of this paper, when using the present framework, it is convenient
to work with the $\Gamma$ operators on an abstract level. The results then also
apply to situations where several different thermodynamic baths are considered,
or in more general settings where a specific operator needs to be conserved by
the spontaneous evolution of the system~\cite{Faist2018PRX_workcost}.

The resources required to enable non-free operations are counted using an
explicit system that provides these resources, such as an \emph{information
  battery}.  An information battery is a large register $W$ whose associated
operator $\Gamma_W$ is simply $\Gamma_W=\Ident_W$ (i.e., $H_W = 0$).  The
information battery is required to be in a state of the special form
$\tau_W^m = P_W^m/\tr[P_W^m]$ where $P_W^m$ is a projector of rank
$\ee^{m}$. That is, $\tau_W^m$ has uniform eigenvalues over a given rank
$\ee^{m}$.
We denote the \emph{charge} or \emph{resource value} of a battery state
$\tau_W^m$ by $w(\tau_W^m) = \ln(d) - m$, where $d$ is the dimension of the
information battery.  The value $w(\tau)$ measures the amount of purity present
in the state $\tau$, which is the basic resource required to implement maps that
are not already Gibbs-preserving maps.
We choose to measure $w(\tau)$ in units of number of pure nats, equal to
$\ln(2)$ times a number of pure qubits. %
A Gibbs-preserving map that acts jointly on a system and an information battery,
and which maps the input battery state $\tau$ to the output battery state
$\tau'$, is deemed to \emph{consume an amount of work} $w = w(\tau) - w(\tau')$.

The resources can be counted in terms of thermodynamic work in units of energy
if we are given a heat bath at inverse temperature $T$.  Recall that a pure
qubit can be converted to $kT\ln(2)$ work using a Szil\'ard engine, where $k$ is
Boltzmann's constant~\cite{Szilard1929ZeitschriftFuerPhysik}.  By counting
purity in nats instead of qubits, we get rid of the $\ln(2)$ factor: A number
$\lambda$ of pure nats can be converted into $\lambda\,kT$ thermodynamic work
using a Szil\'ard-type engine. We count work exclusively in equivalent of pure nats, for simplicity, as opposed
to units of energy. The two are directly related by a factor $\beta^{-1}=kT$.  Furthermore, this eliminates the factor $\beta$ from
otherwise essentially information-theoretic expressions, and our theorems thus
directly apply to cases where $\Gamma_X,\Gamma_{X'}$ are any abstract positive
semi-definite operators which are not necessarily defined via a Hamiltonian.

Let $\Phi_{XW\to X'W}$ be a Gibbs-preserving map acting on an information
battery $W$, and let $\tau_W^{m}$, $\tau_W^{m'}$ be two information battery
states.  An implementation running the operation $\Phi_{XW\to X'W}$ with the
given input and output battery states is tasked to (a) make available the input
battery state, (b) apply the operation $\Phi_{XW\to X'W}$, and (c) check that
the output battery state is appropriate (e.g., for possible future re-use).  For
the verification in Point (c) it is sufficient to measure the two-outcome POVM
$`{ P_W^{m'}, \Ident-P_W^{m'} }$; as long as the first outcome is observed, it
is always possible to bring the state to $\tau_W^{m'}$ by applying a completely
thermalizing operation on the support of $P_W^{m'}$ (here, this is a completely
randomizing or completely symmetrizing operation).  In the constructions
presented in the present paper, we allow this verification measurement to fail
with a small fixed probability $\epsilon>0$.

A convenient mathematical object to characterize what the operation does on the
system is the following.  The \emph{effective work process}
$\mathcal{T}_{X\to X'F}$ associated with $\Phi_{XW\to X'W}$ and
$(\tau_W^m,\tau_W^{m'})$ is the trace non-increasing map defined as
\begin{align}
  \mathcal{T}_{X\to X'}(\cdot)
  &= \tr_W\left[ P_W^{m'}\; \Phi_{XW\to X'W}`\big( (\cdot) \otimes \tau_W^m ) \right]\ .
    \label{eq:effective-work-process}
\end{align}

The question of implementing a process $\mathcal{E}$ becomes the issue of
finding a Gibbs-preserving map along with battery states such that the
associated effective work process is close to $\mathcal{E}$.  Specifically, if
$\norm{ \mathcal{T}_{X\to X'} - \mathcal{E}_{X\to X'}}_\diamond \leqslant
\epsilon$, then we can assert that the failure probability in Point (c) above is
bounded by $\epsilon$ for all possible inputs on $X$; the operation therefore
implements $\mathcal{E}_{X\to X'}$ accurately with high success probability.

A useful characterization of which processes can be implemented using an
information battery is given by the following proposition.

\begin{proposition}[{{\protect\cite[Proposition~I]{Faist2018PRX_workcost}}}]
  \noproofref
  \label{prop:fw-T-battery-characterization}
  Let $\Gamma_X,\Gamma_{X'}\geqslant 0$, $\mathcal{T}_{X\to X'}$ be a
  completely positive, trace non-increasing map, and $w\in\mathbb{R}$. Then, the following are
  equivalent:
  \begin{enumerate}[label=(\roman*)]
  \item We have
    $ \mathcal{T}_{X\to X'}(\Gamma_X) \leqslant \ee^{w}\, \Gamma_{X'} $;
  \item \label{item:prop-fw-T-battery-characterization-exist-GPM} For all
    $\delta>0$ there exists an information battery $W$ and two battery states
    $\tau_W,\tau'_W$ such that $w(\tau_W) - w(\tau'_W) \leqslant w+\delta$, and
    there exists a Gibbs-preserving map $\Phi_{XW\to X'W}$ with
    $\mathcal{T}_{X\to X'}$ the effective work process associated with
    $\Phi_{XW\to X'W}$ and $(\tau_W, \tau'_W)$.
\end{enumerate}
\end{proposition}

Therefore, to show that one can implement $\mathcal{E}_{X\to X'}$ with
Gibbs-preserving maps while expending work $w$, it suffices to exhibit a map
$\mathcal{T}_{X\to X'}$ that is $\epsilon$-close to $\mathcal{E}_{X\to X'}$ in
diamond distance and that satisfies
$\mathcal{T}_{X\to X'}(\Gamma_X)\leqslant\ee^{w}\Gamma_{X'}$.  From the proof in
\cite{Faist2018PRX_workcost} we know in
Point~\ref{item:prop-fw-T-battery-characterization-exist-GPM} above that $W$,
$\tau_W \equiv \tau_W^{m}$ and $\tau_W' \equiv \tau_W^{m'}$ can be chosen freely as long
as $m' - m = w(\tau_W) - w(\tau'_W) \geqslant w$ and that the corresponding
Gibbs-preserving map is given by
\begin{align}
  \Phi_{XW\to X'W}(\cdot)
  &= \mathcal{T}_{X\to X'}`\big[ \tr_W`\big(P_W^m\, (\cdot)) ]\,\otimes\tau_W^{m'}\ .
    \label{eq:explicit-GPM-from-effective-work-process-and-taus}
\end{align}

In Ref.~\cite{Faist2018PRX_workcost}, the resource cost $w$ of
implementing a process $\mathcal{E}_{X\to X'}$ (any completely positive,
trace-preserving map) up to an accuracy $\epsilon\geqslant 0$ in terms of
proximity of the process matrix given a fixed input state $\sigma_X$, counted in
pure nats, was shown to be given by the \emph{coherent relative entropy}
\begin{align}\label{eq:coherent-relative-entropy}
  w = 
  - \DCohz[\epsilon]{*\mathcal{E}_{X\to X'}(\sigma_{XR_X})}{X}{X'}{\Gamma_X}{\Gamma_{X'}}
  = \ln \min_{\substack{\mathcal{T}(\Gamma_X)\leqslant\alpha\Gamma_{X'}\\
  P(\mathcal{T}(\sigma_{XR_X}), \mathcal{E}(\sigma_{XR_X})) \leqslant \epsilon
  }} \alpha\ ,
\end{align}
where $\sigma_{XR_X}$ is the purification of $\sigma_X$ on a system
$R_X\simeq X$ given by $\ket\sigma_{XR} = \sigma_X^{1/2}\,\ket\Phi_{X:R_X}$, and
where the optimization ranges over completely positive, trace non-increasing maps
$\mathcal{T}_{X\to X'}$. The coherent relative entropy enjoys a collection of properties in relation to
the conditional min- and max-entropy, and to the min- and max-relative entropy.
It satisfies the following asymptotic equipartition property: For a completely
positive, trace-preserving map $\mathcal{E}_{X\to X'}$ and quantum state
$\sigma_X$ we have for $0<\epsilon<1$ that
\begin{align}
  \lim_{n\to\infty} \frac1n
  \DCohz[\epsilon]`\big{*\mathcal{E}_{X\to X'}^{\otimes n}(\sigma_{XR}^{\otimes n})}{X^n}{X'^n}{\Gamma_X^{\otimes n}}{\Gamma_{X'}^{\otimes n}}
  = \DD{\sigma_X}{\Gamma_X} - \DD{\mathcal{E}(\sigma_X)}{\Gamma_{X'}}\ .
  \label{eq:coh-rel-entr-AEP}
\end{align}

\subsection{Thermal operations}

The framework of Gibbs-sub-preserving maps is technically convenient, but it is unclear whether any Gibbs-sub-preserving operation can be implemented at no work cost using other frameworks. This includes for example thermal operations that might be considered more operational

Here, we consider the alternative framework of \emph{thermal operations}~\cite{Brandao2013_resource,Horodecki2013_ThermoMaj,Brandao2015PNAS_secondlaws}. Each system $S$ of interest has an associated Hamiltonian $H_S$ and is not
interacting with the other systems. For a given fixed inverse temperature
$\beta$, we allow the following operations to be carried out for free:
\begin{enumerate}[label=(\roman*)]
\item Apply any unitary operation that commutes with the total Hamiltonian;
\item Bring in any ancillary system in its Gibbs state at inverse temperature
  $\beta$; and
\item Discard any system.
\end{enumerate}
The most general transformation a system $S$ can undergo under this set of rules
is a \emph{thermal operation}.  A thermal operations is any process that can be
implemented using an additional system $B$ with any Hamiltonian $H_B$ and with
any unitary $U_{SB}$ satisfying $[U_{SB}, H_S + H_B] = 0$, resulting in the
completely positive, trace-preserving map
\begin{align}
  \Phi_S(\cdot) = \tr_B\bigl[ U_{SB} \, \bigl( (\cdot)\otimes\gamma_B \bigr) \,
  U_{SB}^\dagger \bigr]\ ,
\end{align}
where $\gamma_B = e^{-\beta H_B} / \tr[e^{-\beta H_B}]$ is the Gibbs state of
the bath system $B$.  Observe that any concatenation of thermal operations is
again a thermal operation.

Clearly, any thermal operation $\Phi_S$ leaves the thermal
state $\gamma_S = \ee^{-\beta H_S}/\tr[\ee^{-\beta H_S}]$ on $S$
invariant. Hence, any lower bound on the work cost of an implementation derived
in the framework of Gibbs-preserving maps also applies to thermal
operations. We use the same definitions of work and the effective work process for thermal
operations as we defined for Gibbs-preserving maps earlier: an information
battery is used to account for work, and the effective work process associated
with a thermal operation $\Phi_{XW\to XW}$ with respect to battery states
$(\tau_W^m, \tau_W^{m'})$ is also defined by~\eqref{eq:effective-work-process}.

When considering only states that commute with the Hamiltonian, a powerful tool
to characterize possible state transformations is the notion of
thermomajorization~\cite{Horodecki2013_ThermoMaj}.  In the fully quantum regime,
there is in contrast no known simple mathematical characterization of the work
required to implement a quantum process with thermal operations.  In fact,
because thermal operations are time-covariant, it is impossible to
implement processes that are not time-covariant, even if the latter might admit
an implementation with a Gibbs-preserving map~\cite{Faist2015NJP_Gibbs}.

We will later use a primitive that transforms a thermal state into a pure energy
eigenstate. The next statement follows directly from \cite[Eq.~(8) and Suppl.\@
Note~4]{Horodecki2013_ThermoMaj}.
\begin{proposition}
  \noproofref
  \label{thm:thermal-to-zero-thermomaj-workcost}
  Let $\gamma_X = \ee^{-\beta H_X}/\tr[\ee^{-\beta H_X}]$ be the thermal state
  on a system $X$ with Hamiltonian $H_X$, and let $\ket{E}_X$ be a pure energy
  eigenstate of $H_X$.  There exists a thermal operation $\Phi_{XW}$ on an
  information battery with battery states $(\tau_W, \tau_W')$ such that
  $\Phi_{XW}`\big(\gamma_X\otimes\tau_W) = \proj{E}_X\otimes\tau'_W$ and such
  that $w(\tau_W)-w(\tau_W')$ can be chosen arbitrarily close to
  $\beta E + \ln\tr[\ee^{-\beta H_X}]$.
\end{proposition}

\section{Thermodynamic capacity}
\label{sec:thermodyn-capacity}
\label{sec:main-result}
\label{sec:statement-problem-conditions-univ-impl}

\subsection{Definition}\label{sec:definition}
Let $X,X'$ be quantum systems, $\mathcal{E}_{X\to X'}$ be a quantum process, and
${\epsilon>0}$. We seek a free thermodynamic operation (either a thermal
operation or a Gibbs preserving map) $\Phi_{X^nW\to X'^nW}$ that acts on
$X^{\otimes n}$ and a battery $W$, with output on $X'^{\otimes n}$ and $W$, as
well as information battery states $\tau_W^{\mathrm{(i)}}$ and
$\tau_W^{\mathrm{(f)}}$, such that:
\begin{enumerate}[label=(\roman*)]
\item The effective work process $\mathcal{T}_{X^n\to X'^n}$ of
  $\Phi_{X^nW\to X'^nW}$ with respect to
  $\left(\tau_W^{(\mathrm{i})},\tau_W^{(\mathrm{f})}\right)$ is $\epsilon$-close in diamond
  distance to $\mathcal{E}_{X\to X'}^{\otimes n}$;
\item We seek to minimize the work consumption per copy $w$ given by
  \begin{align}
    w = \frac1n\left[ w\left(\tau_W^{(\mathrm{i})}\right) - w\left(\tau_W^{(\mathrm{f})}\right)\right]\ .
  \end{align}
\end{enumerate}

Our main result is a collection of three independent constructions of such
implementations in different regimes, using either Gibbs-preserving maps or
thermal operations. In each case, the amount of work consumed per copy is given
by a quantity which we call the \emph{thermodynamic capacity} of the process,
and which turns out to be the minimal work cost an implementation satisfying the
above conditions can achieve. The thermodynamic capacity of a completely
positive, trace-preserving map $\mathcal{E}_{X\to X'}$ relative to operators
$\Gamma_X, \Gamma_{X'} > 0$ is defined as
\begin{align}
  \label{eq:thermodynamic-capacity-with-relative-entropy}
  T(\mathcal{E})
  = \sup_{\sigma_X} \, `\Big{
    \DD{\mathcal{E}_{X\to X'}(\sigma_X)}{\Gamma_{X'}}
    - \DD{\sigma_X}{\Gamma_X}
  }\ .
\end{align}
In a fully thermodynamic context where $\Gamma_X = \ee^{-\beta H_X}$ and
$\Gamma_{X'} = \ee^{-\beta H'_{X'}}$, one can choose to express the
thermodynamic capacity in units of energy rather than in nats, in which case a
pre-factor $\beta^{-1}$ may be included in the definition above such that the
thermodynamic capacity is a difference of free energies
\begin{align}
  T(\mathcal{E}) = \sup_{\sigma} `\Big{ F_{H'}(\mathcal{E}(\sigma)) - F_H(\sigma) }
  \quad\text{with}\quad
  F_H(\rho) = \beta^{-1} \DD{\rho}{\ee^{-\beta H}}\ .
\end{align}

\paragraph{Construction for trivial Hamiltonians}
First, in Section \ref{sec:approach-using-quasi-convexity} we consider the
special case where $\Gamma_X=\Ident_X$ and $\Gamma_{X'}=\Ident_{X'}$
corresponding to trivial Hamiltonians and show that simple considerations based
on properties of known entropy measures guarantee the existence of a universal
implementation of $\mathcal{E}^{\otimes n}$ with either thermal operations or
Gibbs-preserving maps.

\paragraph{Construction using Gibbs-preserving maps}
Second, in Section \ref{sec:proof-Weps-iid} we consider the case of general $\Gamma_X,\Gamma_{X'}$ and we construct
a universal implementation of $\mathcal{E}_{X\to X'}^{\otimes n}$ with
Gibbs-preserving maps, based on new typicality considerations.

\paragraph{Construction using thermal operations}
Third, for arbitrary Hamiltonians we construct in Section \ref{sec:optimal-universal-protocol-TO-timecovariant} a universal implementation of $\mathcal{E}_{X\to X'}^{\otimes n}$ with thermal operations, assuming that $\mathcal{E}$ is time-covariant, i.e., that it commutes with the time evolution
operation.

\subsection{Properties}

The thermodynamic capacity is a convex optimization program. Namely, the objective function of the optimization in~\eqref{eq:thermodynamic-capacity-with-relative-entropy} can be written as
\begin{align}
  \hspace*{3em}
  &\hspace*{-3em}
  \DD{\mathcal{E}_{X\to X'}(\sigma_X)}{\Gamma_{X'}}
  - \DD{\sigma_X}{\Gamma_X}
    \nonumber\\
  &= -\Hfn(\mathcal{E}_{X\to X'}(\sigma_X)) + \Hfn(\sigma_X)
    - \tr`*[ \mathcal{E}_{X\to X'}(\sigma_X)\,\ln\Gamma_{X'} ]
    + \tr`*[ \sigma_X\,\ln\Gamma_{X} ]
    \nonumber\\
  &= \HH[\rho]{E}[X']
    - \tr`*[ \mathcal{E}_{X\to X'}(\sigma_X)\,\ln\Gamma_{X'} ]
    + \tr`*[ \sigma_X\,\ln\Gamma_{X} ]
    \ ,
    \label{eq:mkuoyiftukghlj}
\end{align}
where we defined the state $\rho_{EX'} = V_{X\to X'E} \sigma_X V^\dagger$ using
a Stinespring dilation $V_{X\to X'E}$ of $\mathcal{E}_{X\to X'}$ into an
environment system $E$, satisfying
$\mathcal{E}_{X\to X'}(\cdot) = \tr_E`*[ V\,(\cdot)\,V^\dagger]$. The
conditional entropy is concave in the quantum state as
$\HH[\rho]{E}[X'] = { -\DD{\rho_{EX'}}{\Ident_E\otimes\rho_{X'}} }$ and the
quantum relative entropy is jointly convex. The other terms
in~\eqref{eq:mkuoyiftukghlj} are linear. Hence, the
optimization~\eqref{eq:thermodynamic-capacity-with-relative-entropy} is a convex
optimization that can be carried out efficiently for small system
sizes~\cite{BookBoyd2004ConvexOptimization}. Indeed, we have successfully
computed the thermodynamic capacity of simple example quantum channels acting on
few qubits with Python code, using the QuTip
framework~\cite{QuTipPy_4_1_0,Johansson2013CPC_QuTip2} and the CVXOPT
optimization software~\cite{CVXOPTPy_1_1_9} (see also \cite{Berta21} for a
direct algorithm).

The thermodynamic capacity is additive~\cite{Navascues2015PRL_nonthermal}. As a
consequence of this property, it is not necessary to include a stabilization
over a reference system in the definition of the thermodynamic capacity. That
is, had we optimized over bipartite states $\sigma_{XR}$ with a reference system
$R$ for any $\Gamma_R$, on which the process acts as the identity process, we
would be effectively computing $T(\mathcal{E}\otimes\IdentProc[R]{})$. However,
additivity implies that $T(\mathcal{E}\otimes\IdentProc[R]{}) = T(\mathcal{E})$.

\begin{proposition}[Additivity of thermodynamic
  capacity~\cite{Navascues2015PRL_nonthermal}]%
  \label{prop:thermo-capacity-additive}
  For $\Gamma_{X},\Gamma_{X'},\Gamma_Z, \Gamma_{Z'} > 0$ and quantum
  channels $\mathcal{E}_{X\to X'}$, $\mathcal{F}_{Z\to Z'}$ we have
  \begin{align}
    T(\mathcal{E}\otimes\mathcal{F}) = T(\mathcal{E})+T(\mathcal{F})\ .
  \end{align}
\end{proposition}

For completeness we provide an independent proof of additivity, to ensure
validity in the general setting of abstract $\Gamma$ operators.

\begin{proof}[**prop:thermo-capacity-additive]
  Let $\sigma_X,\tau_Z$ be states achieving the thermodynamic capacity of
  $T(\mathcal{E})$ and $T(\mathcal{F})$, respectively. Then,
  $\sigma_X\otimes\tau_Z$ is a candidate for $T(\mathcal{E}\otimes\mathcal{F})$,
  yielding
  \begin{align}
    T(\mathcal{E}\otimes\mathcal{F})
    &\geqslant
    \DD`\big{\mathcal{E}(\sigma)\otimes\mathcal{F}(\tau)}{\Gamma_{X'}\otimes\Gamma_{Z'}}
    - \DD`\big{\sigma\otimes\tau}{\Gamma_{X}\otimes\Gamma_{Z}}
      \nonumber\\
    &= \DD{\mathcal{E}(\sigma)}{\Gamma_{X'}} - \DD{\sigma}{\Gamma_X}
    + \DD{\mathcal{F}(\tau)}{\Gamma_{Z'}} - \DD{\tau}{\Gamma_Z}
    \nonumber\\
    &= T(\mathcal{E}) + T(\mathcal{F})\ .
  \end{align}
  Now, let $\zeta_{XZ}$ achieve the optimum for
  $T(\mathcal{E}\otimes\mathcal{F})$.  Let $V_{X\to E_1X'}$, $W_{Z\to E_2Z'}$ be
  Stinespring isometries of $\mathcal{E}$ and $\mathcal{F}$ respectively, such
  that $\mathcal{E}(\cdot) = \tr_{E_1}`*[V\,(\cdot)\,V^\dagger]$ and
  $\mathcal{F}(\cdot) = \tr_{E_2}`*[W\,(\cdot)\,W^\dagger]$.  Let
  $\rho_{E_1E_2X'Z'} = (V\otimes W)\, \zeta \,(V\otimes W)^\dagger$. Then, we have
  \begin{align}
    T(\mathcal{E}\otimes\mathcal{F})
    &= \DD`\big{(\mathcal{E}\otimes\mathcal{F})(\zeta)}{\Gamma_{X'}\otimes\Gamma_{Z'}}
      - \DD`\big{\zeta_{XZ}}{\Gamma_X\otimes\Gamma_Z}
      \nonumber\\
    &= \HH[\rho]{E_1E_2}[X'Z']
      - \tr`*[\rho_{X'Z'} \ln`*(\Gamma_{X'}\otimes\Gamma_{Z'})]
      + \tr`*[\zeta_{XZ} \ln`*(\Gamma_{X}\otimes\Gamma_{Z})]\ ,
      \nonumber\\
    &= \HH[\rho]{E_1E_2}[X'Z']
      - \tr`*[\rho_{X'} \ln`*(\Gamma_{X'})] - \tr`*[\rho_{Z'}\ln`*(\Gamma_{Z'})]
      \nonumber\\
    &\quad  + \tr`*[\zeta_{X} \ln`*(\Gamma_{X})] + \tr`*[\zeta_Z\ln`*(\Gamma_{Z})]
      \label{eq:iuewfwlkfjgifs}
  \end{align}
  since $\ln(A\otimes B) = \ln(A)\otimes\Ident + \Ident\otimes\ln(B)$.  Invoking
  the chain rule of the von Neumann entropy, and then strong sub-additivity of
  the entropy, we see that
  $\HH[\rho]{E_1E_2}[X'Z'] = \HH[\rho]{E_1}[X'Z'] + \HH[\rho]{E_2}[E_1X'Z']
  \leqslant \HH[\rho]{E_1}[X'] + \HH[\rho]{E_2}[Z']$. Hence, we have
  \begin{align}
    \text{\eqref{eq:iuewfwlkfjgifs}}
    &\leqslant \HH[\rho]{E_1}[X'] - \tr`*[\rho_{X'}\ln(\Gamma_{X'})]
      + \tr`*[\zeta_X\ln`*(\Gamma_X)]
      \nonumber\\
    &\quad + \HH[\rho]{E_2}[Z'] - \tr`*[\rho_{Z'}\ln(\Gamma_{Z'})]
      + \tr`*[\zeta_Z\ln`*(\Gamma_Z)]
      \nonumber\\
    &\leqslant T(\mathcal{E}) + T(\mathcal{F})\ ,
  \end{align}
  where the last inequality holds because the reduced states $\zeta_X, \zeta_Z$
  are optimization candidates for $T(\mathcal{E})$ and $T(\mathcal{F})$,
  respectively.
\end{proof}

A special case worth mentioning is when $\Gamma_X=\Ident_X$, $\Gamma_{X'} = \Ident_{X'}$, which corresponds to the
situation where the Hamiltonians of $X$ and $X'$ are trivial. For any quantum channel $\mathcal{E}_{X\to X'}$, let $V_{X\to X'E}$ be a Stinespring dilation isometry with
$\mathcal{E}_{X\to X'}`*(\cdot) = \tr_E`*[V\,(\cdot)\,V^\dagger]$. Then, we have
\begin{align}
  T(\mathcal{E})
  = \sup_\sigma`*{ \Hfn(\sigma_X) - \Hfn(\mathcal{E}(\sigma_X)) }
  = \sup_\sigma \HH[V\sigma V^\dagger]{E}[X'] \ .
\end{align}
That is, the thermodynamic capacity characterizes by how much the channel is
capable of reducing the entropy of its input, or equivalently, how much entropy
the channel is capable of dumping into the environment when conditioned on the
output. We note that the quantity $-T(\mathcal{E})$ has previously been studied
in the information theory literature as the entropy gain of quantum channels
\cite{Alicki2004arXiv_isotropic,Devetak2006CMP_multiplicativity,Holevo2011ISIT_entropygain,Holevo2010DM_infinitedim,Holevo2011TMP_CJ,BookHolevo2012_QuSystemsChannelsInformation,Buscemi2016PRA_reversibility,Gour2021PRR_entropychannel}. Our
work can be seen as giving a precise operational interpretation to this quantity.

\subsection{Optimality}

Here, we show that any universal implementation that obeys our stated conditions
in Section \ref{sec:definition} must necessarily consume an amount of work that
is lower bounded by the thermodynamic capacity.  That is, any universal
implementation that consumes an amount of work equal to the thermodynamic
capacity is optimal. This lower bound is simple to prove, because a universal
implementation of a process must necessarily be a good implementation for any
individual i.i.d.\@{} input state, a situation where the optimal work cost is
known~\cite{Faist2018PRX_workcost}. Furthermore, any scheme that satisfies the
requirements of Section \ref{sec:statement-problem-conditions-univ-impl} at work
cost $w$ per copy counted with standard battery states of
Ref.~\cite{Faist2018PRX_workcost}, has an effective process
$\mathcal{T}_{X^n\to X'^n}$ on the systems that obeys
$\mathcal{T}(\Gamma_X^{\otimes n}) \leqslant \ee^{nw}\Gamma_{X'}^{\otimes n}$.
This is because any thermal operation is in particular a Gibbs-preserving map,
and the work cost is characterized by \cref{prop:fw-T-battery-characterization}.
The following shows that for any such implementation, the work consumed $w$ per
copy cannot be less than the thermodynamic capacity of the process.

\begin{proposition}%
  \label{lemma:main-thm-easy}
  Let $\epsilon>0$, $\Gamma_X,\Gamma_{X'} > 0$, $\mathcal{E}_{X\to X'}$ a
  completely positive, trace-preserving map, and $\mathcal{T}_{X^n\to X'^n}$ a
  completely positive, trace non-increasing map such that we have
  $\norm{\mathcal{T} - \mathcal{E}^{\otimes
      n}}_\diamond/2\leqslant\epsilon$. For $w\in\mathbb{R}$ such that
  $\mathcal{T}_{X^n\to X'^n}(\Gamma_X^{\otimes n})
  \leqslant\ee^{nw}\,\Gamma_{X'}^{\otimes n}$, we have in the limit $n\to\infty$
  that $w\geqslant T(\mathcal{E})$.
\end{proposition}

\begin{proof}[**lemma:main-thm-easy]
  Let $\mathcal{T}$ with
  $\frac12\norm{\mathcal{E} - \mathcal{T}}_\diamond\leqslant\epsilon$,
  $\sigma_X$ be a quantum state, and
  $\ket\sigma_{XR_X} = \sigma_X^{1/2}\,\ket\Phi_{X:R_X}$. Then, by definition of
  the diamond norm it must hold that
  $D`\big(\mathcal{E}(\sigma_{XR_X}), \mathcal{T}(\sigma_{XR_X})) \leqslant
  \epsilon$, which implies that
  $P`\big(\mathcal{E}(\sigma_{XR_X}), \mathcal{T}(\sigma_{XR_X})) \leqslant
  \sqrt{2\epsilon}$. We have that $\mathcal{T}$ is a valid optimization
  candidate for the definition of the coherent relative entropy and thus
  \begin{align}
    -\DCohz[\sqrt{2\epsilon}]`\big{*\mathcal{E}_{X\to X'}^{\otimes n}(\sigma_{XR_X}^{\otimes n})}{X^n}{X'^n}{\Gamma_X^{\otimes n}}{\Gamma_{X'}^{\otimes n}}
    \leqslant nw\ .
  \end{align}
  For $n\to\infty$, we can employ the asymptotic equipartition of the coherent relative
  entropy~\eqref{eq:coh-rel-entr-AEP} to see that
  \begin{align}
    \DD{\mathcal{E}(\sigma_X)}{\Gamma_{X'}} -
    \DD{\sigma_X}{\Gamma_X} \leqslant w\ .
  \end{align}
  Since this inequality holds for all $\sigma_X$, we deduce that $T(\mathcal{E})\leqslant w$.
\end{proof}

\section{Construction \#1: Trivial Hamiltonians}
\label{sec:approach-using-quasi-convexity}

\subsection{Statement and proof sketch}

Instead of constructing explicitly an implementation that satisfies the requirements of Section \ref{sec:statement-problem-conditions-univ-impl}, one might hope that the implementation could be given implicitly as the solution of a semi-definite program representing an entropy measure. This proof idea was indeed exploited in other contexts in Refs.~\cite{Berta2011_reverse,Berta2014IEEETIT_InfGainMeas}. Here, we define the one-shot entropy-like quantity
\begin{align}\label{eq:univ-process-work-cost-def-Weps}
  \WProc[\epsilon]{\mathcal{E}}{X}{X'}{\Gamma_X}{\Gamma_{X'}}
  =  \min_{\substack{\mathcal{T}(\Gamma_X)
  \leqslant \ee^{w} \Gamma_{X'}\\
  \frac12\norm*{\mathcal{T} - \mathcal{E}}_\diamond \leqslant \epsilon}}
  w \ ,
\end{align}
where $\mathcal{T}_{X\to X'}$ ranges over all trace non-increasing, completely
positive maps. The proof strategy would then be to relate this entropy measure
to the coherent relative entropy, and to exploit known properties of the latter
in the i.i.d.\@ regime to provide an upper bound to the expression
\begin{align}
  \frac{1}{n}\WProc[\epsilon]{\mathcal{E}^{\otimes n}}{X^n}{X'^n}{\Gamma_X^{\otimes n}}{\Gamma_{X'}^{\otimes n}}\ .
\end{align}
Should this upper bound behave like $T(\mathcal{E})$ to leading order, then the
$\mathcal{T}$ equal to the optimal solution
to~\eqref{eq:univ-process-work-cost-def-Weps} defines an implementation in terms
of Gibbs-preserving maps thanks to \cref{prop:fw-T-battery-characterization}. It
turns out that this proof strategy works well in the special case of trivial
Hamiltonians, but fails in the general case.

The core technical statement that underlies our Construction \#1 is summarized
in the following theorem.

\begin{theorem}%
  \label{thm:constr-trivialH}%
  Let $\mathcal{E}_{X\to X'}$ be a completely positive, trace-preserving map,
  and $\epsilon>0$. Then we have
  \begin{align}
    \lim_{n\to\infty} \frac1n
    \WProc[\epsilon]{\mathcal{E}^{\otimes n}}{X^n}{X'^n}{\Ident_{X^n}}{\Ident_{X'^n}}
    = T(\mathcal{E})\ ,
  \end{align}
  where
  $T(\mathcal{E}) = \max_{\sigma_X}\left\{\Sfn(\sigma_X)
    -\Sfn(\mathcal{E}(\sigma_X))\right\}$.
\end{theorem}

This implementation is constructed by taking the implicit optimal solution
$\mathcal{T}_{X^n\to X'^n}$ in the semi-definite
program~\eqref{eq:univ-process-work-cost-def-Weps} for
$\frac1n \WProc[\epsilon]{*\mathcal{E}_{X\to X'}^{\otimes n}}{X^n}{X'^n}{\Ident_{X^n}}{\Ident_{X'^n}}$, and using
\cref{prop:fw-T-battery-characterization} to construct an associated
Gibbs-preserving map acting on battery states
via~\eqref{eq:explicit-GPM-from-effective-work-process-and-taus}.  In summary,
for any $\delta'>0$, for $n$ large enough and choosing any $m,m'$ such that
$m-m' \leqslant nT(\mathcal{E}) + \delta'$, the full implementation map in terms
of $\mathcal{T}_{X^n\to X'^n}$ becomes
\begin{align}
  \Phi_{X^nW\to X'^nW}(\cdot)
  &= \mathcal{T}_{X^n\to X'^n}`\big( \tr_W`[ P_W^m (\cdot) ] ) \otimes \tau_W^{m'}\ .
\end{align}
We emphasise that \cref{thm:constr-trivialH} exactly covers the entropy gain of
quantum channels as studied in
\cite{Alicki2004arXiv_isotropic,Devetak2006CMP_multiplicativity,Holevo2011ISIT_entropygain,Holevo2010DM_infinitedim,Holevo2011TMP_CJ,BookHolevo2012_QuSystemsChannelsInformation,Buscemi2016PRA_reversibility,Gour2021PRR_entropychannel}.

\begin{proof}[*thm:constr-trivialH]
  By using the post-selection technique (\cref{x:post-selection-technique}) and
  recalling that the fixed-input state case is given by the coherent relative
  entropy, we find
  \begin{align}
    \WProc[\epsilon]`\big{*\mathcal{E}^{\otimes n}_{X\to X'}}{X^n}{X'^n}{\Ident_{X^n}}{\Ident_{X'^n}}
    \leqslant
    - \DCohz[\epsilon/\!\poly(n)]`\big{*\mathcal{E}^{\otimes n}_{X\to X'}(\zeta_{X^nR_X^n})}{X^n}{X'^n}{\Ident_{X^n}}{\Ident_{X'^n}}\ .
  \label{eq:kadau4pwaoau}
\end{align}
In the case of trivial Hamiltonians, the coherent relative entropy reduces to
the smooth max-entropy (cf.\@~\cite[Props.~28 and~26]{Faist2018PRX_workcost} and
also Ref.~\cite{Faist16}). More precisely, we have
\begin{align}
  \DCohz[\epsilon]`\big{\rho}{X}{X'}{\Ident_X}{\Ident_{X'}}
  \geqslant - \Hmaxf[\rho][c\epsilon^\alpha]{E}[X'] + g(\epsilon)\ ,
\end{align}
where $\ket\rho_{X'R_XE}$ is a pure state, where $c>0$, $0<\alpha<1$, $g(\epsilon)$ are universal and do not depend on the state or the dimensions of the systems, and the smooth max-entropy is defined as
\begin{align}
  \Hmaxf[\rho][\epsilon]{E}[X']
  &= \min_{P(\hat\rho, \rho)\leqslant\epsilon}
    \Hmaxf[\hat\rho]{E}[X']\ ;
  \\
  \Hmaxf[\hat\rho]{E}[X']
  &= \max_{0\leq\omega_{X'}\leq\Ident} \,
    \ln\, \norm[\big]{\hat\rho_{EX'}^{1/2}\omega_{X'}^{1/2}}^2_1\ .
\end{align}
Thus, we have
\begin{align}
  \eqref{eq:kadau4pwaoau}
  \leqslant \Hmaxf[\rho][\epsilon^\alpha/\!\poly(n)]{ E^n }[ X'^n ] + g(\epsilon)\ ,
  \label{eq:fiosdufanjo}
\end{align}
where
$\rho_{X'^nE^n} = V_{X\to X'E}^{\otimes n} \zeta_{X^n} (V^\dagger)^{\otimes n} =
\int d\sigma\, (V\sigma V^\dagger)^{\otimes n}$ and $V_{X\to X'E}$ is a
Stinespring dilation isometry of $\mathcal{E}_{X\to X'}$ as
$\mathcal{E}_{X\to X'}(\cdot) = \tr_E `*[ V_{X\to X'E}\,(\cdot)\,
V^\dagger]$. At this point we invoke two facts. First, note that the de Finetti
state can be written as a mixture of only $\poly(n)$ i.i.d.\@{} states, instead of a
continuous average (\cref{x:post-selection-technique}): There exists a set
$`{ \sigma_i }$ of at most $\poly(n)$ states and a distribution $`{ p_i }$ such
that $\zeta_{X^n} = \sum_i p_i \sigma_i^{\otimes n}$. Second, we invoke the
property that the conditional max-entropy is quasi-convex up to a penalty term,
namely, that the conditional max-entropy of $\sum_i p_i \rho_i$ is less than or
equal to the maximum over the set of max-entropies corresponding to each
$\rho_i$, plus a term proportional to the logarithm of the number of terms in
the sum~\cite[Lemma~11]{Morgan2014IEEETIT_prettystrong}. Hence, with
$\rho_i = V\,\sigma_i\,V^\dagger$, we get
\begin{align}
  \text{\eqref{eq:fiosdufanjo}}
  \leqslant
  \max_i \Hmaxf[\rho_i^{\otimes n}][\epsilon^\alpha/\!\poly(n)]{ E^n }[ X'^n ]
  + \ln(\poly(n)) + g(\epsilon) \ .
  \label{eq:fiodubashfjnk}
\end{align}
Now, we are in business because the max-entropy is evaluated on an i.i.d.\@{} state,
and we know that it asymptotically goes to the von Neumann entropy in this
regime~\cite{Tomamichel2009IEEE_AEP}. Also,
$\lim_{n\to\infty} (1/n)\big\{\ln(\poly`(n)) + g(\epsilon)\big\}=0$ and hence
\begin{align}
  \lim_{n\to\infty} \frac1n
  \WProc[\epsilon]`\big{*\mathcal{E}^{\otimes n}_{X\to X'}}{X^n}{X'^n}%
  {\Ident_{X^n}}{\Ident_{X'^n}}
  &\leqslant 
  \max_i\, \HH[\rho_i]{E}[X']
  \nonumber\\
  &= 
  \max_i `*{ \Hfn(\sigma_i) - \Hfn(\mathcal{E}(\sigma_i)) }
  \nonumber\\
  &\leqslant 
    \max_\sigma `*{ \Hfn(\sigma) - \Hfn(\mathcal{E}(\sigma)) }
    \nonumber\\
  &= T(\mathcal{E})
\end{align}
noting that $\HH{E}[X'] = \HH{EX'} - \HH{X'} = \HH{X} - \HH{X'}$.
\end{proof}

\subsection{Challenges for %
  extension to non-trivial Hamiltonians}
\label{sec:entropic-proof-approach-nontrivial-Hamiltonians}

Naturally, one might ask whether it is possible to extend this proof to the case
of non-trivial $\Gamma$ operators.  Interestingly, this is not possible, at least
not in a naive way.  The problem is that we need a quasi-convexity property of
the form
\begin{multline}
  - \DCohz[\epsilon]`\big{*\mathcal{E}_{X\to X'}(\sigma_{XR_X})}{X}{X'}{\Gamma_{X}}{\Gamma_{X'}}
  \\
  \stackrel{\text{\large?}}{\leqslant}
  \max_i 
  `*( - \DCohz[\epsilon]`\big{*\mathcal{E}_{X\to X'}(\sigma^i_{XR_X})}{X}{X'}{\Gamma_{X}}{\Gamma_{X'}} ) + \text{(penalty)}\ ,
  \label{eq:quasi-convexity-coh-rel-entr-false}
\end{multline}
where $\sigma_X = \sum p_i \sigma_X^i$ and
$\ket\sigma_{XR} = \sigma_X^{1/2}\,\ket\Phi_{X:R_X}$,
$\ket{\sigma^i}_{XR} = (\sigma^i_X)^{1/2}\,\ket\Phi_{X:R_X}$, and where the
$\text{(penalty)}$ term scales in a favourable way in $n$, say of order
$\ln(\poly(M))$ where $M$ is the number of terms in the convex decomposition as
for the max-entropy. In fact, Eq.~\eqref{eq:quasi-convexity-coh-rel-entr-false}
is false, as can be shown using an explicit counterexample on a two-level system
which we present below. As this example is based on physical reasons, the
coherent relative entropy is not even approximately quasi-convex. We note that a
priori we cannot rule out a quasi-convexity property that might have a penalty
term that depends on properties of the $\Gamma$ operators, yet such a term is
likely to scale unfavourably with $n$.

Our example is as follows. Consider a two-level system with a Hamiltonian $H$
with energy levels $\ket0,\ket1$ at corresponding energies $E_0=0$ and
$E_1>0$. The corresponding $\Gamma$ operator is $\Gamma = g_0\proj0 + g_1\proj1$
with $g_0 = 1$, $g_1 = \ee^{-\beta E_1}$.  Consider the process consisting in
erasing the input and creating the output state $\ket+$, where we define
$\ket\pm = [\ket0\pm\ket1]/\sqrt2$. That is, we consider the process
$\mathcal{E}(\cdot) = \tr[\cdot]\,\proj+$. Suppose the input state is maximally
mixed, $\sigma=\Ident/2$, such that
$\rho_{X'R_X} = \proj+_{X'}\otimes\Ident_{R_X}/2$. If $E_0=0$ and
$E_1\to\infty$, then this process requires a lot of work; intuitively, with
probability $1/2$ we start in the ground state $\ket0$ and need to prepare the
output state $\ket+$ which has high energy.

For $\epsilon=0$, we can see this
because the input state is full rank, hence $\mathcal{T}=\mathcal{E}$; then
$\mathcal{E}(\Gamma) = \tr[\Gamma]\proj+$ and the smallest $\alpha$ such that
$\mathcal{E}(\Gamma) \leqslant \alpha\Gamma$ is given by
\begin{multline}
\alpha/\tr[\Gamma] = \norm[\big]{\Gamma^{-1/2}\proj+\Gamma^{-1/2}}_\infty =
\dmatrixel+{\Gamma^{-1}} = (g_0^{-1} + g_1^{-1})/2
\\
= (1 + \ee^{\beta E_1})/2 \geqslant \ee^{\beta E_1}/2\ .
\end{multline}
Noting that $\tr[\Gamma]\geqslant 1$, we have
$\alpha \geqslant \ee^{\beta E_1}/2$, and hence the energy cost of the
transformation $\Ident/2\to\ket+$ is
\begin{align}
  \text{energy cost}
  = - \beta^{-1}\DCohz{*\mathcal{E}_{X\to X'}(\sigma_{XR_X})}{X}{X'}{\Gamma}{\Gamma}
  = \beta^{-1} \ln\alpha
  \geqslant E_1 - \beta^{-1}\ln(2)\ .
\end{align}
Clearly, this work cost can become arbitrarily large if $E_1\to\infty$.  On the
other hand, we can perform the transformation $\ket+\to\ket+$ obviously at no
work cost; similarly, $\ket-\to\ket+$ can be carried out by letting the system
time-evolve under its own Hamiltonian for exactly the time interval required to
pick up a relative phase $(-1)$ between the $\ket0$ and $\ket1$ states.  This
also costs no work because it is a unitary operation that commutes with the
Hamiltonian.  We thus have our counter-example to the quasi-convexity of the
coherent relative entropy.  The transformation $\Ident/2\to\ket+$ is very hard,
but the individual transformations $\ket\pm\to\ket+$ are trivial, noting that
$\Ident/2=(1/2)\proj+ + (1/2)\proj-$.

We show in \cref{app:smooth} how to make the above claim robust against an
accuracy tolerance $\epsilon\geq0$.

\section{Construction \#2: Gibbs-preserving maps}\label{sec:proof-Weps-iid}

\subsection{Statement and proof sketch}

Here, we present a general construction of a universal implementation of an
i.i.d.\@{} process using Gibbs-preserving maps according to the requirements
of~\cref{sec:definition}. The idea is to explicitly construct an implementation
using a novel notion of quantum typicality. We introduce notions of quantum
typicality that apply to quantum processes and universally capture regions of
the Hilbert space where the conditional entropy (respectively the relative
entropy difference) has a given value. This generalizes existing notions of typical projectors to a quantum typical operator that applies to bipartite states, is relative to a $\Gamma$ operator,
and universal.

The main result behind the construction in this section is the following
theorem.

\begin{theorem}\label{thm:construction-using-typicality-full}
  Let $\Gamma_X,\Gamma_{X'} > 0$, $\mathcal{E}_{X\to X'}$ be a completely
  positive, trace-preserving map, and $\epsilon>0$. Then, for $\delta>0$ and
  $n\in\mathbb{N}$ large enough there exists a completely positive map
  $\mathcal{T}_{X^n\to X'^n}$ such that:
  \begin{enumerate}[label=(\roman*)]
  \item $\mathcal{T}_{X^n\to X'^n}$ is trace non-increasing;
  \item
    $\norm[\big]{\mathcal{T}_{X^n\to X'^n} -
      \mathcal{E}_{X\to{}X'}^{\otimes{}n}}_\diamond \leqslant \epsilon$;
  \item $\mathcal{T}_{X^n\to X'^n}`\big(\Gamma_X^{\otimes n})
    \leqslant \ee^{n[T(\mathcal{E}) + 4\delta + n^{-1}\ln(\poly(n))]}\,\Gamma_{X'}^{\otimes n}$.
  \end{enumerate}
\end{theorem}

Note that we have $n^{-1}\ln(\poly(n)) \to 0$ as $n\to\infty$, and that we can
take $\delta\to 0$ after taking $n\to\infty$. Thanks to
\cref{prop:fw-T-battery-characterization}, the mapping
$\mathcal{T}_{X^n\to X'^n}$ defines an implementation of the i.i.d.\@{} process
$\mathcal{E}_{X\to X'}^{\otimes n}$ in terms of Gibbs-preserving maps and a
battery, whose work cost rate is given to leading order by the thermodynamic
capacity $T(\mathcal{E})$ after taking $\delta\to0$.

As for Construction~\#1, the full Gibbs-preserving map implementing the required
process is assembled in two steps, first constructing the map
$\mathcal{T}_{X^n\to X'^n}$ in \cref{thm:construction-using-typicality-full} and
then using \cref{prop:fw-T-battery-characterization} to obtain the full
Gibbs-preserving map. Let $V_{X\to X'E}$ be a Stinespring dilation isometry of
$\mathcal{E}_{X\to X'}$. For $\delta>0$, we introduce a universal conditional
and relative typical smoothing operator $M_{E^nX'^n}^{x,\delta}$ (see later
\cref{defn:universal-relative-conditional-typical-smoother} and
\cref{thm:universal-relative-conditional-typical-smoother}) with
$x = -nT(\mathcal{E})$ and relative to $\Gamma_{X'E} \equiv V\Gamma_X V^\dagger$
and $\Gamma_{X'}$. The map $\mathcal{T}_{X^n\to X'^n}$ is then constructed as
\begin{align}
  \mathcal{T}_{X^n\to X'^n}(\cdot)
  &= \tr_{E^n} \left[ M_{E^nX'^n}^{x,\delta} \,V_{X\to X'E}^{\otimes n}\, (\cdot)\,
    V_{X\leftarrow X'E}^{\dagger\,\otimes n} M_{E^nX'^n}^{x,\delta\,\dagger} \right]\ .
\end{align}
Finally, we employ \cref{prop:fw-T-battery-characterization} to construct an associated
Gibbs-preserving map acting on battery states via~\eqref{eq:explicit-GPM-from-effective-work-process-and-taus}.  For any $\delta'>0$, for $n$ large enough and choosing any $m,m'$ such that
$m-m' \leqslant nT(\mathcal{E}) + 4\delta + n^{-1}\ln\poly(n) + \delta'$, the
full implementation map in terms of $\mathcal{T}_{X^n\to X'^n}$ becomes
\begin{align}
  \Phi_{X^nW\to X'^nW}(\cdot)
  &= \mathcal{T}_{X^n\to X'^n}`\big( \tr_W`[ P_W^m (\cdot) ] ) \otimes \tau_W^{m'}\ .
\end{align}

\subsection{Construction via universal conditional and relative typicality}

The main ingredient of our proof is a notion of a universal conditional and
relative typical smoothing operator that enables us to discard events
that are very unlikely to appear in the process while accounting for how much
they contribute to the overall work cost. This operator is inspired by
similar constructions in Refs.~\cite{Bjelakovic2003arXiv_revisted,Berta2015QIC_monotonicity}. However, in additional to being ``relative'' as in \cite{Bjelakovic2003arXiv_revisted} our smoothing
operator is also simultaneously ``conditional'' and ``universal''.

\begin{definition}
  \label{defn:universal-relative-conditional-typical-smoother}
  Let $\Gamma_{AB},\Gamma_B'\geqslant 0$ and $x\in\mathbb{R}$.  A \emph{universal
    conditional and relative typical smoothing operator} $M_{A^nB^n}^{x,\delta}$
  with parameter $\delta>0$ is an operator on $A^nB^n$ that satisfies the
  following conditions:
  \begin{enumerate}[label=(\roman*)]
  \item \label{item:univ-rel-cond-typ-sm-trnonincreasing}
    $`\big(M^{x,\delta}_{A^nB^n})^\dagger\; M^{x,\delta}_{A^nB^n} \leqslant\Ident$\ ;
  \item \label{item:univ-rel-cond-typ-sm-high-state-weight} There exists $\xi>0$
    independent of $n$ with the following property: For any pure state
    $\ket{\rho}_{ABR}$ with $\rho_{AB}$ (respectively $\rho_B$) in the support
    of $\Gamma_{AB}$ (respectively $\Gamma_B'$) and such that
    $\DD{\rho_{AB}}{\Gamma_{AB}} - \DD{\rho_B}{\Gamma_B'} \geqslant x$, it
    holds that
    \begin{align}
      \Re`*{ \bra{\rho}_{ABR}^{\otimes n} \, M^{x,\delta}_{A^nB^n}\,
      \ket{\rho}_{ABR}^{\otimes n} } \geqslant 1 - \poly(n)\exp(-n\xi)\ ;
    \end{align}
  \item \label{item:univ-rel-cond-typ-sm-conditional-Gamma-weight}
    $\tr_{A^n}`\Big[M^{x,\delta}_{A^nB^n}\,\Gamma_{AB}^{\otimes n}\,
    `\big(M^{x,\delta}_{A^nB^n})^\dagger] \leqslant
    \poly`(n)\,\ee^{-n(x-4\delta)}\,\Gamma_{B}'^{\otimes n}$\ .
  \end{enumerate}
\end{definition}

Note that the smoothing operator is defined as a general operator of norm
bounded by one, as opposed to the usual definition of typical subspaces or
typical projectors. The main reason is that it is not known to us in general if
such an object can be chosen to be a projector.  By using the real part in
Point~\ref{item:univ-rel-cond-typ-sm-high-state-weight} above, we ensure that a
process that applies the operator $M_{A^nB^n}^{x,\delta}$ preserves coherences
when it is applied to a superposition of several states
$`{ \ket{\rho}_{ABR}^{\otimes n} }$. This property would not have been ensured
if instead, we had merely asserted that
$M_{A^nB^n}^{x,\delta}\ket{\rho}_{ABR}^{\otimes n}$ and
$\ket{\rho}_{ABR}^{\otimes n}$ have high absolute value overlap or are close in
fidelity. If $M_{A^nB^n}^{x,\delta}$ is a projector then the expression reduces
to $\tr(M_{A^nB^n}^{x,\delta}\rho)$ as one usually considers for projectors on
typical subspaces.

The core technical statement of Construction \#2 is to show the existence of a
universal conditional and relative smoothing operator, which is as follows.

\begin{proposition}\label{thm:universal-relative-conditional-typical-smoother}
  Let $\Gamma_{AB},\Gamma_B'\geqslant 0$, $x\in\mathbb{R}$, as well as
  $n\in\mathbb{N}$ and $\delta>0$.  There exists a universal conditional and
  relative typical smoothing operator $M_{A^nB^n}^{x,\delta}$ that is
  furthermore permutation-invariant. Moreover, if $[\Gamma_{AB}, \Ident_A\otimes\Gamma_B']=0$, then
  $M^{x,\delta}_{A^nB^n}$ can be chosen to be a projector satisfying
  $[M^{x,\delta}_{A^nB^n},\Gamma_{B}'^{\otimes n}] = 0$ and
  $[M^{x,\delta}_{A^nB^n},\Gamma_{AB}^{\otimes n}] = 0$.
\end{proposition}

In the following, %
we present the proof of \cref{thm:construction-using-typicality-full} based on
the existence of such the smoothing operator from
\cref{thm:universal-relative-conditional-typical-smoother}. The more technical
proof of \cref{thm:universal-relative-conditional-typical-smoother}
is then given in \cref{sec:constr-univ-cond-rel-typ-smoother-proof}.

\begin{proof}[*thm:construction-using-typicality-full]
  Let $V_{X\to X'E}$ be a Stinespring dilation of $\mathcal{E}_{X\to X'}$ into
  an environment system $E\simeq X\otimes X'$. For $n\in\mathbb{N}$ we need to find a suitable candidate implementation
  $\mathcal{T}_{X^n\to X'^n}$. Let
  \begin{align}
    x = -\max_{\sigma_X}\Big\{ \DD{\mathcal{E}(\sigma_X)}{\Gamma_{X'}}
    -  \DD{\sigma_X}{\Gamma_X} \Big\} = -T(\mathcal{E})\ .
  \end{align}
  For any $\delta>0$ let $M^{x,\delta}_{E^nX'^n}$ be the operator constructed by
  \cref{thm:universal-relative-conditional-typical-smoother}, with the system
  $E$ playing the role of the system $A$, with
  $V_{X\to X'E}\,\Gamma_X\,V_{X\leftarrow X'E}^\dagger$ as $\Gamma_{AB}$ and
  with $\Gamma_{X'}$ as $\Gamma'_{B}$.  Now, define
  \begin{align}
    \mathcal{T}_{X^n\to X'^n}(\cdot) = \tr_{E^n}`*[
    M^{x,\delta}_{E^nX'^n} V_{X\to X'E}^{\otimes n} \; `\big(\cdot) \;
    `\big(V^\dagger_{X\leftarrow X'E})^{\otimes n} `\big(M^{x,\delta}_{E^nX'^n})^\dagger]
  \end{align}
  noting that $\mathcal{T}_{X^n\to X'^n}$ is trace non-increasing by
  construction thanks to
  Property~\ref{item:univ-rel-cond-typ-sm-trnonincreasing} of
  \cref{defn:universal-relative-conditional-typical-smoother}.
  
  Let $\ket\sigma_{XR_X}$ be any pure state, and define
  $\ket\rho_{X'ER_X} = V_{X\to X'E}\,\ket\sigma_{XR_X}$. By construction,
  $\DD`\big{\rho_{EX'}}{(V_{X\to X'E}\Gamma_{X}V^\dagger)} - \DD{\rho_{X'}}{\Gamma_{X'}}
  = \DD{\sigma_{X}}{\Gamma_{X}} - \DD{\mathcal{E}(\sigma_{X})}{\Gamma_{X'}}
  \geqslant x$. Then Property~\ref{item:univ-rel-cond-typ-sm-high-state-weight} of
  \cref{thm:universal-relative-conditional-typical-smoother} tells us that there
  exists a $\xi>0$ independent of both $\rho$ and $n$ such that
  \begin{align}
    \Re`*{
    \bra{\rho}_{X'ER_X}^{\otimes n}\, M^{x,\delta}_{E^nX'^n}\,
    \ket{\rho}_{X'ER_X}^{\otimes n}
    }
    \geqslant 1 - \poly(n)\,\exp(-n\xi)\ .
  \end{align}
  The conditions of \cref{prop:iid-process-diamond-norm-with-any-W} are
  fulfilled, with
  $W_{X^n\to X'^nE^n} = M^{x,\delta}_{A^nB^n} \, V_{X\to X'E}^{\otimes n}$,
  thanks furthermore to the fact that $M_{E^nX'^n}^{x,\delta}$ is
  permutation-invariant as guaranteed by
  \cref{thm:universal-relative-conditional-typical-smoother}. Hence, we have
  \begin{align}
    \frac12\norm[\big]{ \mathcal{T}_{X^n\to X'^n} -
    \mathcal{E}_{X\to X'}^{\otimes n} }_\diamond
    \leqslant \poly(n)\,\exp(-n\xi/2)\ .
  \end{align}
  For $n\in\mathbb{N}$ large enough this becomes smaller than any fixed
  $\epsilon>0$.  Furthermore, by
  Property~\ref{item:univ-rel-cond-typ-sm-conditional-Gamma-weight} of
  \cref{defn:universal-relative-conditional-typical-smoother}, we have that
  \begin{align}
    \mathcal{T}_{X^n\to X'^n}`\big(\Gamma_{X}^{\otimes n})
    &= \tr_{E^n}`\big[ M^{x,\delta}_{E^nX'^n} \, `\big(V_{X\to X'E} \,
    \Gamma_X V_{X\leftarrow X'E}^\dagger )^{\otimes n}
    (M^{x,\delta}_{E^nX'^n})^\dagger ]
    \nonumber\\
    &\leqslant
    \poly(n)\,\ee^{-n(x-4\delta)}\,\Gamma_{X'}^{\otimes n}
  \end{align}
  as required.
\end{proof}

\subsection{Universal conditional and relative typical smoothing operator}
\label{sec:constr-univ-cond-rel-typ-smoother-proof}

We now turn to the proof of
\cref{thm:universal-relative-conditional-typical-smoother}, giving an explicit
construction of a universal conditional and relative typical smoothing
operator. As the proof of
\cref{thm:universal-relative-conditional-typical-smoother} is quite lengthy, it
can be instructive to consider a simpler version of our typical smoothing
operator which applies in the case where the Hamiltonians are trivial. We carry
out this analysis in \cref{appx:universal-conditional-typical-projector}.

\begin{proof}[*thm:universal-relative-conditional-typical-smoother]
  \allowdisplaybreaks
  First, we claim that we can assume $\Gamma_{AB}>0$ and $\Gamma_B'>0$ without
  loss of generality.  Indeed, if either operator is not positive definite, then
  we can first construct the operator $\widetilde{M}_{A^nB^n}^{x,\delta}$
  associated with modified operators $\widetilde{\Gamma}_{AB}>0$ and
  $\widetilde{\Gamma}_B'>0$ where all the zero eigenvalues of $\Gamma_{AB}$ and
  $\Gamma_B'$ are replaced by some arbitrary fixed strictly positive constant
  (e.g., one); we can then set
  $M_{A^nB^n}^{x,\delta} = P^{\Gamma'}_{B^n} \widetilde{M}_{A^nB^n}^{x,\delta}
  P^\Gamma_{A^nB^n} $, where $P^{\Gamma}_{A^nB^n}$ (respectively
  $P^{\Gamma'}_{B^n}$) is the projector onto the support of
  $\Gamma_{AB}^{\otimes n}$ (respectively $\Gamma_B'^{\otimes n}$).  The
  operator $M_{A^nB^n}^{x,\delta}$ constructed in this way satisfies all of the
  required properties. For the remainder of this proof we thus assume that $\Gamma_{AB}>0$ and
  $\Gamma_B'>0$.

  Let $`*{R_{A^nB^n}^k}$ be the POVM constructed by
  \cref{prop:energy-measurement-POVM-n-systems} for
  $H_{AB} = -\ln(\Gamma_{AB})$.  Similarly, let $`*{S_{B^n}^\ell}$ be the
  corresponding POVM constructed in
  \cref{prop:energy-measurement-POVM-n-systems} for
  $H_{B}' = -\ln(\Gamma'_{B})$.  Also, as before, we denote by
  $\Pi^\lambda_{A^nB^n}$ and by $\Pi^{\mu}_{B^n}$ the projectors on the
  Schur-Weyl blocks labelled by the Young diagrams
  $\lambda\in\Young{d_{A}d_{B}}{n}$ and $\mu\in\Young{d_{B}}{n}$.
  Let
  \begin{align}
    M^{x,\delta}_{A^nB^n} = \sum_{\substack{ k,\ell,\lambda,\mu\;:\\
    k-\bar{H}(\lambda)-\ell+\bar{H}(\mu) \geqslant x - 4\delta}}
    S^\ell_{B^n}\,\Pi^{\mu}_{B^n}\,\Pi^\lambda_{A^nB^n}\,R^k_{A^nB^n}\ .
    \label{eq:univ-cond-rel-typ-smoother-construction}
  \end{align}
  Note that $[S^\ell_{B^n}, \Pi^{\mu}_{B^n}] = 0$ because $S^\ell_{B^n}$ is
  permutation-invariant, and
  $[\Ident_{A^n}\otimes S^\ell_{B^n}, \Pi^{\lambda}_{A^nB^n}] = 0$ because
  $\Ident_{A^n}\otimes S^\ell_{B^n}$ is permutation-invariant.  Recall also that
  $[\Ident_{A^n}\otimes \Pi^\mu_{B^n}, \Pi^{\lambda}_{A^nB^n}] = 0$ for the same
  reason. The operator $M_{A^nB^n}^{x,\delta}$ is permutation-invariant by
  construction.  Then, we have
  \begin{align}
    M^{x,\delta\ \dagger}_{A^nB^n} M^{x,\delta}_{A^nB^n}
    &= \sum_{\substack{ k,\ell,\lambda,\mu,\\
    k',\ell',\lambda',\mu'\;:\\
    k-\bar{H}(\lambda)-\ell+\bar{H}(\mu) \geqslant x - 4\delta\\
    k'-\bar{H}(\lambda')-\ell'+\bar{H}(\mu') \geqslant x - 4\delta}}
    R^{k}_{A^nB^n}\,\Pi^{\lambda}_{A^nB^n}\, \Pi^{\mu}_{B^n}\, S^{\ell}_{B^n}
    S^{\ell'}_{B^n}\, \Pi^{\mu'}_{B^n} \, \Pi^{\lambda'}_{A^nB^n} \, R^{k'}_{A^nB^n}
    \nonumber\\[1ex]
    &=
    \sum_{\substack{ k,k',\ell,\lambda,\mu\;:\\
    k-\bar{H}(\lambda)-\ell+\bar{H}(\mu) \geqslant x - 4\delta\\
    k'-\bar{H}(\lambda)-\ell+\bar{H}(\mu) \geqslant x - 4\delta}}
    R^{k}_{A^nB^n } \,
    `\big(\Pi^{\lambda}_{A^nB^n}\, \Pi^{\mu}_{B^n}\, S^{\ell}_{B^n})
    \, R^{k'}_{A^nB^n}
    \nonumber\\[1ex]
    &= \sum_{k,k'}\; R^k_{A^nB^n} `*(
      \sum_{\substack{\ell,\lambda,\mu\:\\
    k-\bar{H}(\lambda)-\ell+\bar{H}(\mu) \geqslant x - 4\delta\\
    k'-\bar{H}(\lambda)-\ell+\bar{H}(\mu) \geqslant x - 4\delta
    }} \Pi^{\lambda}_{A^nB^n}\, \Pi^{\mu}_{B^n}\, S^{\ell}_{B^n}
    )  R^{k'}_{A^nB^n}
    \nonumber\\[1ex]
    &\leqslant \sum_{k,k'}\; R^k_{A^nB^n} R^{k'}_{A^nB^n}
    \nonumber\\[1ex]
    &= \sum_{k}\; R^k_{A^nB^n} = \Ident_{A^nB^n}
  \end{align}
  recalling that the operators
  $(\Pi^\lambda_{A^nB^n}, \Pi^\mu_{B^n}, S^\ell_{B^n})$ form a commuting set
  of projectors, and where in the third line the inner sum is taken to be the
  zero operator if no triplet $(\ell,\lambda,\mu)$ satisfies the given
  constraints.  This shows Property~\ref{item:univ-rel-cond-typ-sm-trnonincreasing}.

  Now, consider any state $\ket\rho_{ABR}$, where $R$ is any reference system,
  and assume that
  $\DD{\rho_{AB}}{\Gamma_{AB}} - \DD{\rho_B}{\Gamma'_B} \geqslant x$.  Rewrite
  this condition as %
  \begin{align}
    x \leqslant -H(\rho_{AB}) - \tr[\rho_{AB}\ln\Gamma_{AB}]
     + H(\rho_B) + \tr[\rho_B\ln\Gamma'_B]\ .
  \end{align}
  We write
  \begin{align}
    \bra{\rho}_{ABR}^{\otimes n} \,
    M^{x,\delta}_{A^nB^n} \, \ket{\rho}_{ABR}^{\otimes n}
    &=
    \sum_{\substack{k,\ell,\lambda,\mu\;:\\
    k-\bar{H}(\lambda)-\ell+\bar{H}(\mu)\geqslant x-4\delta}}
    \bra{\rho}_{ABR}^{\otimes n} 
    `\big(S_{B^n}^\ell \Pi_{B^n}^\mu \Pi_{A^nB^n}^\lambda R_{A^nB^n}^k) \,
    \ket\rho_{ABR}^{\otimes n}
      \nonumber\\
    &= \ \blacksquare_1\  +\  \blacksquare_2\ \ ,
    \label{eq:toreygfusdbhkjnoit903}
  \end{align}
  where we define
  \begin{subequations}
    \begin{align}
      \blacksquare_1 &=
                       \sum_{\substack{k,\ell,\lambda,\mu\;:\\
      k \geqslant-\!\tr[\rho_{AB}\ln\Gamma_{AB}] - \delta\\
      \bar{H}(\lambda) \leqslant H(\rho_{AB}) + \delta \\
      \ell \leqslant -\!\tr[\rho_{B}\ln\Gamma'_{B}] + \delta\\
      \bar{H}(\mu)\geqslant H(\rho_{B}) - \delta}}
      \bra{\rho}_{ABR}^{\otimes n} 
      `\big(S_{B^n}^\ell \Pi_{B^n}^\mu \Pi_{A^nB^n}^\lambda R_{A^nB^n}^k) \,
      \ket\rho_{ABR}^{\otimes n}\ ;
      \\
      \blacksquare_2 &=
                       \sum_{\substack{k,\ell,\lambda,\mu\;:\\
      k-\bar{H}(\lambda)-\ell+\bar{H}(\mu)\geqslant x-4\delta\hspace{1ex}\text{AND}\\
      [~
      k < -\!\tr[\rho_{AB}\ln\Gamma_{AB}] - \delta\hspace{1ex}\text{OR}\\
      \bar{H}(\lambda) > H(\rho_{AB}) + \delta\hspace{1ex}\text{OR}\\
      \ell > -\!\tr[\rho_{B}\ln\Gamma'_{B}] + \delta\hspace{1ex}\text{OR}\\
      \bar{H}(\mu) < H(\rho_{B}) - \delta
      ~]}}
      \bra{\rho}_{ABR}^{\otimes n} 
      `\big(S_{B^n}^\ell \Pi_{B^n}^\mu \Pi_{A^nB^n}^\lambda R_{A^nB^n}^k) \,
      \ket\rho_{ABR}^{\otimes n}\ ,
      \label{eq:guhilufgydsuhijflkjbaidlsj}
    \end{align}
  \end{subequations}
  further noting that the conditions in the sum defining $\blacksquare_1$ indeed
  imply that
  $k-\bar{H}(\lambda)-\ell+\bar{H}(\mu) \geqslant -\tr[\rho_{AB}\ln\Gamma_{AB}]
  - H(\rho_{AB}) + \tr[\rho_{B}\ln\Gamma'_{B}] + H(\rho_{B}) - 4\delta \geqslant
  x - 4\delta$. We first consider $\blacksquare_1$. Define the projectors
  \begin{subequations}
    \begin{align}
      X_1 &= \sum_{k\geqslant-\!\tr[\rho_{AB}\ln\Gamma_{AB}] - \delta}
            R^k_{A^nB^n}\ ;
      & X_1^\perp &= \Ident - X_1\ ;
      \\
      X_2 &= \sum_{\bar{H}(\lambda) \leqslant H(\rho_{AB}) + \delta}
            \Pi_{A^nB^n}^\lambda\ ;
      & X_2^\perp &= \Ident - X_2\ ;
      \\
      X_3 &= \sum_{\bar{H}(\mu) \geqslant H(\rho_{B}) - \delta}
            \Pi_{B^n}^\mu\ ;
      & X_3^\perp &= \Ident - X_3\ ;
        \\
      X_4 &=
            \sum_{\ell\leqslant-\!\tr[\rho_{B}\ln\Gamma'_{B}] + \delta}
            S^\ell_{B^n}\ ;
      & X_4^\perp &= \Ident - X_4\ ,
    \end{align}
  \end{subequations}
  and observe that
  \begin{align}
    \Re`*{\ \blacksquare_1\ }\ =\ 
    \Re`*{\bra{\rho}_{ABR}^{\otimes n} \; `\big(\;X_4\;X_3\;X_2\;X_1\;) \;
    \ket{\rho}_{ABR}^{\otimes n} }\ .
  \end{align}
  Thanks to \cref{prop:energy-measurement-POVM-n-systems}, we have
  $\norm{\;X_1^\perp\; \ket\rho_{ABR}^{\otimes n}} \leqslant 2\exp(-n\eta/2)$,
  recalling that $\norm{P\ket\psi} = \sqrt{\tr[P\psi]}$, and hence
  \begin{align}
    \hspace*{3em}
    &\hspace*{-3em}
    \Re`*{\bra{\rho}_{ABR}^{\otimes n} \; X_4\;X_3\;X_2\;X_1 \;
    \ket{\rho}_{ABR}^{\otimes n} }\notag
    \\
    &= \Re`*{\bra{\rho}_{ABR}^{\otimes n} \; X_4\;X_3\;X_2 \;
    \ket{\rho}_{ABR}^{\otimes n} }
    - \Re`*{\bra{\rho}_{ABR}^{\otimes n} \; X_4\;X_3\;X_2\;X_1^\perp \;
      \ket{\rho}_{ABR}^{\otimes n} }\notag
    \\
    &\geqslant \Re`*{\bra{\rho}_{ABR}^{\otimes n} \; X_4\;X_3\;X_2 \;
    \ket{\rho}_{ABR}^{\otimes n} } - 2\exp(-n\eta/2)
  \end{align}
  using Cauchy-Schwarz to assert that
  $\Re(\braket\chi\psi) \leqslant \abs{\braket\chi\psi} \leqslant
  \norm{\ket\chi}\,\norm{\ket\psi}$. Similarly, using \cref{prop:entropy-measurement-POVM-n-systems}, we have
  $\norm{\;X_2^\perp\; \ket\rho_{ABR}^{\otimes n}} \leqslant
  \poly(n)\exp(-n\eta/2)$.  Also, we have
  $\norm{\;X_3^\perp\; \ket\rho_{ABR}^{\otimes n}} \leqslant
  \poly(n)\exp(-n\eta/2)$, and
  $\norm{\;X_4^\perp\; \ket\rho_{ABR}^{\otimes n}} \leqslant 2\exp(-n\eta/2)$,
  yielding
  \begin{align}
    \Re`*{\bra{\rho}_{ABR}^{\otimes n} \; X_4\;X_3\;X_2 \;
    \ket{\rho}_{ABR}^{\otimes n} }
    &\geqslant 
    \Re`*{\bra{\rho}_{ABR}^{\otimes n} \; X_4\;X_3 \;
      \ket{\rho}_{ABR}^{\otimes n} } - \poly(n)\,\exp(-n\eta/2)\ ;
      \\
    \Re`*{\bra{\rho}_{ABR}^{\otimes n} \; X_4\;X_3 \;
    \ket{\rho}_{ABR}^{\otimes n} }
    &\geqslant 
    \Re`*{\bra{\rho}_{ABR}^{\otimes n} \; X_4 \;
      \ket{\rho}_{ABR}^{\otimes n} } - \poly(n)\,\exp(-n\eta/2)\ ;
      \\
    \Re`*{\bra{\rho}_{ABR}^{\otimes n} \; X_4 \;
    \ket{\rho}_{ABR}^{\otimes n} }
    &\geqslant 1 - 2\,\exp(-n\eta/2)\ .
  \end{align}
  We take all these $\eta$'s to be the same, by choosing if necessary the
  minimum of the four possibly different $\eta$s. Hence, we have
  \begin{align}
    \Re`*{\ \blacksquare_1\ }\ \geqslant \ 1 - \poly(n)\,\exp(-n\eta/2)\ .
  \end{align}
  Now we consider the term $\blacksquare_2$.  We know that
  \begin{subequations}
    \begin{align}
      \norm*{ R^k_{A^nB^n} \ket\rho_{ABR}^{\otimes n} }
      &\leqslant \exp(-n\eta/2)
      &&\text{if $k<-\tr[\rho_{AB}\ln\Gamma_{AB}] - \delta$}\ ; \\
      \norm*{ \Pi_{A^nB^n}^\lambda \ket{\rho}_{ABR}^{\otimes n} }
      &\leqslant \poly(n)\exp(-n\eta/2)
      &&\text{if $\bar{H}(\lambda) > H(\rho_{AB}) + \delta$}\ ;\\
      \norm*{ S^\ell_{B^n} \ket{\rho}_{ABR}^{\otimes n} }
      &\leqslant \exp(-n\eta/2)
      &&\text{if $\ell>-\tr[\rho_{B}\ln\Gamma'_{B}] + \delta$}\ ;\\
      \norm*{ \Pi_{B^n}^\mu \ket{\rho}_{ABR}^{\otimes n} }
      &\leqslant \poly(n)\exp(-n\eta/2)
      &&\text{if $\bar{H}(\mu) < H(\rho_{B}) - \delta$}
    \end{align}
  \end{subequations}
  recalling that $\norm{P\ket\psi} = \sqrt{\tr[P\psi]}$.  So, for each term in
  the sum~\eqref{eq:guhilufgydsuhijflkjbaidlsj}, we have
  \begin{align}
      \abs*{
      \bra{\rho}_{ABR}^{\otimes n} 
      `\big(S_{B^n}^\ell \Pi_{B^n}^\mu \Pi_{A^nB^n}^\lambda R_{A^nB^n}^k) \,
      \ket\rho_{ABR}^{\otimes n}
      }
    &=
      \abs*{
      `\big(\bra{\rho}_{ABR}^{\otimes n} 
      S_{B^n}^\ell \Pi_{B^n}^\mu \Pi_{A^nB^n}^\lambda) `\big(R_{A^nB^n}^k \,
      \ket\rho_{ABR}^{\otimes n})
      }
      \nonumber\\
    &\leqslant
      \norm*{\; R_{A^nB^n}^k \, \ket\rho_{ABR}^{\otimes n}\; }
      \cdot \norm*{\; `\big(S_{B^n}^\ell \Pi_{B^n}^\mu \Pi_{A^nB^n}^\lambda) \,
      \ket\rho_{ABR}^{\otimes n}\; }
      \nonumber\\
    &\leqslant \poly(n)\,\exp(-n\eta/2)
  \end{align}
  using the Cauchy-Schwarz inequality and because at least one of the four
  conditions is violated, causing at least one of the two the norms to decay
  exponentially (noting also that $S_{B^n}^\ell, \Pi_{B^n}^\mu, \Pi_{A^nB^n}^\lambda$ all commute).  Because
  there are only at most $\poly(n)$ terms, we have
  \begin{align}
    \abs{\;\blacksquare_2\;}\ &\leqslant
    \sum_{\substack{k,\ell,\lambda,\mu\;:\\
    k-\bar{H}(\lambda)-\ell+\bar{H}(\mu)\geqslant x-4\delta\hspace{1ex}\text{AND}\\
    [~
    k < -\!\tr[\sigma_X\ln\Gamma_X] - \delta\hspace{1ex}\text{OR}\\
    \bar{H}(\lambda) > H(\sigma_X) + \delta\hspace{1ex}\text{OR}\\
    \ell > -\!\tr[\rho_{X'}\ln\Gamma_{X'}] + \delta\hspace{1ex}\text{OR}\\
    \bar{H}(\mu)< H(\mathcal{E}(\sigma_X)) - \delta
    ~]}}
    \abs*{
      \bra{\rho}_{ABR}^{\otimes n} 
    `\big(S_{B^n}^\ell \Pi_{B^n}^\mu \Pi_{A^nB^n}^\lambda R_{A^nB^n}^k) \,
    \ket\rho_{ABR}^{\otimes n}
    }
    \nonumber\\
    &\leqslant \poly(n)\exp(-n\eta/2)\ .
  \end{align}
  Hence, we have
  \begin{align}
    \Re`*{ \bra{\rho}_{ABR}^{\otimes n} \,
    M^{x,\delta}_{A^nB^n} \, \ket{\rho}_{ABR}^{\otimes n}
      }\ 
    &=\ 
    \Re`*{\ \blacksquare_1\ }\ +\ \Re`*{\ \blacksquare_2\ }
      \nonumber\\
    &\geqslant\ 
    \Re`*{\ \blacksquare_1\ }\ -\ \abs*{\ \blacksquare_2\ }
      \nonumber\\
    &\geqslant\  1 - \poly(n)\exp(-n\eta/2)
  \end{align}
  proving Property~\ref{item:univ-rel-cond-typ-sm-high-state-weight} for
  $\xi=\eta/2$.  Note that $\xi$ does not depend on the state
  $\ket\sigma_{XR}$. Now, we prove
  Property~\ref{item:univ-rel-cond-typ-sm-conditional-Gamma-weight}. Using
  \cref{lemma:Ai-coherent-can-ignore} and dropping some subsystem indices for
  readability, we have
  \begin{align}
    &\tr_{A^n} `\big[ M^{x,\delta}_{A^nB^n} \Gamma_{AB}^{\otimes n}
    `\big(M^{x,\delta}_{A^nB^n})^\dagger ]
    \notag\\
    &\leqslant \poly(n) \sum_{\substack{k,\ell,\lambda,\mu\;:\\
    k-\bar{H}(\lambda)-\ell+\bar{H}(\mu) \geqslant x-4\delta}}
    \tr_{A^n}`*[ S^\ell \Pi^\mu \Pi^\lambda R^k
    \, \Gamma^{\otimes n} \,
    R^k \Pi^\lambda \Pi^\mu S^\ell]\ .
    \label{eq:klBHjVg}
  \end{align}
  Recall that, using \cref{prop:energy-measurement-POVM-n-systems} and
  \cref{lemma:Schur-Weyl-bipartite-trAn-PiLambdaPiLambdaprime},
  \begin{align}
    R_{A^nB^n}^k \, \Gamma_{AB}^{\otimes n}
    &\leqslant \ee^{-nk}\, R_{A^nB^n}^k\leqslant \ee^{-nk}\, \Ident_{A^nB^n}\ ;
    \\
    \Pi^\mu_{B^n} \tr_{A^n}`*[ \Pi^\lambda_{A^nB^n}] \Pi^\mu_{B^n}
    &\leqslant \poly(n)\,\exp(n(\bar{H}(\lambda)-\bar{H}(\mu)))\,\Ident_{B^n}\ ;
    \\
    S_{B^n}^\ell
    &\leqslant \ee^{n\ell}\; S_{B^n}^\ell \, \Gamma_{B}'^{\otimes n}
    \leqslant \ee^{n\ell}\; \Gamma_{B}'^{\otimes n}
  \end{align}
  further recalling that $[R_{A^nB^n}^k, \Gamma_{AB}^{\otimes n}] = 0$ and
  $[S_{B^n}^\ell , \Gamma_{B}'^{\otimes n}] = 0$.  Combining these together
  yields
  \begin{align}
    \text{\eqref{eq:klBHjVg}}
    &\leqslant \poly(n) \sum_{\substack{ k,\ell,\lambda,\mu\;:\\
    k-\bar{H}(\lambda)-\ell+\bar{H}(\mu) \geqslant x-4\delta}}
    \ee^{-nk}\,
    S^\ell\, \Pi^\mu \, \tr_{A^n}`\big[\Pi^\lambda_{A^nB^n}] \, \Pi^\mu\, S^\ell
    \nonumber\\
    &\leqslant \sum_{\substack{ k,\ell,\lambda,\mu\;:\\
    k-\bar{H}(\lambda)-\ell+\bar{H}(\mu) \geqslant x-4\delta}}
    \poly(n)\,\ee^{-nk+n(\bar{H}(\lambda)-\bar{H}(\mu))}\, S^\ell_{B^n}
    \nonumber\\
    &\leqslant \sum_{\substack{ k,\ell,\lambda,\mu\;:\\
    k-\bar{H}(\lambda)-\ell+\bar{H}(\mu) \geqslant x-4\delta}}
    \poly(n)\,\ee^{-n(k-\bar{H}(\lambda)+\bar{H}(\mu)-\ell)}\, \Gamma_{B}'^{\otimes n}
    \nonumber\\[.7ex]
    &\leqslant 
    \poly(n)\,\ee^{-n(x-4\delta)}\, \Gamma_{B}'^{\otimes n}\ .
  \end{align}

  Finally, suppose that $[\Gamma_{AB}, \Gamma_B'] = 0$, meaning that we can
  choose a simultaneous eigenbasis for $\Gamma_{AB}$ and $\Gamma_{B'}$.  Then
  the operator $M_{A^nB^n}^{x,\delta}$ is a projector, as can be seen
  in~\eqref{eq:univ-cond-rel-typ-smoother-construction} since in that case
  $`{ S_{B^n}^\ell }, `{ \Pi^\mu_{B^n} }, `{ \Pi^\lambda_{A^nB^n} }, `{
    R^k_{A^nB^n} }$ are all complete sets of projectors all elements of which
  commute pairwise between different sets.  Furthermore,
  $\Gamma_{B'}^{\otimes n}$ and $\Gamma_{AB}^{\otimes n}$ both commute with all
  of these projectors and therefore also with $M_{A^nB^n}^{x,\delta}$.
\end{proof}

\section{Construction \#3: Thermal operations}
\label{sec:optimal-universal-protocol-TO-timecovariant}

\subsection{Statement and proof sketch}
We now present a construction of a universal thermodynamic implementation of a
time-covariant i.i.d.\@{} process, using the framework of thermal operations instead
of Gibbs-preserving maps.

\begin{theorem}
  \label{x:universal-magical-thermal-op}
  Let $X$ be a quantum system, $H_X$ a Hermitian operator, $\beta\geqslant 0$,
  $\mathcal{E}_{X\to X}$ a completely positive, trace-preserving map satisfying
  \begin{align}
    \mathcal{E}_{X\to X}(\ee^{-iH_Xt}\,(\cdot)\,\ee^{iH_Xt})
    = \ee^{-iH_Xt}\,\mathcal{E}_{X\to X}(\cdot)\,\ee^{iH_Xt}
    \qquad\text{for all \(t\in\mathbb{R}\).}
  \end{align}
  Let $\epsilon>0$.  Let $\delta>0$ be small enough and $n\in\mathbb{N}$ be
  large enough.  Then, there exists an information battery $W$, a thermal
  operation $\Phi_{X^nW}$, and battery states $\tau_W^{(\mathrm{i})}$ and
  $\tau_W^{(\mathrm{f})}$ such that:
  \begin{enumerate}[label=(\roman*)]
  \item The effective work process $\mathcal{T}_{X^n\to X^n}$ associated with
    $\Phi_{X^nW}$ and
    $\left(\tau_W^{(\mathrm{i})}, \tau_W^{(\mathrm{f})}\right)$ satisfies
    \begin{align}
      \frac12\norm{ \mathcal{T}_{X^n\to X^n} - \mathcal{E}_{X\to X'}^{\otimes n} }_\diamond
      \leqslant \epsilon\ ;
    \end{align}
  \item The work cost per copy satisfies
    \begin{align}
      \lim_{\delta\to 0} \lim_{n\to \infty}
      \frac1n\left[w\left(\tau_W^{\mathrm{(i)}}\right) - w\left(\tau_W^{\mathrm{(f)}}\right)\right]
      =  T(\mathcal{E})\ .
    \end{align}
  \end{enumerate}
\end{theorem}

The main idea in the present construction is to first carry out a Stinespring
dilation unitary explicitly using suitable ancillas as the environment system,
and then to apply a conditional erasure process that resets the ancillas to a
standard state while using the output of the process as side information. The idea of implementing a process in this fashion was also employed in Ref.~\cite{Faist2015NatComm}.

Our core technical contribution for Construction \#3 is to show how to build a
thermodynamic protocol for universal conditional erasure, using the idea of
position-based decoding~\cite{Anshu2017PRL_convexsplit,Anshu2019IEEETIT_oneshot,Anshu2018IEEETIT_redistribution,Anshu2018IEEETIT_SlepianWolf,Anshu2019IEEETIT_compression,Majenz2017PRL_catdecoupling,Anshu2020ITIT_partially,Berta18}. The assembly of the full thermal operation is slightly more involved than
Constructions~\#1 and~\#2, because we cannot use \cref{prop:fw-T-battery-characterization}. The construction will be illustrated in \cref{fig:ConstrProofThermalOp-IIDCase}, using a conditional erasure
primitive whose construction is illustrated in \cref{fig:ConstrProofThermalOp}.

\subsection{Universal conditional erasure}

Conditional erasure is a task that is of independent interest because it
generalizes Landauer's erasure principle to situations where a quantum memory is
available.  A protocol for thermodynamic conditional erasure of a system using a
memory as quantum side information was given in ref.~\cite{delRio2011Nature} for
trivial Hamiltonians.  Here, we study the problem of finding a universal
protocol for conditional erasure, whose accuracy is guaranteed for any input
state on $n$ copies of a system, and where the system and memory Hamiltonians
can be arbitrary.

\begin{definition}[Universal conditional erasure]%
  \label{defn:universal-conditional-erasure}
  Consider two systems $S,M$.  Let $\sigma_S$ be a fixed state, let
  $\mathcal{S}_{SM} = `{ \rho_{SM} }$ be an arbitrary set of states on
  $S\otimes M$, and let $\delta'\geqslant 0$. A \emph{universal conditional
    $\delta'$-erasure process} of $S$ using $M$ as side information is a
  completely positive, trace non-increasing map $\mathcal{T}_{SM\to SM}$ such
  that for all $\rho_{SM}\in \mathcal{S}_{SM}$, and writing $\ket{\rho}_{SMR}$ a
  purification of $\rho_{SM}$, we have
  \begin{align}
    F`\big( \mathcal{T}_{SM\to SM}(\rho_{SMR}),  \sigma_S \otimes \rho_{MR} )
    \geqslant 1 - \delta'\ .
  \end{align}
\end{definition}

We provide a thermodynamic protocol for universal conditional erasure.

\begin{proposition}\label{thm:univ-conditional-erasure-PBD-as-thermalop}
  Let $S,M$ be systems with Hamiltonians $H_S,H_M$ and let $\gamma_S$ refer to
  the thermal state on $S$.  Let $\mathcal{S}_{SM}$ be an arbitrary set of
  states on $S\otimes M$.  Let $m\geqslant 0$ such that $e^m$ is integer.  Let
  $P_{SM}$ be a Hermitian operator satisfying
  $0\leqslant P_{SM}\leqslant \Ident$ and $[P_{SM}, H_S+H_M]=0$, and assume that
  there exists $\kappa,\kappa'\geqslant 0$ such that for all
  $\rho_{SM}\in\mathcal{S}_{SM}$ we have
  \begin{subequations}\label{eq:univ-isom-cond-eras-conditions-hypo-test}
    \begin{align}
      \tr`\big[P_{SM}\,\rho_{SM}]
      &\geqslant 1 - \kappa\ ;
      \\
      \tr`*[P_{SM}\,`*(\gamma_S\otimes\rho_{M})]
      &\leqslant \frac{\kappa'}{e^m}\ .
    \end{align}
  \end{subequations}
  Then, there exists a thermal operation $\mathcal{R}_{SMJ\to SMJ}$ acting on the
  systems $SM$ and an information battery $J$, such that the effective work
  process $\mathcal{T}_{SM\to SM}$ of $\mathcal{R}_{SMJ\to SMJ}$ with respect to
  the battery states $(\tau_J^{m},\ket0_J)$ is a universal conditional
  $(2\kappa+4\kappa')$-erasure process with $\sigma_S=\gamma_S$ for the set of
  states $\mathcal{S}_{SM}'$, where $\mathcal{S}_{SM}'$ is the convex hull of
  $\mathcal{S}_{SM}$.
\end{proposition}

The proof of \cref{thm:univ-conditional-erasure-PBD-as-thermalop} is developed
in the rest of this section. We start by reformulating the ideas of the
convex-split lemma, the position-based decoding, and the catalytic decoupling
schemes~\cite{Anshu2017PRL_convexsplit,Anshu2019IEEETIT_oneshot,Anshu2018IEEETIT_redistribution,Anshu2018IEEETIT_SlepianWolf,Anshu2019IEEETIT_compression,Majenz2017PRL_catdecoupling,Anshu2020ITIT_partially,Berta18}
to form a protocol for universal conditional erasure. The underlying ideas of
the following proposition are the same as, e.g., in
Ref.~\cite{Anshu2019IEEETIT_oneshot}.  Yet, our technical statement differs in
some aspects and that is why we provide a proof for completeness. The setting is
depicted in \cref{fig:ConstrProofThermalOp}.

\begin{figure}
  \centering
  \includegraphics{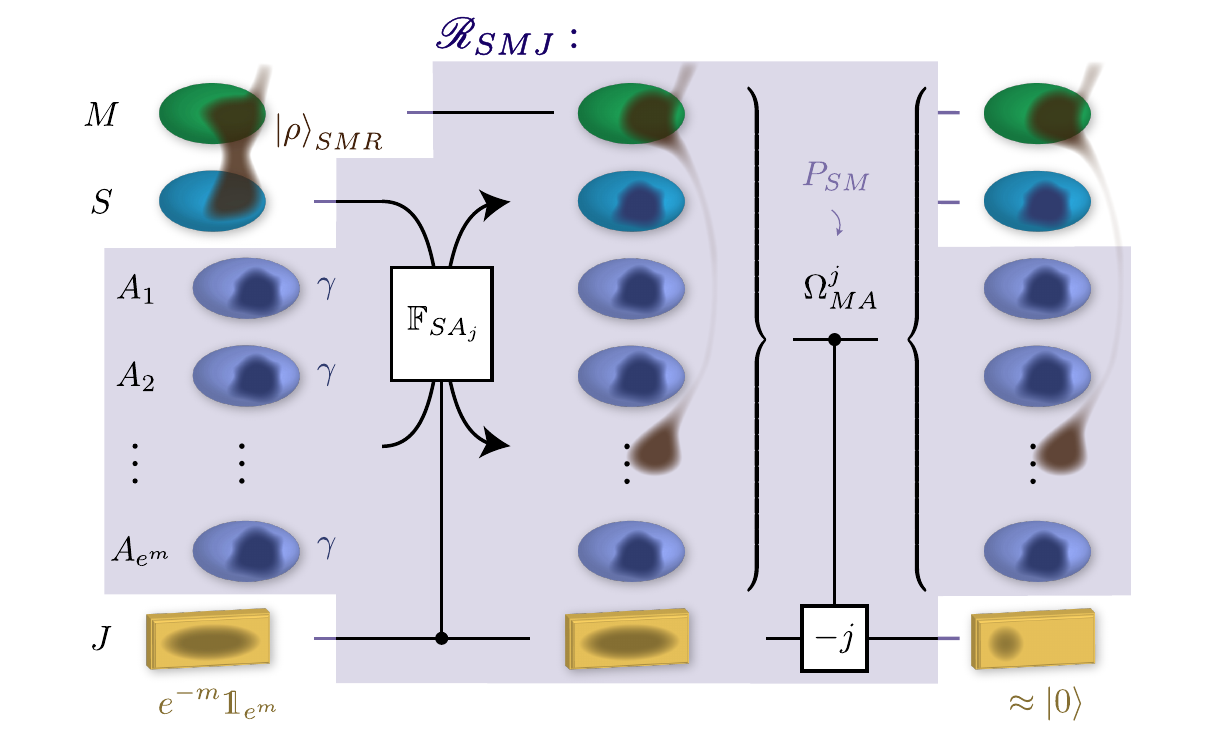}
  \caption{Construction of the thermal operation for universal conditional
    erasure using position-based decoding~\protect\cite{Anshu2019IEEETIT_oneshot},
    illustrating the construction in the proof of
    \cref{thm:univ-conditional-erasure-PBD-as-thermalop} and
    \cref{x:universal-isometry-cond-erasure-J}.  We define a map
    $\mathcal{R}_{SMJ}$ that acts on a system $S$ to reset, a quantum memory $M$
    and a register $J$, which is promised to be initialized in the uniformly
    mixed state $e^{-m}\Ident_{e^{m}}$ of rank $e^m$ for a fixed and known value
    of $m$.  A state $\rho_{SM}$ of the system and the memory is purified by a
    reference system $R$ (not pictured).  The map $\mathcal{R}_{SMJ}$ outputs
    the system $S$ in a state close to the thermal state $\gamma_S$ and the
    register $J$ in a state close to the pure state $\ket0_J$, all while
    ensuring that $\rho_{MR}$ remains unchanged (up to small errors), for all
    states $\rho_{SM}$ in a given class of states $\mathcal{S}_{SM}$.  The
    routine is provided a POVM effect $P_{SM}$ whose task is to distinguish
    $\rho_{SM}$ from $\gamma_S\otimes\rho_M$ in a hypothesis test for all
    $\rho_{SM}\in\mathcal{S}_{SM}$.  As long as $m$ is not too large (as
    determined by how well $P_{SM}$ can perform this distinguishing), the
    procedure completes successfully.  To implement $\mathcal{R}_{SMJ}$ (shaded
    region) we involve $\ee^m$ ancillas $A = A_1\ldots A_{\ee^m}$ with
    $A_j \simeq S$, each initialized in the thermal state
    $\gamma_{A_j} = \gamma_S$.  Then $S$ and $A_j$ are coherently swapped
    ($\mathbb{F}_{SA_j}$) conditioned on the value stored in $J$.  If $m$ is not
    too large, a POVM $`{ \Omega^j_{MA} }$ can infer the value $j$ stored in
    $J$, up to a small error; the POVM is constructed from $P_{SM}$.  We then
    coherently reset the $J$ register to zero by conditioning on this outcome
    (up to a small error).  The full procedure is a thermal operation where the
    ancillas are the heat bath and $J$ is an information battery such that $m$
    work has been extracted in units of pure nats (see main text).}
  \label{fig:ConstrProofThermalOp}
\end{figure}

\begin{lemma}[Conditional erasure unitary using position-based decoding]
  \label{x:universal-isometry-cond-erasure-J}
  Consider two systems $S,M$ and fix $m\geqslant 0$ such that $e^m$ is
  integer. Let $J$ be a large register of dimension at least $2e^m$, and choose
  a fixed basis $`{\ket{j}_J}$. Now, let $\gamma_S$ be any state,
  $\mathcal{S}_{SM}$ an arbitrary set of quantum states on $S\otimes M$,
  $P_{SM}$ a Hermitian operator satisfying $0\leqslant P_{SM} \leqslant\Ident$,
  and assume that there exists $\kappa,\kappa'\geqslant 0$ such that for all
  $\rho_{SM}\in\mathcal{S}_{SM}$ the
  conditions~\eqref{eq:univ-isom-cond-eras-conditions-hypo-test} hold.
  Furthermore, let $A = A_1\otimes \cdots \otimes A_{e^m}$ be a collection of
  ancilla systems with each $A_j\simeq S$, and let
  $A' = A'_1\otimes \cdots\otimes A'_{e^m}$ be a copy of the full collection of
  ancilla systems. We write a purification of $\gamma_{A_j}$ on $A'_j$ as
  $\ket\gamma_{A_jA'_j} = \gamma_{A_j}^{1/2}\,\ket\Phi_{A_j:A'_j}$.  Let
  $\mathcal{S}_{SM}'$ be the convex hull of $\mathcal{S}_{SM}$.  Then, there
  exists a unitary operator $W^{(m)}_{SMAJ\to SMAJ}$ satisfying the following
  property: For any reference system $R$, for any pure tripartite state
  $\ket\rho_{SMR}$ with $\rho_{SM}\in\mathcal{S}_{SM}'$, and for any $\ket{j}_J$
  with $1\leqslant j\leqslant e^m$, we have
  \begin{align}
    \Re`*{
    `\big( \bra{\hat\tau^j(\rho_{SMR})}_{RMSAA'}\otimes\bra0_J ) \;
    W^{(m)}_{SMAJ} \;
      `\big(\ket\rho_{RMS}\otimes
          \ket\gamma_{A_\cdot A'_\cdot}^{\otimes e^m}\otimes \ket{j}_J)
      }
      \geqslant 1 - (2\kappa+4\kappa')\ ,
    \label{eq:overlap-Wm-on-rhoinit-with-tauj-final}
  \end{align}
  where we have defined
  \begin{align}
    \ket{\hat\tau^j(\rho_{SMR})}_{RMSAA'}
    = 
    \ket\rho_{A_jMR}\otimes\ket\gamma_{SA'_j}\otimes
    [\ket\gamma^{\otimes({e^m}-1)}]_{AA'\setminus A_jA'_j}
    \label{eq:tau-j-after-swap-with-superancilla}
  \end{align}
  and by the notation $AA'\setminus A_jA'_j$ we refer to all $AA'$ systems
  except $A_j A'_j$. Moreover, for any observables $H_S$, $H_M$ such that
  $[P_{SM}, H_S+H_M] = 0$, the unitary $W^{(m)}_{SMAJ}$ may be chosen such that
  $[H_S + H_M + \sum H_{A_j}, W^{(m)}_{SMAJ}] = 0$, where $H_{A_j} = H_S$.
\end{lemma}

Intuitively, we absorb the initial randomness present in the register $J$, e.g.,
given to us by the environment in a mixed state, and return it in a pure state;
$J$ can therefore be identified as an information battery.  Similarly, $A$ can
be identified as a heat bath.

\begin{proof}[**x:universal-isometry-cond-erasure-J]
  First observe that we can assume $\mathcal{S}_{SM}$ to be a convex set,
  because any convex combination of states in $\mathcal{S}_{SM}$ also satisfies
  the conditions~\eqref{eq:univ-isom-cond-eras-conditions-hypo-test}.  For the
  rest of the proof we assume without loss of generality that
  $\mathcal{S}_{SM}=\mathcal{S}'_{SM}$.

  The operator $W$ is defined in two steps. The first operation simply consists
  on conditionally swapping $S$ with $A_j$, depending on the value stored in
  $J$.  Then, we infer again from $MA$ which $j$ we swapped $S$ with, in order
  to coherently reset the register $J$ back to the zero state
  (approximately). We define the first unitary operation as $W^{(1)}$, acting on
  systems $SAJ$
  \begin{align}
    W^{(1)}_{SAJ} =  \sum_{j} \mathbb{F}_{SA_{j}}\otimes\proj{j}_J\ ,
  \end{align}
  where $\mathbb{F}_{SA_j}$ denotes the swap operator between the two designated
  systems. Observe that $W^{(1)}$ maps $\rho$ onto $\hat\tau^j$ according to
  \begin{align}
    \hspace*{5em}
    &\hspace*{-5em}
      W^{(1)}_{SQJ}
    `*( \ket\rho_{RMS}\otimes \ket\gamma_{A_\cdot A'_\cdot}^{\otimes {e^m}}\otimes \ket{j}_J )
      \nonumber\\
    &=\quad
        \ket\rho_{RMA_j} \otimes \ket\gamma_{SA_j'} \otimes
        \left[\ket\gamma^{\otimes ({e^m}-1)}\right]_{AA'\setminus A_jA_j'} \otimes \ket{j}_J
        \nonumber\\
    &=\quad
      \ket{\hat\tau^j}_{SRMAA'}\otimes \ket{j}_J\ .
  \end{align}
  The second step is more tricky. We need to infer from the systems $MA$ alone which
  $j$ was stored in $J$.  Fortunately the answer is provided in the form of
  position-based decoding~\cite{Anshu2019IEEETIT_oneshot}, using a pretty good
  measurement.  Define
  \begin{align}
    \Lambda^j_{MA} = P_{MA_j}\otimes\Ident_{A\setminus A_j}
  \end{align}
  such that $`{\Lambda^j_{MA}}$ is a set of positive operators.  We can form a
  POVM $`{ \Omega^j_{MA} }_j \cup `{ \Omega^\perp_{MA} }$ by normalizing the
  $\Lambda^j$'s as follows:
  \begin{align}
    \Omega^j_{MA}
    &= \Lambda_{MA}^{-1/2}\,\Lambda^j_{MA}\,\Lambda_{MA}^{-1/2}\ ;
    &
    \Lambda_{MA}
    &= \sum_j \Lambda^j_{MA}\ ;
    &
      \Omega^\perp_{MA}
    &= \Ident - \sum_j \Omega^j_{MA}\ .
  \end{align}
  We would now like to lower bound $\tr[\Omega^j_{MA}\hat\tau^j_{MA}]$.
  Following the proof
  of~\cite[\theoremautorefname~2\null]{Anshu2019IEEETIT_oneshot}, we first invoke
  the Hayashi-Nagaoka inequality~\cite{Hayashi2003IEEETIT_formulas}, which
  states that for any operators $0\leqslant A \leqslant \Ident$, $B\geqslant 0$,
  we have
  \begin{align}
    \Ident - (A+B)^{-1/2}\,A\,(A+B)^{-1/2} \leqslant 2(\Ident - A) + 4B\ .
  \end{align}
  Applying this inequality with $A = \Lambda^j_{MA}$ and
  $B=\sum_{j'\neq j} \Lambda^{j'}_{MA}$ we obtain
  \begin{align}
    \tr`*[`*(\Ident-\Omega^j)\hat\tau^j_{MA}]
    &\leqslant 2 \tr`*[`*(\Ident - \Lambda^j_{MA}) \hat\tau^j_{MA}]
      + 4\sum_{j'\neq j}\tr`*[\Lambda^{j'}_{MA}\hat\tau^j_{MA}]
      \nonumber\\
    &\leqslant 2 \tr`*[`*(\Ident - P_{SM}) \rho_{SM}]
      + 4m\,\tr`*[ P_{SM} `*( \gamma_S\otimes\rho_M) ]
      \nonumber\\
    &\leqslant 2\kappa + 4\kappa'\ .
      \label{eq:ofieruygt8ydsui}
  \end{align}
Now, let $\SHIFTop_J(x) = \sum_j \ketbra{j+x}{j}_J$
  denote the SHIFT operation on the $J$ register, modulo ${e^m}$; note that $`\big(
  \SHIFTop_J(x) )^\dagger = \SHIFTop_J(-x)$. We define
  \begin{align}
    W^{(2)}_{MAJ}
    &=`*( \sum_{j} \Omega^{j}_{MA}\otimes \SHIFTop_J(-j) )\ ;
    &
    W'_{SMAJ}
    &=
    W^{(2)}_{MAJ}
    W^{(1)}_{SAJ}
  \end{align}
  and we see that $W'^\dagger W'\leqslant\Ident$ thanks to
  \cref{x:POVM-controlled-unitary}. Then, we have
  \begin{align}
    \hspace*{5em}
    &\hspace*{-5em}
    W'_{SMAJ} \,
    `*(\ket\rho_{RMS}\otimes \ket\phi_{A_\cdot A'_\cdot}^{\otimes {e^m}}\otimes \ket{j}_J)
      \nonumber\\
    &=\quad
    `*( \sum_{j'} \Omega^{j'}_{MA}\otimes \SHIFTop_J(-j') ) \;
    `*( \ket{\hat\tau^j}_{SRMAA'}\otimes \ket{j}_J )
      \nonumber\\
    &=\quad
      \sum_{j'} `*( \Omega^{j'}_{MA}\,\ket{\hat\tau^j}_{RMSAA'} ) \otimes \ket{j-j'}\ .
  \end{align}
  Thanks to \cref{x:complete-subunitary-to-unitary}, the operator $W'_{SMAJ}$ can be completed to a
  full unitary $W_{SMAJ}$ by using an extra qubit in the $J$ register, and such
  that $\bra0_J W_{SMAJ} \ket{j}_J = \bra0_J W'_{SMAJ} \ket{j}_J$ for all
  $j=1,\ldots,{e^m}$ (with the convention that $\ket{j}_J$ for
  $j\leqslant {e^m}$ forces the extra qubit to be in the zero state). So,
  recalling~\eqref{eq:ofieruygt8ydsui},
  \begin{align}
    \hspace*{3em}
    &\hspace*{-3em}
    `*( \bra{\hat\tau^j}_{RMSAA'}\otimes\bra{0}_J )
    W_{SMAJ} \,
    `*(\ket\rho_{RMS}\otimes \ket\phi_{A_\cdot A'_\cdot}^{\otimes {e^m}}\otimes \ket{j}_J)
    \nonumber\\
    &=
    `*( \bra{\hat\tau^j}_{RMSAA'}\otimes\bra{0}_J )
    W'_{SMAJ} \,
    `*(\ket\rho_{RMS}\otimes \ket\phi_{A_\cdot A'_\cdot}^{\otimes {e^m}}\otimes \ket{j}_J)
    \nonumber \\
    &= \matrixel{\hat\tau^j}{\Omega^j_{MA}}{\hat\tau^j}_{RMSAA'}
    \nonumber \\
    &\geqslant 1 - (2\kappa+4\kappa')\ .
  \end{align}
  To prove the last part of the claim, let $H_S, H_M$ be observables such that
  $[P_{SM}, H_S+H_M] = 0$ and $[H_S, \gamma_S]=0$. Let $H_{A_j} = H_S$ and we
  write $H_A = \sum_j H_{A_j}$.  For all $j$, we have
  \begin{align}
    [H_S + H_M + H_A, \Lambda^j_{MA}]
    = `\big[H_S + {\textstyle \sum_{j'\neq j}} H_{A_{j'}}, \Lambda^j_{MA}] +
    [H_M + H_{A_j}, P_{MA_j}]
    = 0\ .
  \end{align}
  This implies that $[H_S + H_M + H_A, \Lambda_{MA}] = 0$, and in turn
  $`\big[H_S + H_M + H_A, \Lambda_{MA}^{-1/2}] = 0$, and thus also
  $[H_S + H_M + H_A, \Omega^j] = 0$.  Hence, we have
  \begin{align}
    `\big[H_S + H_M + H_A, W^{(2)}_{MAJ}] = 0\ .
  \end{align}
  Clearly, $[H_S + H_M + H_A, W^{(1)}_{SAJ}] = 0$, and hence
  $[H_S+H_M+H_A, W'_{SMAJ}] = 0$.  Using
  \cref{x:complete-energy-preserving-subunitary-to-unitary} instead of
  \cref{x:complete-subunitary-to-unitary}, we may further enforce
  $[H_S + H_M + H_A, W_{SMAJ}] = 0$, as required.
\end{proof}

We now give the proof of \cref{thm:univ-conditional-erasure-PBD-as-thermalop}.

\begin{proof}[*thm:univ-conditional-erasure-PBD-as-thermalop]
  Let $W^{(m)}_{SMAJ}$ be the energy-conserving unitary as in
  \cref{x:universal-isometry-cond-erasure-J} and define the thermal operation
  \begin{align}
    \mathcal{R}_{SMJ}(\cdot) = \tr_{A}`\big[ W^{(m)}_{SMAJ}
    `\big((\cdot)\otimes\gamma_{A}) W^{(m)\,\dagger}_{SMAJ}]\ .
    \label{eq:univ-cond-eras-PBD-Rmap}
  \end{align}
  Identifying $J$ as an information battery, the associated effective work
  process of $\mathcal{R}_{SMJ}$ with respect to $(\tau_J^m,\ket0_J)$ is
  \begin{align}
    \mathcal{T}_{SM\to SM}(\cdot)
    &= \tr_{A}`\big[ \bra0_J W^{(m)}_{SMAJ}
      `\big((\cdot)\otimes\gamma_A\otimes\tau_J^{m})
      W^{(m)\,\dagger}_{SMAJ} \ket0_J ]\ .
  \end{align}
  Let $\rho_{SM}\in\mathcal{S}_{SM}'$ and let $\ket\rho_{SMR}$ be a purification
  of $\rho_{SM}$.  We have that the state vector
  \begin{align}
  e^{-m/2}\sum_{j=1}^{e^m} \bra0_J
  W_{SMAJ}^{m}`\big(\ket{\rho}_{SMR}\otimes\ket\gamma_{AA'}^{\otimes
    e^m}\otimes \ket{j}_J) \otimes\ket{j}_{R_J}
  \end{align}
  is a purification of
  $\mathcal{T}_{SM\to SM}(\rho_{SMR})$, where $R_J$ is an additional register.
  Similarly, the state vector
  \begin{align}
  e^{-m/2}\sum_{j=1}^{e^m}\ket{\hat\tau^j(\rho_{SMR})}_{RMSAA'}\otimes\ket{j}_{R_J}
  \end{align}
  is a purification of $\gamma_S\otimes\rho_{MR}$. Then, with Uhlmann's theorem
  we find
  \begin{align}
    \hspace*{1em}&\hspace*{-1em}
    F`\big( \mathcal{T}_{SM\to SM}(\rho_{SMR}) , \gamma_S\otimes\rho_{MR})
    \nonumber\\
    &\geqslant
      e^{-m}\sum_{j=1}^{e^m}
      \Re`*{
    `\big( \bra{\hat\tau^j(\rho_{SMR})}_{RMSAA'}\otimes\bra0_J ) \;
    W^{(m)}_{SMAJ} \;
      `\big(\ket\rho_{RMS}\otimes
          \ket\gamma_{A_\cdot A'_\cdot}^{\otimes e^m}\otimes \ket{j}_J)
      }
    \nonumber\\
    &\geqslant
      1 - (2\kappa+4\kappa')\ ,
  \end{align}
  making use of~\eqref{eq:overlap-Wm-on-rhoinit-with-tauj-final}.
\end{proof}

\subsection{Construction via universal conditional erasure}

This section is devoted to the proof of \cref{x:universal-magical-thermal-op}.
The strategy is to exploit the fact that time-covariant processes admit a
Stinespring dilation with an energy-conserving unitary using an environment
system with a separate Hamiltonian. This property enables us to map the problem
of implementing such a process directly to a conditional erasure problem with a
system and memory that are non-interacting.

The following lemma formalizes the property of time-covariant processes we make
use of. Various proofs of this lemma can be found
in~\cite{Scutaru1979RMP_covariant},~\cite[Appendix~B]{Keyl1999JMP_cloning}
and~\cite[Theorem 25]{PhDMarvian2012_symmetry}.

\begin{lemma}[Stinespring dilation of covariant  processes~{\protect\cite{Scutaru1979RMP_covariant,Keyl1999JMP_cloning,PhDMarvian2012_symmetry}}]%
  \label{x:covariant-process-H-env}
  Let $X$ be a quantum system with Hamiltonian $H_X$, and $\mathcal{E}_{X\to X}$
  be a completely positive, trace-preserving map that is covariant with respect
  to time evolution. That is, for all $t$ we have
  \begin{align}
    \mathcal{E}_{X\to X}(\ee^{-iH_X t}\,(\cdot)\,\ee^{iH_X t})
    = \ee^{-iH_X t}\, \mathcal{E}_{X\to X}(\cdot)\,\ee^{iH_X t} \ .
  \end{align}
  Then, there exists a system $E$ with Hamiltonian $H_E$ including an eigenstate $\ket{0}_E$ of zero energy, as well as a unitary $V_{EX\to EX}$ such that
  \begin{align}
    \mathcal{E}_{X\to X}(\cdot) = \tr_E`*[ V\,`*(\proj0_E \otimes (\cdot))\,V^\dagger]
    \label{eq:EXX-covariant-Stinespring}
  \end{align}
  as well as $V\, `(H_X+H_E) \, V^\dagger = H_X + H_E$.
\end{lemma}

We provide an additional proof in \cref{app:missing}. The main idea behind the
construction in the following proof of \cref{x:universal-magical-thermal-op} is
depicted in \cref{fig:ConstrProofThermalOp-IIDCase}.

\begin{figure}
  \centering
  \includegraphics{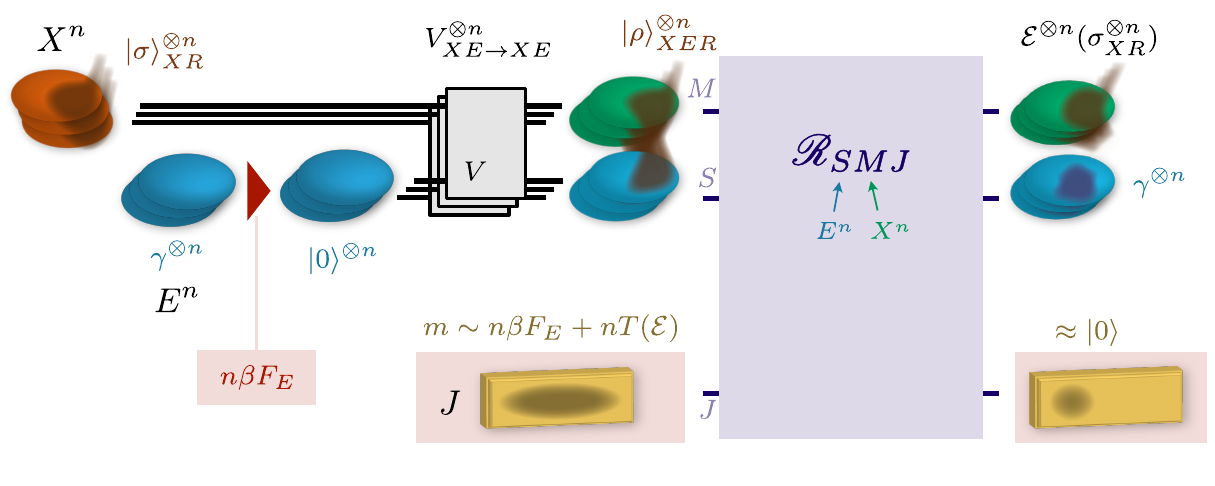}
  \caption{The conditional erasure procedure in \cref{fig:ConstrProofThermalOp}
    can be used to construct an i.i.d.\@{} implementation of a given time-covariant
    process (\cref{x:universal-magical-thermal-op}).  First we apply an
    energy-conserving Stinespring dilation of the process on all input copies,
    using a zero-initialized ancilla as environment system $E$ for each copy.
    We then invoke the conditional erasure procedure $\mathcal{R}_{E^nX^nJ}$ to
    reset $E^n$ to the thermal state $\gamma_E^{\otimes n}$ using $X'^n$ as a
    memory, while extracting work using an information battery $J$.  Here, the
    projector that can distinguish $\rho_{EX'}^{\otimes n}$ from
    $\Ident_{E^n}\otimes\rho_{X'^n}$ is the universal conditional typical
    projector given by \cref{prop:univ-quant-cond-typical-subspace}.  The fact
    that $\mathcal{R}_{E^nX^nJ}$ preserves the correlations
    $[\mathcal{E}(\sigma_{XR})]^{\otimes n}$ between the memory (output systems
    $X'^n$) and the reference $R^n$ ensures that the process is implemented
    accurately.  The amount of work extracted by $\mathcal{R}_{E^nX^nJ}$ is
    $m \sim n[\beta F_E+T(\mathcal{E})]$ but $\sim n \beta F_E$ work has to be
    paid to prepare the initially pure $E^n$ ancillas, where
    $\beta F_E=-\ln\tr(\ee^{-\beta H_E})$.  The overall work extracted is
    $\sim T(\mathcal{E})$ per copy.}
  \label{fig:ConstrProofThermalOp-IIDCase}
\end{figure}

\begin{proof}[*x:universal-magical-thermal-op]
  \allowdisplaybreaks
  Thanks to \cref{x:covariant-process-H-env}, there exists an environment system
  $E$ with Hamiltonian $H_E$, as well as an energy-conserving unitary $V_{XE}$
  and a state $\ket0_E$ of zero energy such
  that~\eqref{eq:EXX-covariant-Stinespring} holds.  Let
  $F_E = -\beta^{-1}\ln(Z_E)$ with $Z_E = \tr[\ee^{-\beta H_E}]$. We define
  \begin{align}
    x = \min_\sigma `*{ \DD{\sigma}{\ee^{-\beta H_X}}
    - \DD{\mathcal{E}(\sigma)}{\ee^{-\beta H_X}} }
    = - T(\mathcal{E}) \ .
  \end{align}
  Writing $\rho_{XE} = V_{XE} `*(\proj0_E\otimes\sigma_X) V_{XE}^\dagger$, we
  have that
  $x = \min_{\sigma_X} \bigl\{ -\Sfn(\sigma_X) + \beta\tr[\sigma_X H_X] +
    \Sfn(\rho_{X}) - \beta\tr[\rho_X H_X] \bigr\}$. By
  $\tr[\sigma_X H_X] = \tr`\big[(\proj0_E\otimes\sigma_X)(H_X + H_E)] =
  \tr`\big[\rho_{XE} `*(H_X+H_E)]$, we see that
  \begin{align}
    x  = \min_{\sigma_X} `*{
    -\Sfn(\rho_{XE}) + \Sfn(\rho_{X}) + \beta\tr[\rho_E H_E]
    }\ .
  \end{align}
  Observe that for any such $\rho_{XE}$, we have
  \begin{align}
    -\HH[\rho]{E}[X] + \beta \tr[\rho_E H_E]
    &\geqslant -\HH[\rho]{E} + \beta \tr[\rho_E H_E] + \ln(Z) - \ln(Z)
      \nonumber\\
    &= \DD{\rho_E}{\gamma_E} + \beta F_E
      \geqslant \beta F_E
  \end{align}
  using the sub-additivity of the von Neumann entropy and the fact that relative
  entropy is positive for normalized states. Hence, we have $x\geqslant \beta F_E$.

  Let
  \begin{align}
    \mathcal{S}_{E^nX^n}
    = `\big{ \rho_{EX}^{\otimes n}\ :\ \rho_{EX} =
      V_{XE}(\proj0_E\otimes\sigma_X)V_{XE}^\dagger\ \text{for some \(\sigma_X\)} }\ ,
  \end{align}
  noting that for all $\rho_{EX}^{\otimes n}\in\mathcal{S}_{E^nX^n}$, we have
  $\DD{\rho_{EX}}{\ee^{-\beta(H_X+H_E)}} - \DD{\rho_X}{\ee^{-\beta H_X}} =
  \DD{\sigma}{\ee^{-\beta H_X}} - \DD{\mathcal{E}(\sigma)}{\ee^{-\beta H_X}}
  \geqslant x$.  Let $P_{E^nX^n}^{x,\delta}$ be the universal typical and
  relative conditional operator furnished by
  \cref{thm:universal-relative-conditional-typical-smoother}, where
  $\Gamma_X = \ee^{-\beta H_X}$ and
  $\Gamma_{XE} = \ee^{-\beta`(H_X+H_E)} = \Gamma_X\otimes\Gamma_E$ with
  $\Gamma_E = \ee^{-\beta H_E}$.  Since $\Gamma_{XE}$ commutes with
  $\Ident_E\otimes\Gamma_X$,
  \cref{thm:universal-relative-conditional-typical-smoother} guarantees that
  $P_{E^nX^n}^{x,\delta}$ is a projector which furthermore commutes with
  $\Gamma_{XE}^{\otimes n}$ and $\Gamma_X^{\otimes n}$.  We proceed to show that
  $P_{E^nX^n}^{x,\delta}$ can perform a hypothesis test between
  $\rho_{EX}^{\otimes n}$ and $\gamma_E^{\otimes n}\otimes\rho_X^{\otimes n}$.
  Recalling \cref{defn:universal-relative-conditional-typical-smoother} we have
  \begin{align}
    \tr`\big[ P_{E^nX^n}^{x,\delta} \rho_{EX}^{\otimes n} ]
    &\geqslant 1- \kappa\ ,
  \end{align}
  with $\kappa = \poly`(n)\,\ee^{-n\eta}$ for some $\eta>0$ independent of
  $\rho$ and $n$.  By construction we have
  $\Ident_X\otimes\Gamma_E = \Gamma_X^{-1/2}\Gamma_{XE}\Gamma_X^{-1/2}$, and so
  thanks to Point~\ref{item:univ-rel-cond-typ-sm-conditional-Gamma-weight} of
  \cref{defn:universal-relative-conditional-typical-smoother} we can compute
  \begin{align}
    \tr_{E^n}`\big[ P_{E^nX^n}^{x,\delta} \Gamma_E^{\otimes n} ]
    &= 
    `\big(\Gamma_{X}^{-1/2})^{\otimes n}
    \tr_{E^n}`\big[ P_{E^nX^n}^{x,\delta} \Gamma_{XE}^{\otimes n} ]
    `\big(\Gamma_{X}^{-1/2})^{\otimes n}
      \nonumber\\
    &\leqslant \poly`(n) \exp`(-n(x-4\delta))\,\Ident_{X^n}\ ,
  \end{align}
  where we furthermore used the fact that $P_{E^nX^n}^{x,\delta}$ commutes with
  $\Gamma_{XE}^{\otimes n}$ and with $\Gamma_{X}^{\otimes n}$.  We therefore see
  using $\gamma_E = \Gamma_E/\tr[\Gamma_E]$ that
  \begin{align}
    \tr`\big[ P_{E^nX^n}^{x,\delta} \rho_X^{\otimes n}\otimes\gamma_E^{\otimes n} ]
    &\leqslant \frac1{\tr[\Gamma_E^{\otimes n}]} \, \poly(n)\exp`\big(-n(x-4\delta))\,
    \tr`\big[\rho_X^{\otimes n}]
    \nonumber\\
    &= \poly(n)\,\exp`\big(-n(x - \beta F_E - 4\delta))\ .
  \end{align}
  Let
  \begin{align}
    \ee^m = 
    \bigl\lfloor \exp`\big{n(x - \beta F_E - 4\delta - \eta)} \bigr\rfloor\ ,
    \label{eq:3g28foikjls}
  \end{align}
  such that
  $\tr`\big[ P_{E^nX^n}^{x,\delta} \rho_X^{\otimes n}\otimes\gamma_E^{\otimes n}
  ] \leqslant e^{-m}\kappa'$ by choosing $\kappa' = \poly(n)\ee^{-n\eta}$.

  Now let $J$ be a register of dimension at least $2\ee^m$ and let
  $\mathcal{R}_{E^nX^nJ}$ be the thermal operation furnished by
  \cref{thm:univ-conditional-erasure-PBD-as-thermalop} for $S=E^n$, $M=X^n$,
  $\mathcal{S}_{E^nX^n}$, $P_{E^nX^n}^{x,\delta}$, $m$, $\kappa$, and $\kappa'$
  as defined above.  Here, we have assumed that $x > \beta F_E$, and that
  furthermore $\delta,\eta$ are small enough such that
  $4\delta+\eta < (x-\beta F_E)$; if instead $x=\beta F_E$ then we can set
  $e^m=1$ and
  $\mathcal{R}_{E^nX^nJ}(\cdot) = \tr_{E^n}(\cdot)\otimes\gamma_E^{\otimes n}$
  (which is a thermal operation) in the following.

  We proceed to show that the effective work process
  $\mathcal{T}^{\mathcal{R}}_{E^nX^n\to E^nX^n}$ of $\mathcal{R}_{E^nX^nJ}$ with respect to
  $(\tau_J^m,\ket0_J)$ is close to the partial trace map
  $\mathcal{T}_{E^nX^n\to E^nX^n}^{(0)}(\cdot) =
  \tr_{E^n}(\cdot)\otimes\gamma_E^{\otimes n}$ in diamond distance.  We invoke
  the post-selection technique (\cref{x:post-selection-technique}) to show this.
  Let $\zeta_{E^nX^n}$ be the de Finetti state which
  via~\eqref{eq:de-Finetti-state-purified-with-poly-iid-states} can be written
  as the convex combination of a finite number of i.i.d.\@{} states
  \begin{align}
    \zeta_{E^nX^n} = \sum p_i \phi_i^{\otimes n}\ .
  \end{align}
  Hence $\zeta_{E^nX^n}$ lies in the convex hull of $\mathcal{S}_{E^nX^n}$, and
  from \cref{thm:univ-conditional-erasure-PBD-as-thermalop} and
  \cref{defn:universal-conditional-erasure} we see that for a purification
  $\ket\zeta_{E^nX^nR}$ of $\zeta_{E^nX^n}$ we have
  \begin{align}
    F`\big( \mathcal{T}^{\mathcal{R}}_{E^nX^n\to E^nX^n}(\zeta_{E^nX^nR}), 
    \gamma_E^{\otimes n}\otimes \tr_{E^n}(\zeta_{E^nX^nR}) )
    \geq 1 - (2\kappa+4\kappa')\ .
  \end{align}
  Using $D(\rho,\sigma) \leq \sqrt{1 - F(\rho,\sigma)}$ along with
  \cref{x:post-selection-technique} we find
  \begin{align}
    \frac12\norm{
    \mathcal{T}^{\mathcal{R}}_{E^nX^n\to E^nX^n} - \mathcal{T}_{E^nX^n\to E^nX^n}^{(0)}
    }_{\diamond}
    \leqslant \sqrt{2\kappa + 4\kappa'} = \poly(n)\,\ee^{-n\eta/2}\ .
  \end{align}

  We can start piecing together the full process.  Our overall protocol needs to
  (a)~bring in a heat bath $E^n$, i.e., ancillas initialized in their thermal
  state, (b)~prepare the states $\ket0_E^{\otimes n}$ on the ancillas using an
  auxiliary information battery (denoted by $W'$ below), (c)~apply the
  energy-conserving unitary $V_{XE}^{\otimes n}$, (d)~apply
  $\mathcal{R}_{E^nX^nJ}$ using an information battery $J$ initialized in the
  state $\tau_J^m$, and (e)~discard the ancillas.

  As explained in \cref{sec:resource-theory-thermodynamics}, there exists a
  thermal operation $\widetilde\Phi_{E^nW'}$ on the ancillas and an information
  battery $W'$ along with battery states $(\tau_{W'}^{(1)}, \tau_{W'}^{(2)})$
  such that
  $\widetilde\Phi_{E^nW'}( \gamma_E^{\otimes n} \otimes \tau_{W'}^{(1)} ) =
  \proj0_E^{\otimes n}\otimes \tau_{W'}^{(2)}$ and with
  $w(\tau_{W'}^{(1)}) - w(\tau_{W'}^{(2)})$ arbitrarily close to $-\beta n F_E$.
  Now let $W = J\otimes W'$,
  $\tau_W^{(\mathrm{i})} = \tau_{W'}^{(1)}\otimes \tau_J^{m}$,
  $\tau_W^{(\mathrm{f})} = \tau_{W'}^{(2)}\otimes\proj0_J$, and define
  \begin{align}
    \Phi_{X^nW}(\cdot)
    &= \tr_{E^n}`\Big[
      \mathcal{R}_{E^nX^nJ}`\Big(
      V_{XE}^{\otimes n} \;
      \widetilde\Phi_{E^nW'}`\big( (\cdot)\otimes\gamma_E^{\otimes n}) \;
      (V_{XE}^{\otimes n})^\dagger
      )
      ]\ .
  \end{align}
  The map $\Phi_{X^nW}$ is a thermal operation because it is a concatenation of
  thermal operations.  The overall heat bath is formed of the systems $E^n$, the
  ancillas $A^n$ used in the implementation of $\mathcal{R}_{E^nX^nJ}$, as well
  as the implicit heat bath used in the implementation of
  $\widetilde{\Phi}_{E^nW'}$.  The system $W=J\otimes W'$ is the information
  battery.  We can verify that the associated effective work process with
  respect to $`\big(\tau_W^{(\mathrm{i})},\tau_W^{(\mathrm{f})})$ is
  \begin{align}
    \mathcal{T}_{X^n}(\cdot)
    &= 
      \bra[\big]{0}_J
      \tr_{E^n}`\Big[
      \mathcal{R}_{E^nX^nJ}`\Big(
      V_{XE}^{\otimes n} \;
      \tr_{W'}`\big[P_{W'}^{(2)}\;
      \widetilde\Phi_{E^nW'}`\big(
      (\cdot)\otimes\tau_{W'}^{(1)}\otimes\tau_J^{m}\otimes\gamma_E^{\otimes n}
      )
      ] \;
      (V_{XE}^{\otimes n})^\dagger
      )
      ]
      \ket[\big]{0}_J
      \nonumber\\
    &=
      \tr_{E^n}`\Big[
      \bra[\big]{0}_J
      \mathcal{R}_{E^nX^nJ}`\Big(
      `\big[
      V_{XE}^{\otimes n} \; `\big(
      (\cdot)\otimes\proj0_E^{\otimes n}
      )\;
      (V_{XE}^{\otimes n})^\dagger ]
      \otimes\tau_J^{m}
      )
      \ket[\big]{0}_J
      ]
      \nonumber\\
    &=
      \tr_{E^n}`\Big[
      \mathcal{T}^{\mathcal{R}}_{E^nX^n}`\Big(
      V_{XE}^{\otimes n} \; `\big(
      (\cdot)\otimes\proj0_E^{\otimes n}
      )\;
      (V_{XE}^{\otimes n})^\dagger 
      )
      ]
      \nonumber\\
    &=
      \tr_{E^n}`\Big[
      V_{XE}^{\otimes n} \; `\big(
      (\cdot)\otimes\proj0_E^{\otimes n}
      )\;
      (V_{XE}^{\otimes n})^\dagger 
      ] + \Delta_{X^n}(\cdot)
      \nonumber\\
    &=
      \mathcal{E}_{X\to X}^{\otimes n}`(\cdot)
      + \Delta_{X^n}(\cdot)\ ,
  \end{align}
  where
  $\Delta_{X^n}(\cdot) = \tr_{E^n}`\big(
  \mathcal{T}^{\mathcal{R}}_{X^nE^n}(\cdot) - \mathcal{T}^{(0)}_{X^nE^n}(\cdot)
  )$ satisfies
  $(1/2)\norm{\Delta_{X^n}}_\diamond \leqslant \poly(n)\ee^{-n\eta/2}$.
  Therefore for any fixed $\epsilon$ and for $n$ large enough we have
  $(1/2)\norm{\mathcal{T}_{X^n} - \mathcal{E}_{X\to X}^{\otimes n}}_\diamond
  \leqslant \epsilon$.

  The associated work cost per copy satisfies
  \begin{align}
    \lim_{\delta\to0} \lim_{n\to\infty}
    \frac1n`\big[w(\tau_W^{(\mathrm{i})}) - w(\tau_W^{(\mathrm{f})})]
    &= \lim_{\delta\to0} \lim_{n\to\infty}
      \frac1n`\big[
      w(\tau_{W'}^{(1)}) - w(\tau_{W'}^{(2)}) - m
      ]
      \nonumber\\
    &=
    \lim_{\delta\to0} \lim_{n\to\infty}
      \frac1n`\big[
      -n\beta F_E
      - n(x-\beta F_E - 4\delta + \eta)
      + \upsilon
      ]
      \nonumber\\
    &=
      T(\mathcal{E})\ ,
  \end{align}
  recalling~\eqref{eq:3g28foikjls}, where $0 \leqslant \upsilon \leqslant 2$
  accounts for the rounding error in~\eqref{eq:3g28foikjls} and a possible
  arbitrarily small difference between $-n\beta F_E$ and
  $w(\tau_{W'}^{(1)}) - w(\tau_{W'}^{(2)})$, and recalling that $\eta\to 0$ as
  $\delta\to 0$.
\end{proof}

\section{Related Results}
\label{sec:bonus-results}

\subsection{Alternative proof of the asymptotic equiparition
  property of the coherent relative entropy}
\label{appx:AEP-coh-rel-entr}

Here we have a new proof of the asymptotic equipartition property (AEP) for the
coherent relative entropy, with an explicit expression of a smoothing process
which does the job.  This isn't directly useful for our universality result, but
it uses similar ideas and gives some intuition about the AEP of the coherent
relative entropy.  Furthermore, we get an explicit process that is near-optimal
in the definition of the coherent relative entropy for an i.i.d.\@{} process matrix.

Recall the definition of the relative typical
subspace~\cite{Bjelakovic2003arXiv_revisted,Berta2015QIC_monotonicity}:

\begin{proposition}[Relative typical
  projector~\protect\cite{Bjelakovic2003arXiv_revisted,Berta2015QIC_monotonicity}]
  \label{x:relative-typical-projector}
  Let $\rho,\tau\geqslant 0$ be operators on a finite dimensional
  Hilbert space $\Hs$ with $\tr(\rho)=1$, and let $\delta>0$.  There
  exists a constant $\eta>0$, and for all $n$ there exists a
  projector $\Pi^{n,\delta}_{\rho|\tau}$ such that the following conditions
  hold:
  \begin{subequations}
    \begin{gather}
      `\big[\Pi^{n,\delta}_{\rho|\tau},\tau^{\otimes n}]=0 \ ;
      \label{eq:rel-typ-proj--commuting-tau}\\
      \ee^{-n`(\MM{\rho}{\tau}+\delta)}\, \Pi^{n,\delta}_{\rho|\tau} \;\leqslant\;
      \Pi^{n,\delta}_{\rho|\tau}\,\tau^{\otimes n}\,\Pi^{n,\delta}_{\rho|\tau} \;\leqslant\;
      \ee^{-n`(\MM{\rho}{\tau}-\delta)}\, \Pi^{n,\delta}_{\rho|\tau}\ ;
      \label{eq:rel-typ-proj--equipartition}\\
      \tr`\big(\Pi^{n,\delta}_{\rho|\tau}\,\rho^{\otimes n}) \geqslant 1 - 2\ee^{- n \eta}\ ,
      \label{eq:rel-typ-proj--exp-good-weight}
    \end{gather}
  \end{subequations}
  where we have defined
  \begin{align}
    \MM{\rho}{\tau} := -\tr(\rho\ln\tau)\ .
  \end{align}
  The usual (weakly) typical projector for a state $\rho$ is obtained by
  choosing $\tau=\rho$:
  \begin{align}
    \Pi^{n,\delta}_\rho = \Pi^{n,\delta}_{\rho|\rho}\ .
  \end{align}
\end{proposition}

The construction of the relative typical projector, as well as the proof of
properties~\eqref{eq:rel-typ-proj--commuting-tau}
and~\eqref{eq:rel-typ-proj--equipartition} are presented in
Refs.~\cite{Bjelakovic2003arXiv_revisted,Berta2015QIC_monotonicity}.  Here we
show property~\eqref{eq:rel-typ-proj--exp-good-weight}.

\begin{proof}[*x:relative-typical-projector]
  The construction of
  Refs.~\cite{Bjelakovic2003arXiv_revisted,Berta2015QIC_monotonicity} satisfies
  properties~\eqref{eq:rel-typ-proj--commuting-tau}
  and~\eqref{eq:rel-typ-proj--equipartition}; it remains to
  prove~\eqref{eq:rel-typ-proj--exp-good-weight}.  Consider the quantity
  $\tr`\big( (\Ident - \Pi^{n,\delta}_{\rho|\tau})\,\rho^{\otimes n} )$, and
  note that it corresponds to the probability that a sequence of measurements of
  copies $\rho$ of the observable $-\ln(\tau)$ ensemble averages to a quantity
  that is $\delta$-far from $\MM{\rho}{\tau}=-\tr(\rho\ln\tau)$.  Let $Z_j$ for
  $j=1,\ldots,n$ be random variables where $Z_j$ is the outcome of the
  measurement of $-\ln(\tau)$ on the $j$-th system of $\rho^{\otimes n}$.  Then
  using Hoeffding's inequality we find
  \begin{align}
    \tr`\big( (\Ident - \Pi^{n,\delta}_{\rho|\tau})\,\rho^{\otimes n} )
    =
    \Pr`*[ \abs*{ \frac1n \sum Z_j  - `\big(-\tr(\rho\ln\tau)) } > \delta ] \leqslant
    2\exp`*( -n\eta )\ ,
  \end{align}
  for some $\eta \geqslant \delta^2/\norm{-\ln\tau}_\infty$, noting that the
  difference between the maximum and minimum eigenvalues of $-\ln\tau$ is upper
  bounded by $2\norm{-\ln\tau}_\infty$.
\end{proof}

Now we may present the new proof of the asymptotic equipartition property of the
coherent relative entropy.

\begin{proposition}
  \label{prop:AEP-coh-rel-entr-typicality}
  \allowdisplaybreaks
  Let $\Gamma_X,\Gamma_{X'}\geqslant 0$, let $R_X\simeq X$ and let
  $\ket{\sigma}_{X:R_X}$ be any state.  Write
  $\rho_{X'R_X} = \mathcal{E}(\sigma_{XR_X}) =
  \rho_{R_X}^{1/2}\,E_{X'R_X}\,\rho_{R_X}^{1/2}$ writing
  $E_{X'R_X} = \mathcal{E}(\Phi_{X:R_X})$.  Let
  $\Gamma_X,\Gamma_{X'}\geqslant 0$. For any $\delta>0$, and for any $n$, let
  \begin{equation}
    \begin{aligned}
      S_{X'^n} &= \Pi^{n,\delta}_{\rho_{X'}|\Gamma_{X'}^{-1}}\ ;
      &
      Q_{X'^n} &= \Pi^{n,\delta}_{\rho_{X'}}\ ;
      &
      P_{X^n} &= \Pi^{n,\delta}_{\sigma_{X}|\Gamma_{R_X}}\ ,
      &
      R_{X^n} &= \Pi^{n,\delta}_{\sigma_{X}}\ .
    \end{aligned}
  \end{equation}
  Then the completely positive map
  \begin{align}
    \mathcal{T}_{X^n\to X'^n}(\cdot)
    = S_{X'^n}\,Q_{X'^n}\, \mathcal{E}_{X\to X'}^{\otimes n}`\big(
    R_{X^n} \, P_{X^n}\, (\cdot) \, P_{X^n} \,R_{X^n}
    ) \,Q_{X'^n}\,S_{X'^n}\ ,
  \end{align}
  is trace-nonincreasing and satisfies
  \begin{subequations}
    \begin{gather}
      \mathcal{T}_{X^n\to X'^n}(\Gamma_X^{\otimes n})
      \leqslant
      \ee^{-n`*[ \DD{\sigma_X}{\Gamma_X} - \DD{\rho_{X'}}{\Gamma_{X'}} - 4\delta ]}\,
      \Gamma_{X'}^{\otimes n}\ ;
      \label{eq:prop-AEP-coh-rel-entr-typicality-cond-T-Gamma}
      \\
      P`\big(\mathcal{T}_{X^n\to X'^n}`\big(\sigma_{XR_X}^{\otimes n}),
      \rho_{X'R_X}^{\otimes n})
      \leqslant 4\ee^{-n\eta'}\ ,
      \label{eq:prop-AEP-coh-rel-entr-typicality-cond-T-rho-close}
    \end{gather}
  \end{subequations}
  for some $\eta'>0$.  This implies that for any $\epsilon>0$,
  \begin{align}
    \lim_{n\to\infty}
    \frac1n \DCohz[\epsilon]`\big{*\rho_{X'R_X}^{\otimes n}}{X^n}{X'^n}%
    {\Gamma_X^{\otimes n}}{\Gamma_{X'}^{\otimes n}}
    \geqslant  \DD{\sigma_X}{\Gamma_X}
    - \DD{\rho_{X'}}{\Gamma_{X'}}\ .
  \end{align}
\end{proposition}

\begin{proof}[*prop:AEP-coh-rel-entr-typicality]
  It will prove convenient to work in the purified space, so let
  $E\simeq X'\otimes R_X$, and let $\ket{E}_{X'R_XE}$ be a purification of
  $E_{X'R_X}$.  Let $V_{X\to X'E}$ be the corresponding isometry which satisfies
  $\ket{E}_{X'R_XE} = V_{X\to X'E}\,\ket\Phi_{X:R_X}$; this isometry is just a
  Stinespring dilation of $\mathcal{E}$.  Now define
  \begin{align}
    \ket{T}_{X'^nR_X^nE^n}
    = S_{X'^n}\,Q_{X'^n}\,P_{R_X^n}\,R_{R_X^n} \; \ket{E}_{X'R_XE}^{\otimes n}\ ,
  \end{align}
  where $P_{R_X^n} = t_{X^n\to R_X^n}(P_{X^n})$ and
  $R_{R_X^n} = t_{X^n\to R_X^n}(R_{X^n})$.
  We begin by showing~\eqref{eq:prop-AEP-coh-rel-entr-typicality-cond-T-Gamma}.
  Writing $\Gamma_{R_X} = t_{X\to R_X}(\Gamma_X)$, we have
  \begin{align}
    \hspace*{1em}
    &\hspace*{-1em}
      (\Gamma_{X'}^{-1/2})^{\otimes n} \, \tr_{R_X^n}`\big[
      T_{X'^nR_X^n}\,\Gamma_{R_X}^{\otimes n}
      ] \, (\Gamma_{X'}^{-1/2})^{\otimes n}
      \nonumber\\
    &= (\Gamma_{X'}^{-1/2})^{\otimes n} \,S_{X'^n} Q_{X'^n} \;
      \tr_{R_X^n}`*[
      R_{R_X^n} \, E_{X'R_X}^{\otimes n} \, R_{R_X^n} \,
      (P_{R_X^n}\,\Gamma_{R_X}^{\otimes n}\,P_{R_X^n})
      ] \;
      Q_{X'^n} S_{X'^n}\,(\Gamma_{X'}^{-1/2})^{\otimes n}
      \nonumber\\
    &\leqslant \ee^{-n(\MM{\sigma_X}{\Gamma_X}-\delta)}\;
      (\Gamma_{X'}^{-1/2})^{\otimes n} \,S_{X'^n} Q_{X'^n} \;
      \tr_{R_X^n}`*[
      R_{R_X^n} \, E_{X'R_X}^{\otimes n} \, R_{R_X^n}
      ] \;
      Q_{X'^n} S_{X'^n}\,(\Gamma_{X'}^{-1/2})^{\otimes n} \ ,
      \label{eq:calc-fudbiafdosafnjdsa}
  \end{align}
  recalling that
  $P_{R_X^n}\,\Gamma_{R_X}^{\otimes n}\,P_{R_X^n} \leqslant
  \ee^{-n(\MM{\sigma_X}{\Gamma_X}-\delta)}\; \Ident_{R_X^n}$.  Now define
  $R_{X'^nE^n}$ as the dual projector of $R_{R_X^n}$ with respect to
  $\ket{E}_{X'R_XE}^{\otimes n}$: Indeed, we have
  $\ket{E}_{X'R_XE} = V_{X\to X'E}\ket{\Phi}_{X:R_X}$; we may thus define
  $R_{X'^nE^n} = (V_{X\to X'E})^{\otimes n}\, R_{X^n} \,
  (V^\dagger)^{\otimes n}$ in such a way that
  $R_{X'^nE^n}\,\ket{E}_{X'R_XE}^{\otimes n} = V_{X\to X'E}^{\otimes n}\,
  `\big(R_{X^n} \otimes \Ident_{R_X^n})\,\ket\Phi_{X:R_X} =
  R_{R_X^n}\,\ket{E}_{X'R_XE}^{\otimes n}$.  Then compute
  \begin{align}
    \hspace*{2em}
    &\hspace*{-2em}
    Q_{X'^n} \, \tr_{R_X^n}`*[
    R_{R_X^n} \, E_{X'R_X}^{\otimes n} \, R_{R_X^n}
    ] \, Q_{X'^n}
    \nonumber \\
    &= Q_{X'^n} \, \tr_{R_X^nE^n}`*[
    R_{X'^nE^n} \, E_{X'R_XE}^{\otimes n} \, R_{X'^nE^n}
    ] \, Q_{X'^n}
    \nonumber\\
    &= Q_{X'^n} \, \tr_{E^n}`*[
    R_{X'^nE^n} \, E_{X'E}^{\otimes n} \, R_{X'^nE^n}
    ] \, Q_{X'^n}
    \nonumber\\
    &\leqslant Q_{X'^n} \, \tr_{E^n}`*[ R_{X'^nE^n} ] \, Q_{X'^n}
    \nonumber \displaybreak[2] \\
    &\leqslant \ee^{n(S(\rho_{X'E})+\delta)}
      Q_{X'^n} \, \tr_{E^n}`*[ R_{X'^nE^n} \rho_{X'E}^{\otimes n} R_{X'^nE^n} ]
      \, Q_{X'^n}
    \nonumber\\
    &\leqslant \ee^{n(S(\rho_{X'E})+\delta)}
      Q_{X'^n} \, \tr_{E^n}`*[ \rho_{X'E}^{\otimes n} ]
      \, Q_{X'^n}
    \nonumber\\
    &= \ee^{n(S(\rho_{X'E})+\delta)}
      Q_{X'^n} \, \rho_{X'}^{\otimes n} \, Q_{X'^n}
    \nonumber\\
    &\leqslant \ee^{n(S(\rho_{X'E})-S(\rho_{X'})+2\delta)}
      Q_{X'^n}
    \nonumber\\
    &\leqslant \ee^{n(S(\sigma_{X})-S(\rho_{X'})+2\delta)}
      \Ident_{X'^n}\ .
  \end{align}
  where we have used the fact that $E_{X'E}\leqslant\Ident_{X'E}$ (since
  $E_{R_X}\leqslant\Ident_{R_X}$), the usual properties of the typical
  projectors, the fact that $[R_{X'^nE^n},\rho_{X'E}^{\otimes n}]=0$, as well as
  the fact that $S(\rho_{X'E}) = S(\rho_{R_X}) = S(\sigma_X)$ because
  $\rho_{X'R_XE}$ is a pure state.  We may then return to
  \begin{align}
    \text{\eqref{eq:calc-fudbiafdosafnjdsa}}
    &\leqslant
      \ee^{-n`(\MM{\sigma_X}{\Gamma_X}-S(\rho_{X'E})+S(\rho_{X'})-3\delta)}
      \; S_{X'^n} \, (\Gamma_{X'}^{-1})^{\otimes n} \, S_{X'^n}
      \nonumber\\
    &\leqslant
      \ee^{-n`(\MM{\sigma_X}{\Gamma_X}+\MM{\rho_{X'}}{\Gamma_{X'}^{-1}}
      -S(\rho_{X'E})+S(\rho_{X'})-4\delta)}
      \; \Ident\ ,
  \end{align}
  recalling that $S_{X'^n}$ and $(\Gamma_{X'}^{-1})^{\otimes n}$ commute.  A
  simple calculation then yields
  \begin{align}
    \hspace*{3em}
    &\hspace*{-3em}
    \MM{\sigma_X}{\Gamma_X}+\MM{\rho_{X'}}{\Gamma_{X'}^{-1}}
    -S(\sigma_{X})+S(\rho_{X'})
    \nonumber\\
    &= -\tr(\sigma_X\ln\Gamma_X) + \tr(\rho_{X'}\ln\Gamma_{X'})
    +\tr(\sigma_X\ln\sigma_X) - \tr(\rho_{X'}\ln\rho_{X'})
      \nonumber\\
    &= \DD{\sigma_X}{\Gamma_X} - \DD{\rho_{X'}}{\Gamma_{X'}}\ ,
  \end{align}
  which proves~\eqref{eq:prop-AEP-coh-rel-entr-typicality-cond-T-Gamma}.  Now we
  go for~\eqref{eq:prop-AEP-coh-rel-entr-typicality-cond-T-rho-close}.  Let
  $\eta_R, \eta_P, \eta_Q, \eta_S>0$ be the corresponding parameters provided
  by \cref{x:relative-typical-projector} for $R_{X^n}$, $P_{X^n}$,
  $Q_{X'^n}$, and $S_{X'^n}$, and let
  $\eta = \min(\eta_R, \eta_P, \eta_Q, \eta_S)$.  Then we may compute
  \begin{align}
    \hspace*{3em}
    &\hspace*{-3em}
    \Re`*{
    \bra{\rho}_{X'ER_X}^{\otimes n} `\big(\sigma_{R_X}^{1/2})^{\otimes n} \,
      \ket{T}_{X'^nR_X^nE^n} }
      \nonumber\\
    &= 
    \Re`*{
      \bra{\rho}_{X'ER_X}^{\otimes n} 
      S_{X'^n} \, Q_{X'^n} \, R_{X'^nE^n} \, P_{X'^nE^n} \,
      \ket{\rho}_{X'R_XE}^{\otimes n} }\ ,
      \label{eq:fiodanfdsoa}
  \end{align}
  where analogously to $R_{X'^nE^n}$ we define
  $P_{X'^nE^n} = V^{\otimes n} P_{X^n} (V^\dagger)^{\otimes n}$, and noting that
  $`\big(\sigma_{R_X}^{1/2})^{\otimes n}\ket{E}_{X'R_XE}^{\otimes n} =
  \ket{\rho}_{X'R_XE}^{\otimes n}$.
  Compute
  \begin{subequations}
    \begin{align}
      \norm[\big]{ (\Ident - P_{X'^nE^n}) \ket{\rho}_{X'R_XE}^{\otimes n} }^2
      &= \bra{\rho}_{X'R_XE}^{\otimes n} \, (\Ident - P_{X'^nE^n}) \,
        \ket{\rho}_{X'R_XE}^{\otimes n} \nonumber\\
      & \qquad =
        \tr`\big(\sigma_{X}\,(\Ident - P_{X^n})) \leqslant 2\exp`(-n\eta)\ ;
      \\
      \norm[\big]{ (\Ident - R_{X'^nE^n}) \ket{\rho}_{X'R_XE}^{\otimes n} }^2
      &= \tr`\big(\sigma_{X}\,(\Ident - R_{X^n})) \leqslant 2\exp`(-n\eta)\ ;
      \\
      \norm[\big]{ (\Ident - Q_{X'^n}) \ket{\rho}_{X'R_XE}^{\otimes n} }^2
      &= \tr`\big(\rho_{X'}\,(\Ident - Q_{X'^n})) \leqslant 2\exp`(-n\eta)\ ;
      \\
      \norm[\big]{ (\Ident - S_{X'^n}) \ket{\rho}_{X'R_XE}^{\otimes n} }^2
      &= \tr`\big(\rho_{X'}\,(\Ident - S_{X'^n})) \leqslant 2\exp`(-n\eta)\ ;
    \end{align}
  \end{subequations}
  exploiting each time property~\eqref{eq:rel-typ-proj--exp-good-weight} of the
  corresponding relative typical projector.
  Now we use the fact that for any states $\ket\psi,\ket{\psi'}$, and for any
  $0\leqslant X \leqslant\Ident$, we have
  $\Re`*{\matrixel{\psi'}{X}{\psi}} = \Re`*{\braket{\psi'}{\psi}} -
  \Re`*{\matrixel{\psi'}{(\Ident-X)}{\psi}} \geqslant
  \Re`*{\braket{\psi'}{\psi}} - \norm[\big]{\ket{\psi'}}
  \norm[\big]{(\Ident-X)\ket{\psi}}$, where the last inequality holds by
  Cauchy-Schwarz.  
  Then
  \begin{align}
    \text{\eqref{eq:fiodanfdsoa}}
    &\geqslant \Re`*{
      \bra{\rho}_{X'ER_X}^{\otimes n} 
      S_{X'^n} \, Q_{X'^n} \, R_{X'^nE^n} \, 
      \ket{\rho}_{X'R_XE}^{\otimes n} } - \sqrt{2}\exp(-n\eta/2)
      \nonumber\\
    &\geqslant \Re`*{
      \bra{\rho}_{X'ER_X}^{\otimes n} 
      S_{X'^n} \, Q_{X'^n} \, 
      \ket{\rho}_{X'R_XE}^{\otimes n} } - 2 \sqrt{2}\exp(-n\eta/2)
      \nonumber\\
    &\geqslant \ldots
      \nonumber\\
    &\geqslant 1 - 4\sqrt{2}\exp(-n\eta/2)\ .
  \end{align}
  Since $\ket\rho_{X'R_XE}^{\otimes n}$ is a purification of
  $\rho_{X'R_X}^{\otimes n}$, and
  $`\big(\sigma_{R_X}^{1/2})^{\otimes n}\ket{T}_{X'^nR_X^nE^n}$ is a
  purification of $\mathcal{T}_{X^n\to X'^n}(\sigma_{XR_X}^{\otimes n})$, we
  have
  \begin{align}
    F`\big(\rho_{X'R_X}^{\otimes n},
    \mathcal{T}_{X^n\to X'^n}(\sigma_{XR_X}^{\otimes n}))
    &\geqslant \abs[\big]{
      \bra{\rho}_{X'ER_X}^{\otimes n} `\big(\sigma_{R_X}^{1/2})^{\otimes n}
      \ket{T}_{X'^nR_X^nE^n}
    }
      \nonumber\\
    &\geqslant \Re`\big{
    \bra{\rho}_{X'ER_X}^{\otimes n} `\big(\sigma_{R_X}^{1/2})^{\otimes n}
      \ket{T}_{X'^nR_X^nE^n} }
      \nonumber\\
    &\geqslant 1 - 4\sqrt{2}\exp(-n\eta/2)\ .
  \end{align}
  Hence
  \begin{align}
    P`*(\rho_{X'R_X}^{\otimes n}, \mathcal{T}_{X^n\to X'^n}(\sigma_{XR_X}^{\otimes n}))
    \leqslant \sqrt{1 - (1 - 4\sqrt{2}\exp(-n\eta/2))^2}
    \leqslant \sqrt{8\sqrt{2}}\,\exp(-n\eta/4)\ ,
  \end{align}
  using the fact that $\sqrt{1 - (1-x)^2} \leqslant \sqrt{2x}$, thus
  proving~\eqref{eq:prop-AEP-coh-rel-entr-typicality-cond-T-rho-close} noting
  that $(8\sqrt{2})^{1/2}\leqslant 4$.
\end{proof}

\subsection{Optimal implementation of any i.i.d.\@{}
  channel with thermal operations on an i.i.d.\@{}
  input state}
\label{appx:impl-TO-any-iidchannel-iidinput}

Here, we show that for any channel $\mathcal{E}_{X\to X'}$, mapping a system $X$
with Hamiltonian $H_X$ to a system $X'$ with Hamiltonian $H_{X'}$, and for any
fixed input state $\sigma_X$, then there exists an optimal implementation with
thermal operations which uses a small amount of coherence, and an amount of work
per copy which is asymptotically equal to
\begin{align}
  W_\mathrm{T.O.} =
  \beta^{-1} \DD{\mathcal{E}(\sigma)}{\ee^{-\beta H_{X'}}}
  - \beta^{-1} \DD{\sigma}{\ee^{-\beta H_X}}\ .
\end{align}

The coherence is counted using a half-infinite energy latter with spacing $x$,
i.e., with Hamiltonian
$  H_C = \sum_{k=0}^{d_C} k x \proj{k}_C$,
as considered in Ref.~\cite{Aberg2014PRL_catalytic}. The state is initialized in
the state $\ket{\eta^L} := L^{-1/2} \sum_{k=0}^{L-1} \ket{\ell_0+k}$, where
$\ell_0$ is a base energy offset.  Such a system may be consumed entirely by the
process, i.e., at the end of the process we may return it in any state we like.
We assume that, in some reasonable setting, such a state on such a system can be
prepared using an amount of work of the same order as the dimension of the
system $d_C$.  In the following, $d_C$ will be sublinear in the number of copies
$n$, so asymptotically for $n\to\infty$, the coherence source will require
negligible resources to create when counted per copy.  Actually, we will use two
such sources $C_1, C_2$: We think of each of these as single-use disposable
systems, and we need such systems at two stages in our protocol.

\begin{proposition}
  \label{prop:iid-implementation-of-coherent-channel-with-thermal-operation}
  Let $X,X'$ be systems, let $H_X, H_{X'}$ be the corresponding Hamiltonians,
  and let $\beta\geqslant 0$.  Let $\mathcal{E}_{X\to X'}$ be a completely
  positive, trace-preserving map, and let $\sigma_X$ be any input state.  Let
  $R\simeq X$ and let $\ket\sigma_{XR} = \sigma_X^{1/2}\,\ket\Phi_{X:R}$ be a
  purification of $\sigma_X$.
  Then for any $0<\theta<1/3$, for any $n$ and for any $\delta>0$ there exists a
  thermal operation which acts on $X\to X'$ and an information battery $W$ whose
  process matrix is $\epsilon$-close to
  $\mathcal{E}^{\otimes n}_{X\to X'}(\sigma_{XR}^{\otimes n})$ and which uses
  two coherence sources $C_{1},C_{2}$ and an amount of work
  $W^{n}_{\mathrm{T.O.}}$, with
  \begin{subequations}
    \begin{align}
      \frac1n W^{n}_{\mathrm{T.O.}}
      &\leqslant F(\rho_{X'}, H_{X'}) - F(\sigma_X, H_X)
        + O(n^{-1/2}) + O(n^{-1}\log(1/\theta))\ ; \\
      \epsilon &\leqslant 3\theta + O(\exp(-cn\delta^2))\ ;\\
      d_{C_i}/n &\leqslant O`\big( \delta/\theta^2 )\ ,
    \end{align}
  \end{subequations}
  with $F(\sigma, H) = \DD{\sigma}{\ee^{-\beta H}}$ and for some $c>0$ depending
  only on $H_X,H_{X'}$.
\end{proposition}

The following corollary is obtained straightforwardly from the above proposition
by choosing $\delta = n^{-1/2+\xi}$, $\theta = n^{-\xi/2}$ for any choice of
$0<\xi<1/4$.

\begin{corollary}
  \label{x:impl-TO-any-iid-channel-iidinput}
  \noproofref %
  Any i.i.d.\@{} channel $\mathcal{E}^{\otimes n}_{X\to X'}$ between $n$ copies of
  systems $X, X'$ with Hamiltonians $H_X, H_{X'}$ can be implemented on a fixed
  i.i.d.\@{} input state $\sigma_X^{\otimes n}$ using thermal operations at a work
  cost rate per copy which is asymptotically equal to
  \begin{align}
    W_{\mathrm{T.O.,\ asympt.}}
    = F(\mathcal{E}(\sigma_X), H_{X'}) - F(\sigma_X, H_X)\ ,
  \end{align}
  and using a vanishing amount of coherence per copy.
\end{corollary}

We need a technical lemma that tells us that whenever we have a Hamiltonian
whose eigenvalues are all close, we may replace that Hamiltonian by an exactly
flat Hamiltonian at a cost of investing an amount of coherence of the order of
the energy spread we would like to flatten out.

\begin{lemma}[Flattening Hamiltonians using~\protect\cite{Aberg2014PRL_catalytic}]
  \label{lemma:flatten-Hamiltonian-using-coherence}
  Let $A$ be a system with a Hamiltonian $H_A$ with $m$ eigenvalues all lying in
  a range $[h_-, h_+]$.  Let $B\simeq A$ be a system with a completely
  degenerate Hamiltonian $H_B = [(h_- + h_+)/2]\, \Ident_B$.  We assume that the
  spacings of the eigenvalues of $H_A$ as well as the value $(h_-+h_+)/2$ are
  multiple of some unit $x$.  Fix $\theta > 0$.  Consider a coherence source $C$
  of dimension $d_C$, with Hamiltonian $ H_C = \sum_{k=0}^{d_C} k x
  \proj{k}_C$. %
  Assume that the coherence source starts in the state
  $\ket\eta_C = L^{-1/2}\sum_{k=0}^{L-1} \ket{\ell_0+k}$ such that
  $\ell_0 \geqslant (h_- + h_+)/x$ and $L \geqslant \theta^{-2} (h_+ - h_-)/x$
  and $d_C \geqslant L + \ell_0 + (h_-+h_+)/x$.
  Then there exists a partial isometry $U_{AC \to BC}$ which commutes exactly
  with the total Hamiltonians
  ($U_{AC\to BC}(H_A + H_C) = (H_B + H_C) U_{AC\to BC}$), such that for any
  $\rho_{AR}$ on any reference system $R$, we have that
  $\rho_{AR}\otimes\proj\eta_C$ is in the support of
  $U_{AC\to BC}\otimes\Ident_R$ and
  \begin{align}
    P`\Big(\tr_C`\big(U \,(\rho_{AR}\otimes\eta) U^\dagger), \rho_{BR})
    \leqslant \theta\ ,
  \end{align}
  where $\rho_{BR} = \IdentProc[A][B]{\rho_{AR}}$ denotes the same state as the
  initial state, but on the systems $BR$.

  The reverse operation $B\to A$ may also be carried out with the consumption of
  a similar coherence source, at the same accuracy.
\end{lemma}

\begin{proof}[*lemma:flatten-Hamiltonian-using-coherence]
  First, we can reduce the problem to a system dimension $m$: Embed the system
  into a bipartite system with a ancilla with trivial Hamiltonian storing the
  degeneracy index, and a second system with nondegenerate Hamiltonian storing
  the actual energy. Then the problem reduces to change the Hamiltonian of the
  second system.
  So we may assume without loss of generality that the Hamiltonian $H_A$ is
  nondegenerate with $m=d_A$.

  Consider the protocol of \AA{}berg~\cite{Aberg2014PRL_catalytic}.  We would
  like to apply a result in the spirit of \AA{}berg's Supplementary
  Proposition~2, but we need a tighter bound.  Denoting the energy levels of $A$
  by $H_A = \sum x z_j\proj{j}$ for integers $z_j$, we apply the global
  energy-conserving unitary on $A$ and $C$ given by
  \begin{align}
    U_{AC\to BC} = \sum_j \ket{j}_B\bra{j}_A \otimes \Delta^{z'-z_j}\ ,
  \end{align}
  where $z' = (h_- + h_+)/(2x)$ represents the fixed energy of the output.
  Then, starting in the state
  $\ket\psi_{AA'} = \sum \psi_{jj'} \ket{j\,j'}_{AA'}$ using a reference system
  $A'\simeq A$ and the initial state $\ket\eta_C$ on $C$, we have
  \begin{align}
    U_{AC\to BC} \,`*( \ket\psi_{AA'}\ket\eta_C )
    = \sum_{jj'} \psi_{jj'}  \ket{j\,j'}_{BA'}\otimes(\Delta^{z'-z_j}\ket\eta_C)\ .
  \end{align}
  We may calculate the overlap with the initial state,
  \begin{align}
    \Re`*{ \bra\psi_{BA'}\bra\eta_C\, U_{AC\to BC}\, \ket\psi_{AA'}\ket\eta_C }
    = \sum_{jj'} \abs{\psi_{jj'}}^2 \Re`*{ \bra\eta_C \Delta^{z'-z_j} \ket\eta_C }
    \geqslant 1 - \frac{h_+ - h_-}{xL}\ ,
  \end{align}
  where we used the fact that $\tr(\Delta^{-a}\,\eta_C) = \max(0, 1-a/L)$ and
  $\abs{z' - z_j} \leqslant (h_+ - h_-)/x$.  Hence, using the fact that the
  partial trace can only increase the fidelity,
  \begin{align}
    F^2`\big(\Phi(\proj\psi_{AA'}), \proj\psi_{BA'})
    \geqslant 1 - \frac{h_+ - h_-}{xL}\ ,
  \end{align}
  where $\Phi(\cdot) = %
  \tr_C`*{ U_{AC\to BC} \, `*[(\cdot)\otimes\proj\eta_C] \, U^\dagger }$.  Note
  that $\ket{\psi}_{AA'}$ is arbitrary at this point.  For any state $\rho_{AR}$,
  using the joint concavity of the fidelity function, further noting that we may
  consider without loss of generality reference systems of the form
  $A'\simeq A$, we have
  \begin{align}
    F^2(\Phi_{A\to B}(\rho_{AR}), \rho_{BR})
    \geqslant
    \min_{\ket\psi_{AA'}} F^2(\Phi_{A\to B}(\proj\psi_{AA'}), \proj\psi_{BA'})
    \geqslant 1 - \frac{h_+ - h_-}{xL}\ ,
  \end{align}
  and thus, since $L \geqslant \theta^{-2} (h_+ - h_-)/x$,
  \begin{align}
    P(\Phi_{A\to B}(\rho_{AR}), \rho_{BR})
    \leqslant \sqrt{\frac{h_+ - h_-}{xL}}
    \leqslant \theta\ .
  \end{align}
  
  The same argument can be applied to the operation $B\to A$.
\end{proof}

\begin{proof}[*prop:iid-implementation-of-coherent-channel-with-thermal-operation]
  Let $`{ R^k_{X^n} }$ be the POVM elements from
  \cref{prop:energy-measurement-POVM-n-systems} for the input energy (for the
  Hamiltonian $H_X$) over the $n$ systems, and let $`{ L^\ell_{X'^n} }$ be the
  corresponding output measurement.

  We exhibit a protocol as a sequence of gentle measurements and thermal
  operations.
  For any $\delta>0$, we measure the projector
  $R^{\approx_\delta \tr(H_X\sigma_X)}_{X^n}$ on the $n$ inputs.  This
  measurement fails with probability $\leqslant 2\exp(-\eta_1 n)$ for
  $\eta_1=2\delta^2/(\Delta H_X)^2$, where $\Delta H_X$ is the difference
  between the maximal and minimal eigenvalues of $H_X$.

  Assume that the input and output Hamiltonians $H_X$ and $H_{X'}$ have
  eigenvalues that are multiples of a spacing $x$ ($x$ may be very small).
  Now the state lies in a subspace of energies in the interval
  $[n(\tr(H_X\sigma_X) \pm \delta)]$.  We invoke
  \cref{lemma:flatten-Hamiltonian-using-coherence} to change this $n$-system
  Hamiltonian to one which is entirely flat, $H_{X^n}' := h\,\Ident_{X^n}$ with
  $h=n\tr(H_X\sigma_X)$.  Given the target approximation parameter $\theta>0$,
  the cost of this operation is the consumption of a coherence source $C_1$ of
  size
  $d_{C_1} = (\theta^{-2} + 2)\,(2n\delta/x) \leqslant O(n\delta/\theta^2)$,
  because $m = 2n\delta/x$.

  Then we invoke the achievability result of Ref.~\cite{Faist2015NatComm}, that
  one can implement any channel over trivial Hamiltonians using thermal
  operations, and an amount of work given approximately by the conditional
  entropy of the environment conditioned on the output.  We use this step to
  implement the channel
  \begin{align}
    S^{\approx_\delta \tr(H_{X'}\rho_{X'})}_{X'^n} \,
    \mathcal{E}^{\otimes n}_{X^n\to X'^n}(
    R^{\approx_\delta \tr(H_X\sigma_X)}_{X^n}\, (\cdot) \,
    R^{\approx_\delta \tr(H_X\sigma_X)}_{X^n}
    ) \,
    S^{\approx_\delta \tr(H_{X'}\rho_{X'})}_{X'^n}
  \end{align}
  up to an accuracy $\theta$ and investing an amount of work
  $n\HH[\rho]{E}[X'] + O(\sqrt{n}) + \Delta(\theta)$ with
  $\Delta(\theta) = O(\log(1/\theta))$.  We can now trivially shift the whole
  Hamiltonian $H_{X^n}' \to H'_{X'^n} := h'\Ident_{X'^n}$ with
  $h' = n\tr(H_{X'}\rho_{X'})$, investing an amount of work $h'-h$.

  Finally, we have the state
  $S^{\approx_\delta \tr(H_{X'}\rho_{X'})}_{X'^n}\, \rho_{X'R}^{\otimes n}\,
  S^{\approx_\delta \tr(H_{X'}\rho_{X'})}_{X'^n}$ but the Hamiltonian is still
  the trivial $H'_{X'^n}$.  Again we invoke
  \cref{lemma:flatten-Hamiltonian-using-coherence} to change to the final
  Hamiltonian, up to accuracy $\theta$, by consuming another coherence source
  $C_2$ of size $d_{C_2} \leqslant O(n\delta/\theta^2)$.

  Because the final state is
  $S^{\approx_\delta \tr(H_{X'}\rho_{X'})}_{X'^n}\, \rho_{X'R}^{\otimes n}\,
  S^{\approx_\delta \tr(H_{X'}\rho_{X'})}_{X'^n}$ instead of
  $\rho_{X'R}^{\otimes n}$, we again pay a ``gentle measurement penalty'' of
  $O(\ee^{-n\eta_2/2})$ where $\eta_2 \sim \delta^2$ (cf.\@
  \cref{x:gentle-measurement-lemma}).

  Finally, counting the total work, total failure probability and total use of
  coherence proves the claim, noting that
  $\tr(H_{X'}\rho_{X'}) - \tr(H_X\sigma_X) + \HH[\rho]{E}[X']
  = F(\rho_{X'}, H_{X'}) - F(\sigma_X, H_X)$.
\end{proof}

\subsection{Single-shot erasure protocol for fixed
  input state and for noninteracting system and memory}
\label{appx:one-shot-conditional-erasure-H-covariant}

In the position-based decoding of Ref.~\cite{Anshu2019IEEETIT_oneshot}, one uses
the optimal distinguishing POVM for $P_{SM}$ obtained from the hypothesis
testing entropy, and we see that for a constant error we can choose $\ln(m)$ to
be proportional to the hypothesis testing entropy.  In fact, this gives us
directly a new erasure protocol in the case of a fixed input state, and in the
case where the system and memory are not interacting ($H_{SM} = H_S + H_M$):

\begin{corollary}
  \label{x:single-shot-conditional-erasure-using-convex-split}
  Let $S,M,R$ be quantum systems, with $H_S$ and $H_M$ the Hamiltonians on $S$
  and $M$.  Let $\ket\rho_{SMR}$ be any pure state such that
  $[\rho_{SM}, H_S+H_M]=0$.  Let $\epsilon > 0$.  Suppose that
  ${\Dhyp[1-\epsilon]{\rho_{SM}}{\gamma_S\otimes\rho_M}} \geqslant
  \ln(2/\epsilon)$.  Then there exists an information battery $J$ with battery
  states $\tau_J,\tau_J'$, and there exists a thermal operation
  $\mathcal{R}_{SMJ\to SMJ}$ acting on $SMJ$ such that:
  \begin{itemize}
  \item The effective work process $\mathcal{T}_{SM\to SM}$ associated with
    $\mathcal{R}_{SMJ\to SMJ}$ and $(\tau_J,\tau_J')$ satisfies
    \begin{align}
      F`\big( \mathcal{T}_{SM\to SM}(\rho_{SMR}) \, ,\, \gamma_S\otimes\rho_{MR} )
      \geqslant 1 - 6\epsilon\ ;
      \label{eq:single-shot-conditional-erasure-using-convex-split--process-accurate}
    \end{align}
  \item The associated work cost (in units of pure nats) satisfies
    \begin{align}
      w`\big(\tau_J) - w`\big(\tau_J')
      \leqslant
      -\Dhyp[1-\epsilon]{\rho_{SM}}{\gamma_S\otimes\rho_M} + \ln(2/\epsilon)\ .
      \label{eq:single-shot-conditional-erasure-using-convex-split--work-cost}
    \end{align}
  \end{itemize}
\end{corollary}

If the register $S$ is to be returned to a pure state instead of a thermal
state, then this can be done separately as a final step, and there is a fixed
cost associated to this.  For instance, for trivial Hamiltonians, we have that
\begin{align}
  \begin{array}{c}
    \text{total work extracted in nats,}\\
    \text{also return \(S\) to pure state}
  \end{array}
  \quad
  &\geqslant\quad
    \Dhyp[1-\epsilon]`*{\rho_{SM}}{\frac{\Ident_S}{d_S}\otimes\rho_M}
    - \ln(d_S) - \ln(2/\epsilon)
    \nonumber\\
  &=\quad
    \Dhyp[1-\epsilon]`\big{\rho_{SM}}{\Ident_S\otimes\rho_M} - \ln(2/\epsilon)\ ,
\end{align}
and recalling that
$\Dhyp[1-\epsilon]{\rho_{SM}}{\Ident_S\otimes\rho_M} \approx
-\Hmaxf[\rho][\epsilon]{S}[M]$~\cite{Dupuis2013_DH,PhDTomamichel2012}, we
recover the expression in Ref.~\cite{delRio2011Nature} up to approximation
terms.

\begin{proof}[*x:single-shot-conditional-erasure-using-convex-split]
  Let $P_{SM}$ be the optimal positive semidefinite operator satisfying
  $0\leqslant P_{SM} \leqslant \Ident_{SM}$ given by
  $\Dhyp[1-\epsilon]{\rho_{SM}}{\gamma_S\otimes\rho_M}$, i.e., such that
  \begin{align}
    \tr`*( P_{SM}\,\rho_{SM} )
    &\geqslant 1 - \epsilon\ ;
    \\
    \tr`*( P_{SM}\,`*(\gamma_S\otimes \rho_{M} ) )
    &= \exp`*{-\Dhyp[1-\epsilon]{\rho_{SM}}{\gamma_S\otimes\rho_M}} \ .
  \end{align}
  We can always take $[P_{SM}, H_S+H_M]=0$ by dephasing $P_{SM}$ in energy
  blocks, if necessary, since $\rho_{SM}$ is time-covariant.  Let
  \begin{align}
    m = \ln\, \bigl\lfloor \epsilon \exp`\big{
    \Dhyp[1-\epsilon]{\rho_{SM}}{\gamma_S\otimes\rho_M} } \bigr\rfloor\ .
  \end{align}
  Then $\ee^m \leqslant %
  \epsilon \exp`\big{ \Dhyp[1-\epsilon]{\rho_{SM}}{\gamma_S\otimes\rho_M} }$ and
  \begin{align}
    \tr`*( P_{SM}\,`*(\gamma_S \otimes \rho_{M} ))
    \leqslant \frac{\epsilon}{\ee^m}\ .
  \end{align}
  We may set $\kappa=\kappa'=\epsilon$, and we are in the setting of
  \cref{thm:univ-conditional-erasure-PBD-as-thermalop} for the single-element set
  $\mathcal{S}_{SM} = `*{ \rho_{SM} }$.
  \Cref{thm:univ-conditional-erasure-PBD-as-thermalop} ensures that there exists
  a thermal operation $\mathcal{R}_{SMJ\to SMJ}$ and battery states
  $\tau_J, \tau_J'$ that satisfies the
  condition~\eqref{eq:single-shot-conditional-erasure-using-convex-split--process-accurate}.
  Furthermore according to \cref{thm:univ-conditional-erasure-PBD-as-thermalop},
  the associated work cost is
  \begin{align}
    w`\big( \tau_J ) - w`\big( \tau_J' )
    = -m\ 
    &\leqslant
    -\ln`\Big( \epsilon \exp`\big{ \Dhyp[1-\epsilon]{\rho_{SM}}{\gamma_S\otimes\rho_M}} - 1 )
      \nonumber\\
    &=
      -\Dhyp[1-\epsilon]{\rho_{SM}}{\gamma_S\otimes\rho_M} - \ln(\epsilon)
      - \ln`\Big( 1 - \frac1\epsilon \ee^{ -\Dhyp[1-\epsilon]{\rho_{SM}}{\gamma_S\otimes\rho_M}}  )
      \nonumber\\
    &\leqslant
      -\Dhyp[1-\epsilon]{\rho_{SM}}{\gamma_S\otimes\rho_M} - \ln(\epsilon)
      -\ln(1/2)
      \ ,
  \end{align}
  using in the last inequality our assumption that
  $\epsilon \exp`\big{ \Dhyp[1-\epsilon]{\rho_{SM}}{\gamma_S\otimes\rho_M}}
  \geqslant 2$.  We have
  proven~\eqref{eq:single-shot-conditional-erasure-using-convex-split--work-cost},
  thereby completing the proof of
  \cref{x:single-shot-conditional-erasure-using-convex-split}.
\end{proof}

\subsection{Single-shot protocol for implementing a
  time-covariant processes using thermal operations}
\label{sec:one-shot-timecovariant-process-H-covariant}

We can translate our protocol for conditional erasure into a protocol for
implementing more general processes. (A similar approach was taken in
Ref.~\cite{Faist2015NatComm} in the case of trivial Hamiltonians).  We obtain
the following result.

\begin{corollary}[Single-shot protocol for a time-covariant process and input state]
  \label{x:single-shot-time-covariant-process-using-convex-split}
  Let $X$ be a quantum system with a Hamiltonian $H_X$, and let
  $\beta\geqslant 0$.  Let $\mathcal{E}_{X\to X}$ be a completely positive,
  trace-preserving map satisfying
  \begin{align}
    \mathcal{E}_{X\to X}(\ee^{-iH_Xt}\,(\cdot)\,\ee^{iH_Xt})
    = \ee^{-iH_Xt}\,\mathcal{E}_{X\to X}(\cdot)\,\ee^{iH_Xt}
    \qquad\text{for all \(t\in\mathbb{R}\).}
  \end{align}
  Let $\sigma_X$ be any quantum state that satisfies $[\sigma_X, H_X] = 0$ and
  let $\epsilon>0$.  There exists an information battery $W$, a thermal
  operation $\Phi_{XW}$ and battery states $\tau_W, \tau_W'$ such that:
  \begin{itemize}
  \item The effective work process $\mathcal{T}_{X\to X}$ associated with
    $\Phi_{XW}$ and $(\tau_W, \tau_W')$ satisfies
    \begin{align}
      F`\big( \mathcal{T}_{X\to X}`\big(\proj\sigma_{XR}) ,
      \mathcal{E}_{X\to X}`\big(\proj\sigma_{XR}) )
      \geqslant 1 - 6\epsilon\ ,
      \label{eq:single-shot-time-covariant-process-using-convex-split--accurate}
    \end{align}
    where $\ket\sigma_{XR}$ is a purification of $\sigma_X$;
  \item The work cost (in units of pure nats) satisfies
    \begin{align}
      w(\tau_W) - w(\tau_W')
      \leqslant
      -\Dhyp[1-\epsilon]{\rho_{EX}}{\gamma_E\otimes\rho_X}
      + \ln(\tr\ee^{-\beta H_X})
      + \ln(2/\epsilon)\ ,
      \label{eq:single-shot-time-covariant-process-using-convex-split--work-cost}
    \end{align}
    where $E$ is an additional system with Hamiltonian $H_E$ and where
    $\rho_{EX} = V`\big(\proj0_E\otimes\sigma_X) V^\dagger$ is the output of a
    time-covariant Stinespring unitary $V_{XE\to XE}$ of $\mathcal{E}_{X\to X}$
    with $\ket{0}_E$ an eigenstate of $H_E$ with eigenvalue $0$; i.e., the
    operator $V_{X\to E}$ satisfies
    $\mathcal{E}_{X\to X}(\cdot) = \tr_E`\big[
    V\,`\big(\proj{0}_E\otimes(\cdot))\,V^\dagger ]$.  We note that
    $\rho_X = \mathcal{E}_{X\to X}(\sigma_X)$.
  \end{itemize}
\end{corollary}

Observe that for any time-covariant process $\mathcal{E}_{X\to X}$, the
existence of $V_{XE\to XE}$ with the desired properties is guaranteed by
\cref{x:covariant-process-H-env}.

\begin{proof}[**x:single-shot-time-covariant-process-using-convex-split]
  We start with the system $X$ and a reference system $R$ in the pure state
  $\ket{\sigma_{XR}}$.  The protocol is assembled using the following steps:
  \begin{itemize}
  \item Bring in an ancilla $E$ which can be used as the Stinespring environment
    system of $\mathcal{E}_{X\to X}$.  The system $E$ is initialized in its
    thermal state $\gamma_E$.  This is a free operation;
  \item Reset the ancilla $E$ to the state $\ket{0}_E$, using
    \cref{thm:thermal-to-zero-thermomaj-workcost}, at a work cost arbitrarily
    close to $\ln\tr[\ee^{-\beta H_X}]$;
  \item Apply the energy-conserving unitary $V_{XE\to XE}$ at no work cost.  The
    resulting state on $XE$ is $\rho_{XE}$;
  \item Invoke \cref{x:single-shot-conditional-erasure-using-convex-split} to
    transform the state $\rho_{XE}$ into a state $\rho'_{XE}$ that is close to
    $\gamma_E\otimes\rho_X$, with
    \begin{align}
      F`\big(\rho'_{XE}, \gamma_E\otimes\rho_X)
      \geq 1 - 6\epsilon\ ,
    \end{align}
    and at a work cost $w$ satisfing
    $w \leqslant -\Dhyp[1-\epsilon]{\rho_{EM}}{\gamma_E\otimes\rho_M} +
    \ln(2/\epsilon)$;
  \item Discard the system $E$, at no work cost.
  \end{itemize}
  
  The
  condition~\eqref{eq:single-shot-time-covariant-process-using-convex-split--accurate}
  is satisfied by construction.  Furthermore, collecting the contributions to
  the work cost
  yields~\eqref{eq:single-shot-time-covariant-process-using-convex-split--work-cost}.
\end{proof}

\section{Discussion}
\label{sec:discussion}

Our results fits in the line of research extending results in thermodynamics
from state-to-state transformations to quantum processes. Implementations of
quantum processes are difficult to construct because they need to reproduce the
correct correlations between the output and the reference system, and not only
produce the correct output state. Here, we have seen that it is nevertheless
possible to implement any quantum process at an optimal work cost: Any
implementation that would use less work would violate the second law of
thermodynamics on a macroscopic scale. As a special case this also provides an
operational interpretation of the minimal entropy gain of a
channel~\cite{Alicki2004arXiv_isotropic,Devetak2006CMP_multiplicativity,Holevo2011ISIT_entropygain,Holevo2010DM_infinitedim,Holevo2011TMP_CJ,BookHolevo2012_QuSystemsChannelsInformation,Buscemi2016PRA_reversibility,Gour2021PRR_entropychannel}.

Our three constructions of optimal implementations of processes are valid in
different settings, and it remains unclear if they can be unified in a single
protocol that presents the advantages of all three constructions.  Namely, is it
possible to use a physically well-justified framework, e.g.\@ thermal
operations, to universally implement any i.i.d.\@ process? We expect this to be
possible only if an arbitrary amount of coherence is allowed, in analogy with
the entanglement embezzling state required in the reverse Shannon
theorem~\cite{Bennett2014_reverse,Berta2011_reverse}.

Finally, the notion of quantum typicality that we have introduced in
\cref{defn:universal-relative-conditional-typical-smoother} and
\cref{thm:universal-relative-conditional-typical-smoother} might be interesting
in its own right. We anticipate that similar considerations might provide
pathways to smooth other information-theoretic
quantities~\cite{Fang2019IEEETIT_simulation,Anshu2020ITIT_partially,Gour19} and
to study the joint typicality
conjecture~\cite{PhDDutil2011_multiparty,Noetzel2012arXiv_two,Drescher2013ISIT_simultaneous,Sen2018arXiv_joint,Anshu18-2}.

\begin{acknowledgments}
  The authors thank \'Alvaro Alhambra, David Ding, Patrick Hayden, Rahul Jain,
  David Jennings, Mart\'\i{} Perarnau-Llobet, Mark Wilde, and Andreas Winter for
  discussions. PhF acknowledges support from the Swiss National Science
  Foundation (SNSF) through the Early PostDoc.Mobility Fellowship No.\@
  {P2EZP2\_165239} hosted by the Institute for Quantum Information and Matter
  (IQIM) at Caltech, from the IQIM which is a National Science Foundation (NSF)
  Physics Frontiers Center (NSF Grant {PHY}-{1733907}), from the Department of
  Energy Award {DE}-{SC0018407}, from the Swiss National Science Foundation
  (SNSF) through the NCCR QSIT and through Project No.~{200020\_16584}, and from
  the Deutsche Forschungsgemeinschaft (DFG) Research Unit {FOR}~{2724}. FB is
  supported by the NSF. This work was completed prior to MB and FB joining the AWS Center for Quantum Computing.
\end{acknowledgments}

\appendix
\section*{Appendix}

\section{Missing proofs}
\label{app:missing}
\label{appx:first}%

\begin{proof}[*lemma:Schur-Weyl-bipartite-trAn-PiLambdaPiLambdaprime]
  A useful expression for $\Pi^\lambda_{A^nB^n}$ may be obtained
  following~\cite[Section V]{Haah2017IEEETIT_sampleoptimal}
  \begin{align}
    \Pi^{\lambda}_{A^nB^n}
    &= \frac{\dim(\mathcal{Q}_\lambda)}{s_\lambda(\diag(\lambda/n))}
      \int dU_{AB}\;  \Pi^\lambda_{A^nB^n}
      `*( U_{AB} \, \diag(\lambda/n)_{AB} \, U_{AB}^{\dagger} )^{\otimes n}
      \, \Pi^\lambda_{A^nB^n}
      \nonumber\\
    &\leqslant
      \poly(n)\,\ee^{n\bar{H}(\lambda)} \,
      \int dU_{AB}\; `*( U_{AB} \, \diag(\lambda/n)_{AB} \, U_{AB}^{\dagger} )^{\otimes n}\ ,
      \label{eq:Schur-Weyl-PiLambda-expr-as-iid-states}
  \end{align}
  recalling that $\Pi^\lambda_{A^nB^n}$ commutes with any i.i.d.\@{} state, with
  $s_\lambda(X) = \tr[q_\lambda(X)]$ and using bounds on
  $\dim(\mathcal{Q}_\lambda)$ and $s_\lambda(\diag(\lambda/n))$ derived in
  Ref.~\cite{Haah2017IEEETIT_sampleoptimal}. Here, $dU_{AB}$ denotes the Haar measure
  over all unitaries acting on $\Hs_{AB}$, normalized such that $\int dU_{AB} = 1$. We then have
  \begin{align}
    \tr_{A^n}`*[\Pi^\lambda_{A^nB^n}]
    \leqslant \poly(n)\,\ee^{n\bar{H}(\lambda)}
    \int dU_{AB}\,
    \tr_{A^n}`*[`*(U_{AB}\,\diag(\lambda/n)_{AB}\,U_{AB}^\dagger)^{\otimes n} ]\ .
  \end{align}
  Observe that for any state $\omega_B$, we have
  \begin{multline}
    \norm[\big]{\Pi^{\lambda'}_{B^n}\,\omega_B^{\otimes n}\,\Pi^{\lambda'}_{B^n}}_\infty
    =
    \norm[\big]{ \mapinlambda{q_{\lambda'}(\omega_B)\otimes
    \Ident_{\mathcal{P}_{\lambda'}}}[\lambda'] }_\infty
    =
    \norm{ q_{\lambda'}(\omega_B) }_\infty
    \leqslant \tr`*[q_{\lambda'}(\omega_B)]
    \\
    \leqslant \poly(n)\,\ee^{-n\bar{H}(\lambda')}
  \end{multline}
  as derived e.g.\@ in~\cite[Eq.~(9)]{Haah2017IEEETIT_sampleoptimal}, and thus
  for any state $\omega_B$,
  \begin{align}
    \Pi^{\lambda'}_{B^n}\,\omega_B^{\otimes n}\,\Pi^{\lambda'}_{B^n}
    \leqslant \poly(n)\,\ee^{-n\bar{H}(\lambda')}\,  \Pi^{\lambda'}_{B^n} \ .
  \end{align}
  Hence, we get
  \begin{align}
    \hspace*{3em}
    &\hspace*{-3em}
    \Pi^{\lambda'}_{B^n}\,\tr_{A^n}`*[\Pi^\lambda_{A^nB^n}]\,\Pi^{\lambda'}_{B^n}
    \nonumber\\
    &\leqslant \poly(n)\,\ee^{n\bar{H}(\lambda)}
    \int dU_{AB}\;
    \Pi^{\lambda'}_{B^n} \,
    `*(\tr_A`*[U_{AB}\,\diag(\lambda/n)_{AB}\,U_{AB}^\dagger])^{\otimes n}
    \Pi^{\lambda'}_{B^n}
    \nonumber\\
    &\leqslant \poly(n)\,\ee^{n\bar{H}(\lambda)}
    \int dU_{AB}\;
      \poly(n)\,\ee^{-n\bar{H}(\lambda')}\, \Pi_{B^n}^{\lambda'}
    \nonumber\\
    &= \poly(n)\,\ee^{n(\bar{H}(\lambda) - \bar{H}(\lambda'))}\, \Pi_{B^n}^{\lambda'}\ ,
  \end{align}
  as required.
\end{proof}

\begin{proof}[*prop:entropy-measurement-POVM-n-systems]
  The Fannes-Audenaert continuity bound~\cite{Fannes1973CMP_continuity,Audenaert2007JPA_sharp} of the entropy
  states that for any $\delta'>0$ there exists $\xi(\delta')>0$ such that for
  any quantum states $\rho,\sigma$ with $D(\rho,\sigma)\leqslant\delta'$ we have
  \begin{align}
    \abs{H(\rho) - H(\sigma)} \leqslant \xi(\delta')\ ,
  \end{align}
  and furthermore $\xi(\delta')$ is monotonically strictly decreasing and
  $\xi(\delta')\to0$ if $\delta'\to0$. Now, let $\delta>0$, let $\xi^{-1}$ be
  the inverse function of $\xi$, and let $\delta'=\xi^{-1}(\delta)$.  Consider
  the set of Young diagrams
  $\Lambda_{\delta'} = `{ \lambda\in\Young{d_A}{n} : D(\diag(\lambda/n),\rho)
    \leqslant \delta' }$.  For all $\lambda\in\Lambda_{\delta'}$, we have that
  $\abs{H(\rho) - \bar{H}(\lambda)} \leqslant \delta$ thanks to the
  Fannes-Audenaert inequality. Then, we have
  \begin{align}
    \tr`*[ `*(\sum_{\lambda:~\bar{H}(\lambda) \in [H(\rho)\pm\delta] }
    \Pi_{A^n}^\lambda ) \rho_A^{\otimes n} ]
    \geqslant
    \tr`*[ `*(\sum_{\lambda \in \Lambda_{\delta'}} \Pi_{A^n}^\lambda )
    \rho_A^{\otimes n} ]
    \label{eq:kjfiduygou9qiowfjdoask}
  \end{align}
  because all terms in the sum in the right hand side are included in the sum on
  the left hand side.  We may now invoke~\cite[Eq.~(6.23)]{PhDHarrow2005} to see
  that
  \begin{align}
    \text{\eqref{eq:kjfiduygou9qiowfjdoask}}
    \geqslant 1 - \poly(n)\exp`*{- n\eta}\ ,
  \end{align}
  where $\eta=\delta'^2/2$.
\end{proof}

\begin{proof}[*prop:energy-measurement-POVM-n-systems]
  The fact that there are only $\poly(n)$ elements follows because there are
  only so many types.  Property~(ii) holds by definition.  Property~(iv) holds
  because $e^{-n(k\pm\delta)}$ is the minimum / maximum eigenvalue of
  $\Gamma_A^{\otimes n}$ in the subspace spanned by
  $R^{\approx_\delta h}_{A^n}$.  Finally, we need to show Property~(iii): This
  follows from a large deviation analysis.  More precisely, let $Z_j$ for
  $j=1,\ldots,n$ be random variables where $Z_j$ represents the measurement
  outcome of $H_A$ on the $j$-th system of the i.i.d.\@{} state
  $\rho_A^{\otimes n}$.  By Hoeffding's inequality, we have that
  \begin{align}
    \Pr`*[ \abs*{(1/n) \sum Z_j - \tr[\rho_A H_A]} > \delta ]
    &\leqslant 2\exp`*(-\frac{2n\delta^2}{\Delta H_A^2})
    \leqslant 2\exp`*(-\frac{n\delta^2}{2\,\norm{H_A}_\infty^2})\ ,
  \end{align}
  where $\Delta H_A$ is the difference between the maximum and minimum
  eigenvalue of $H_A$, and $\Delta H_A \leqslant 2\norm{H_A}_\infty$. Thus, the
  event consisting of the outcomes $k$ satisfying
  $\abs{k-\tr[\rho_A H_A]} \leqslant \delta$ happens with probability at least
  $1-2\ee^{-n\eta}$,
  proving~\eqref{eq:energy-measurement-POVM-n-systems-exp-success-on-iid-states}.
\end{proof}

\begin{proof}[*prop:iid-process-diamond-norm-with-any-W]
  We use the post-selection technique (Theorem \ref{x:post-selection-technique})
  to bound the diamond norm distance between $\mathcal{T}_{X^n\to X'^n}$ and
  $\mathcal{E}_{X\to X'}^{\otimes n}$.  Let $\ket\zeta_{X^n\bar{R}^n R'}$ be the
  purification of the de Finetti state given
  by~\eqref{eq:de-Finetti-state-purified-with-poly-iid-states}.  Calculate
  \begin{align}
    \hspace*{3em}
    &\hspace*{-3em}
      \Re`*{ \bra\zeta_{X^n\bar{R}^n R'} (V_{X\to EX'}^{\otimes n})^\dagger
      W_{X^n\to E^nX'^n} \ket\zeta_{X^n\bar{R}^n R'} }
    \nonumber\\
    &= \sum p_i\,
      \Re`*{ \bra{\phi_i}_{X\bar{R}}^{\otimes n}
      \, (V_{X\to EX'}^{\otimes n})^\dagger
      \, W_{X^n\to E^nX'^n} \ket{\phi_i}_{X\bar{R}}^{\otimes n} }
      \nonumber\\
    &\geqslant 1 - \poly(n)\exp(-n\eta)
  \end{align}
  which implies, recalling that
  $F(\ket\psi,\ket\phi) = \abs{\braket\psi\phi} \geqslant \Re`{\braket\psi\phi}$
  and that $(1-x)^2 \geqslant 1 - 2x$,
  \begin{align}
    F^2`*( V_{X\to EX'}^{\otimes n}\, \ket\zeta_{X^n\bar{R}^n R'} \, , \,
    W_{X^n\to E^nX'^n} \ket\zeta_{X^n\bar{R}^n R'} )
    \geqslant 1 - \poly(n) \exp(-n\eta)
  \end{align}
  and hence
  \begin{align}
    P`*( V_{X\to EX'}^{\otimes n}\, \ket\zeta_{X^n\bar{R}^n R'} \, , \,
    W_{X^n\to E^nX'^n} \ket\zeta_{X^n\bar{R}^n R'} )
    \leqslant \poly(n) \exp(-n\eta/2)\ .
  \end{align}
  Recalling the relations between the trace distance and the purified distance,
  and noting that these distance measures cannot increase under the partial
  trace, we obtain
  \begin{multline}
    D`\big( \mathcal{T}(\zeta_{X^n\bar{R}^n R'}) ,
    \mathcal{E}^{\otimes n}(\zeta_{X^n\bar{R}^n R'}) )
    \leqslant
    P`\big( \mathcal{T}(\zeta_{X^n\bar{R}^n R'}) \;,\;
    \mathcal{E}^{\otimes n}(\zeta_{X^n\bar{R}^n R'}) )
    \\
    \leqslant
    P`\big( W_{X^n\to E^nX'^n} \ket\zeta_{X^n\bar{R}^n R'}   \;,\;
    V_{X\to EX'}^{\otimes n}\, \ket\zeta_{X^n\bar{R}^n R'} )
    \leqslant \poly(n) \exp(-n\eta/2)\ .
  \end{multline}
  The post-selection technique then asserts that
  \begin{align}
    \frac12\norm{ \mathcal{T} - \mathcal{E}^{\otimes n} }_\diamond
    \leqslant \poly(n)\,\exp(-n\eta/2)
  \end{align}
  as claimed.
\end{proof}

\begin{proof}[*x:covariant-process-H-env]
  Let $V'_{X\to XE}$ be any Stinespring dilation isometry of
  $\mathcal{E}_{X\to X}$, such that
  $\mathcal{E}_{X\to X}(\cdot) = \tr_E`*[ V'_{X\to XE}\,(\cdot)\,V'^\dagger
  ]$. For the input state $\ket\Phi_{X:R_X}$, consider the output state
  $\ket{\varphi}_{XER_X}$ corresponding to first time-evolving by some time $t$,
  and then applying $V'$
  \begin{align}
    \ket{\varphi}_{XER_X}
    = V'\,\ee^{-i H_X t}\,\ket\Phi_{X:R_X}
    = \ee^{-i V' H_X V'^\dag t}\, V'\, \ket\Phi_{X:R_X}\ .
  \end{align}
  Now, let us define
  $\ket{\varphi'}_{XER_X} = \ee^{-i H_X t} \, V'\,\ket\Phi_{X:R_X}$. By the
  covariance property of $\mathcal{E}_{X\to X}$ both $\ket{\varphi}$ and
  $\ket{\varphi'}$ have the same reduced state on $X R_X$. Hence, they are
  related by some unitary $W_E^{(t)}$ on the system $E$ which in general depends
  on $t$
  \begin{align}
    \ket{\varphi}_{XER_X} = W^{(t)}_E \, \ket{\varphi'}_{XER_X}\ .
    \label{eq:8efuihdon}
  \end{align}
  We have
  \begin{align}
    \tr_{X}`*[ V' \ee^{-iH_X t} \Phi_{X:R_X} \ee^{i H_X t} V'^\dagger ]
    = W^{(t)}_{E} \tr_X`\big[ V' \Phi_{X:R_X} V'^\dagger ] W^{(t)\,\dagger}_{E}
  \end{align}
  so $W^{(t)}_E$ must define a representation of time evolution, at least on the
  support of the operator $\tr_X`\big[ V' \Phi_{X:R_X} V'^\dagger ]$. Hence, we may write
  $W^{(t)}_E = \ee^{-i H_E t}$ for some Hamiltonian $H_E$, and
  from~\eqref{eq:8efuihdon}, we have for all $t$
  \begin{align}
    V'_{X\to XE} \, \ee^{-i H_{X} t}  = \ee^{-i (H_X + H_E) t} \, V'_{X\to XE}\ .
  \end{align}
  Expanding for infinitesimal $t$ we obtain
  \begin{align}
    V'_{X\to XE} \, H_X = `*( H_X + H_E )  \, V'_{X\to XE}\ .
    \label{eq:ojifudsno}
  \end{align}
  Let $\ket0_E$ be an eigenvector of $H_E$ corresponding to the eigenvalue zero;
  if $H_E$ does not contain an eigenvector with eigenvalue equal to zero, we may
  trivially add a dimension to the system $E$ to accommodate this vector. Then,
  the operator $V'_{X\to XE} \bra0_E$ maps each state of a subset of energy levels of $XE$ to a corresponding energy
  level of same energy on $XE$; it may thus be completed to a fully
  energy-preserving unitary $V_{XE\to XE}$. More precisely, let $\ket{j}_X$ be
  a complete set of eigenvectors of $H_X$ with energies $h_j$. Then
  $\ket{\psi_j'} = V'_{X\to XE}\ket{j}_X$ is an eigenvector of $H_X+H_E$ of
  energy $h_j$ thanks to~\eqref{eq:ojifudsno}. We have two orthonormal sets
  $`\big{ \ket0_E\otimes\ket{j}_X }$ and $`\big{ \ket{\psi_j'}_X }$ in which the
  $j$-th vector of each set has the same energy; we can thus complete these sets
  into two bases $`*{ \ket{\chi_i}_{XE} }$, $`*{ \ket{\chi'_i}_{XE} }$ of
  eigenvectors of $H_X+H_E$, where the $i$-th element of either basis has
  exactly the same energy. This defines a unitary
  $V_{XE\to XE} = \sum_i \ket{\chi'_i}_{XE}\bra{\chi_i}_{XE}$ that is an
  extension of $V'_{X\to XE} \bra0_E$, and that satisfies all the conditions of
  the claim.
\end{proof}

\section{Technical lemmas}\label{appx:technical-lemmas}

\begin{lemma}[Pinching-like operator inequality]%
  \label{lemma:Ai-coherent-can-ignore}
  Let $`{E^i}_{i=1}^M$ be a collection of $M$ operators and $T\geqslant
  0$. Then, we have
  \begin{align}
    `*(\sum E^i)\, T\, `*(\sum E^{j\,\dagger})
    \leqslant M \sum E^i\,T\,E^{i\,\dagger}\ .
  \end{align}
\end{lemma}

\begin{proof}[**lemma:Ai-coherent-can-ignore]
  Call our system $S$ and consider an additional register $C$ of dimension
  $\abs{C} = M$, and let $\ket{\chi}_C = M^{-1/2}\sum_{k=1}^M \ket{k}_C$. Then, we have
  \begin{align}
    `*(\sum E_S^i)\, T_S\, `*(\sum E_S^{j\,\dagger})
    &= \tr_C`*[ `*(\sum E_S^i\otimes\ket{i}_C)\, T_S\,
      `*(\sum E_S^{j\,\dagger}\otimes\bra{j}_C)
      \, `*(\Ident_S\otimes(M\,\proj\chi_C) ) ]
      \nonumber\\
    &\leqslant  M
      \tr_C`*[ `*(\sum E_S^i\otimes\ket{i}_C)\, T_S\,
      `*(\sum E_S^{j\,\dagger}\otimes\bra{j}_C)
      \, `*(\Ident_S\otimes\Ident_C) ]
      \nonumber\\
    &=  M
      \sum E_S^i \, T_S\, E_S^{i\,\dagger}\ ,
  \end{align}
  using $\proj\chi_C\leqslant \Ident_C$.
\end{proof}

\begin{lemma}[Gentle measurement]%
  \label{x:gentle-measurement-lemma}
  Let $\rho$ be a sub-normalized quantum state and
  $0\leqslant Q\leqslant \Ident$. For $\tr[Q\rho]\geqslant 1 - \delta$ we then
  have
  \begin{align}
    P(\rho, Q^{1/2}\rho Q^{1/2}) \leqslant \sqrt{2\delta}\ .
  \end{align}
\end{lemma}

This is a cruder statement than that of, e.g.,
\cite[Lemma~7]{Berta2010_uncertainty}, allowing for a more straightforward
proof.

\begin{proof}[**x:gentle-measurement-lemma]
  We have
  \begin{align}
    \bar{F}(\rho, Q^{1/2}\rho Q^{1/2}) \geqslant
    F(\rho, Q^{1/2}\rho Q^{1/2})
    &= \tr\left[\sqrt{\rho^{1/2}(Q^{1/2}\rho Q^{1/2})\rho^{1/2}}\right]\nonumber\\
    &= \tr\left[Q^{1/2}\rho\right]\geqslant \tr[Q\rho] \geqslant 1 - \delta\ .
  \end{align}
  Then, we get
  $
    P(\rho, Q^{1/2}\rho Q^{1/2})
    \leqslant \sqrt{ 1 - (1 - \delta)^2 }
    \leqslant \sqrt{2\delta}
    $.
\end{proof}

\begin{proposition}[Controlled-unitary using a POVM]%
  \label{x:POVM-controlled-unitary}
  Let $`{ Q^j }$ be a set of positive semi-definite operators on a system $X$
  satisfying $\sum Q^j \leqslant \Ident$, $`{ U^j }$ be a collection of
  unitaries on a system $Y$, and
  \begin{align}
    W_{XY} = \sum_j Q_X^j \otimes U_Y^j\ .
  \end{align}
  Then, we have $W^\dagger W \leqslant \Ident$.
\end{proposition}

\begin{proof}[**x:POVM-controlled-unitary]
  Using an additional register $K$, define
  \begin{align}
    V_{X\to XK} = \sum \sqrt{Q^j} \otimes \ket{j}_K\ .
  \end{align}
  Then, we have $V^\dagger V = \sum Q^j \leqslant \Ident$.  Clearly,
  $VV^\dagger \leqslant \Ident_{XK}$ because $VV^\dagger$ and $V^\dagger V$ have
  the same non-zero eigenvalues. Now, let
  \begin{align}
    W = V^\dagger `*( \sum \Ident_X\otimes U^j_Y \otimes \proj{j}_K) \, V\ .
  \end{align}
  Because the middle term in parentheses is unitary, we manifestly have
  $W^\dagger W \leqslant \Ident$.
\end{proof}

\section{Dilation of energy-conserving operators to unitaries}
\label{appx:dilation-energy-conserving-subunitaries}

This appendix collects a few technical lemmas on constructing an
energy-conserving unitary that extends a given operator of norm less than one.

\begin{proposition}
  \label{x:complete-subunitary-to-unitary}
  Let $W_X$ be an operator on a system $X$, such that
  $W^\dagger W \leqslant \Ident$. Then, there exists a unitary operator $U_{XQ}$
  acting on $X$ and a qubit $Q$ such that for any $\ket\psi_X$,
  \begin{align}
    \bra{0}_Q\,U_{XQ}\,(\ket\psi_X\otimes\ket{0}_Q)
    = W_{X} \ket\psi_X\ .
  \end{align}
  That is, any operator $W$ with $\norm{W}_\infty\leqslant1$ can be dilated to a
  unitary, with a post-selection on the output.
\end{proposition}

\begin{proof}[**x:complete-subunitary-to-unitary]
  Setting
  $V_{X\to XQ} = W\otimes\ket{0}_Q + \sqrt{\Ident - W^\dagger{}W} \otimes
  \ket{1}_Q$, we see that
  $V^\dagger V = W^\dagger W + \Ident - W^\dagger W = \Ident_X$, and hence
  $V_{X\to XQ}$ is an isometry. We can complete this isometry to a unitary
  $U_{XQ}$ that acts as $V$ on the support of $\Ident_X\otimes\proj{0}_Q$ and
  that maps the the support of $\Ident_X\otimes\proj{1}_Q$ onto the
  complementary space to the image of $V$. It then follows that for any
  $\ket\psi_X$, we have
  $U_{XQ}\,(\ket\psi_X\otimes\ket0_Q) = V_{X\to XQ}\, \ket\psi_X =
  (W_X\ket\psi_X)\otimes\ket0_Q + (\ldots)\otimes\ket1_Q$, and the claim
  follows.
\end{proof}

\begin{proposition}
  \label{x:complete-energy-preserving-subunitary-to-unitary}
  Let $X$ be a quantum system with Hamiltonian $H_X$ and $W_X$ be an operator
  with $W^\dagger W \leqslant \Ident$ as well as $[W_X, H_X] = 0$. Then, there
  exists a unitary operator $U_{XQ}$ acting on $X$ and a qubit $Q$ with $H_Q=0$,
  that satisfies $[U_{XQ}, H_X] = 0$ such that
  \begin{align}
    \bra{0}_Q\,U_{XQ}\,\ket{0}_Q = W_{X}\ .
  \end{align}
  That is, any energy-preserving operator $W$ with $\norm{W}_\infty\leqslant1$
  can be dilated to an energy-preserving unitary on an ancilla with a
  post-selection on the output.
\end{proposition}

\begin{proof}[**x:complete-energy-preserving-subunitary-to-unitary]
  First we calculate
  $[W^\dagger W, H_X] = W^\dagger[W, H_X] + [W^\dagger, H_X]\, W = 0 - [W,
  H_X]^\dagger W = 0$.  This implies that
  $[\sqrt{\Ident - W^\dagger W}, H_X] = 0$, as $W^\dagger W$ and
  $\sqrt{\Ident - W^\dagger W}$ have the same eigenspaces.
  We define
  \begin{align}
    V_{X\to XQ} =
    W\otimes\ket{0}_Q + \sqrt{\Ident - W^\dagger{}W} \otimes \ket{1}_Q\ .
  \end{align}
  The operator $V_{X\to XQ}$ is an isometry, because
  $V^\dagger V = W^\dagger W + \Ident - W^\dagger W = \Ident_X$. Furthermore, we
  have
  \begin{align}
  V_{X\to XQ} \, H_X &= (W_X H_X)\otimes\ket0 + (\sqrt{\Ident - W^\dagger W}\,
  H_X)\otimes\ket1\\
  &= (H_X W_X)\otimes\ket0 + (H_X \sqrt{\Ident - W^\dagger W})
  \otimes\ket1 = H_X \, V_{X\to XQ}
  \end{align}
  and thus we find $[V_{X\to XQ}, H_X] = 0$.  Let $`\big{ \ket{j}_X }$ be an
  eigenbasis of $H_X$, and let $\ket{\psi'_j}_{XQ} = V_{X\to XQ} \ket{j}_X$,
  noting that both $\ket{j}_X$ and $\ket{\psi'_j}_{XQ}$ have the same energy.
  The two collections of vectors $`\big{ \ket{j}_X\otimes\ket0_Q }$ and
  $`\big{ \ket{\psi'_j}_{XQ} }$ can thus be completed into two bases
  $`\big{ \ket{\chi_i}_{XQ} }$ and $`\big{ \ket{\chi'_i}_{XQ} }$ of eigenvectors
  of $H_X + H_Q$ where the $i$-th element of both bases have the same energy.
  Define finally $U_{XQ} = \sum_i \ketbra{\chi'_i}{\chi_i}_{XQ}$, noting that by
  construction $U_{XQ}\ket0_Q = V_{X\to XQ}$ and $[U_{XQ}, H_X] = 0$.
\end{proof}

\begin{proposition}
  \label{x:complete-energy-preserving-subunitary-to-unitary-noflag}
  Let $X$ be a quantum system with Hamiltonian $H_X$ and, $W_X$ an operator
  satisfying $W^\dagger W \leqslant \Ident$ such that $[W_X, H_X]=0$. Then, for
  any $\epsilon>0$, there exists a unitary operator $U_X$ satisfying
  $[U_X,H_X]=0$ such that for any states $\ket{\psi}_X, \ket{\psi'}_{X}$
  satisfying $\Re`*{ \matrixel{\psi'}{W}{\psi} } \geqslant 1 - \epsilon$, we
  have $\Re`*{ \matrixel{\psi'}{U}{\psi} } \geqslant 1 - 6\epsilon^{1/4}$.
\end{proposition}

\begin{proof}[*x:complete-energy-preserving-subunitary-to-unitary-noflag]
  Let $F = W^\dagger W \leqslant \Ident$ noting that $F^\dagger = F$.  For some
  $\nu$ with $0<\nu<1$ to be determined later, let $P$ be the projector onto the
  eigenspaces of $F$ corresponding to eigenvalues greater or equal to $\nu$. Define $V = WF^{-1/2} P$.  The operator $V$ is a partial isometry, meaning
  that its singular values are all equal to one or to zero, because
  $V^\dagger V = P F^{-1/2} W^\dagger W F^{-1/2} P = P F^{-1/2} F F^{-1/2} P =
  P$, since $P$ lies within the support of $F$. Observe that
  $[F, H] = [W^\dagger W, H] = W^\dagger [W, H] - [W, H]^\dagger W = 0$ and
  hence $[F^{-1/2}, H] = 0$ and $[P,H]=0$.  Hence, we have
  $[V, H] = [W F^{-1/2} P, H] = 0$.  So we may complete the partial isometry $V$
  into a full unitary $U$ that also commutes with $H$ by acting as the identity
  on the remaining elements of the eigenbasis of $H$ in which $V$ is diagonal.
  We may thus write $U = V + X = WF^{-1/2}P + X$ for some operator $X$
  satisfying $X P = 0$ and $[X, H] = 0$.

  Now, let $\ket\psi$, $\ket{\psi'}$ such that
  $\Re`*{ \matrixel{\psi'}{W}{\psi} } \geqslant 1 - \epsilon$, and write
  \begin{align}
    \Re`*{ \matrixel{\psi'}{U}{\psi} }
    &= \Re`*{ \matrixel{\psi'}{(U - W)}{\psi} } + \Re`*{ \matrixel{\psi'}{W}{\psi} }
    \nonumber\\
    &\geqslant
      \Re`*{ \matrixel{\psi'}{(U - W)}{\psi} }
      + 1 - \epsilon \ .
    \label{eq:fidnsofda3e}
  \end{align}
  We have
  $\dmatrixel{\psi}{P} \geqslant \dmatrixel{\psi}{PFP} = \dmatrixel{\psi}{F} -
  \dmatrixel{\psi}{(\Ident-P)F} \geqslant \dmatrixel{\psi}{F} - \nu$, recalling
  that $\Ident-P$ projects onto the eigenspaces of $F$ whose eigenvalues are
  less than $\nu$. Then,
  $
    \dmatrixel{\psi}{F} = 
    \dmatrixel{\psi}{W^\dagger W} \geqslant
    \matrixel{\psi}{W^\dagger}{\psi'}\matrixel{\psi'}{W}{\psi} =
    \abs{\matrixel{\psi'}{W}{\psi}}^2 \geqslant (1-\epsilon)^2 \geqslant 1 -
    2\epsilon
  $,
  and hence
  $\norm[\big]{(\Ident - P)\ket\psi}^2 = \dmatrixel{\psi}{(\Ident-P)}
  \leqslant 1 - (\dmatrixel{\psi}{F} - \nu)
  \leqslant 1 - (1-2\epsilon) + \nu = 2\epsilon + \nu$.
  Hence, we get
  \begin{align}
    \Re`\big{ \matrixel{\psi'}{(U - W)}{\psi} }
    &= \Re`\big{ \matrixel{\psi'}{(U - W) P}{\psi} }
    + \Re`\big{ \matrixel{\psi'}{(U - W)(\Ident- P)}{\psi} }
      \nonumber\\
    &\geqslant \Re`\big{ \matrixel{\psi'}{(U - W) P}{\psi} } - 2\sqrt{2\epsilon+\nu}
      \label{eq:fiodbsnfois}
  \end{align}
  since by Cauchy-Schwarz
  $\abs[\big]{ \matrixel{\psi'}{(U - W)(\Ident-P)}{\psi} } \leqslant
  \norm[\big]{(U - W)^\dagger\ket{\psi'}}\, \norm[\big]{(\Ident-P)\ket\psi}$,
  where $\norm{(U-W)^\dagger\ket{\psi'}}\leqslant 2$. In order to continue, let $\ket{\chi} = W^\dagger \ket{\psi'} - \ket{\psi}$,
  and calculate
  $
    \braket\chi\chi = \dmatrixel{\psi'}{WW^\dagger} + \braket{\psi}{\psi}
    - 2\Re`*{ \matrixel{\psi'}{W}{\psi} } \leqslant 2 - 2(1-\epsilon) = 2\epsilon
  $,
  and hence we deduce that
  $\norm[\big]{\ket\chi} = \norm[\big]{ W^\dagger \ket{\psi'} - \ket{\psi} }
  \leqslant \sqrt{2\epsilon}$. Then, with $\bra{\psi'} W = \bra{\psi} + \bra{\chi}$ we have
  \begin{align}
    \hspace*{3em}
    &\hspace*{-3em}
      \Re`\big{ \matrixel{\psi'}{(U - W) P}{\psi} }
      = \Re`\big{ \matrixel{\psi'}{W (F^{-1/2} - \Ident) P}{\psi} }
      \nonumber\\
    &= \Re`\big{ \matrixel{\chi}{ (F^{-1/2} - \Ident) P }{\psi} }
      + \Re`\big{ \dmatrixel{\psi}{ (F^{-1/2} - \Ident) P } }
      \nonumber\\
    &\geqslant  \Re`\big{ \matrixel{\chi}{ (F^{-1/2} - \Ident) P }{\psi} }
      \label{eq:fodisafodisahif}
  \end{align}
  since
  $\dmatrixel{\psi}{ (F^{-1/2} - \Ident) P } = \dmatrixel{\psi}{ P (F^{-1/2} -
    \Ident) P } \geqslant 0$ as $P$ commutes with $F^{-1/2}$ and since
  $F=W^\dagger W\leqslant\Ident$ implies that $F^{-1/2}\geqslant\Ident$.  To
  bound the remaining term we first write
  $\abs[\big]{ \matrixel{\chi}{ (F^{-1/2} - \Ident) P }{\psi} } \leqslant
  \norm[\big]{\ket\chi}\norm[\big]{(F^{-1/2}-\Ident) P \ket{\psi}} \leqslant
  \sqrt{2\epsilon/\nu}$; the last inequality follows since $P$ projects onto the
  eigenspaces of $F$ with eigenvalues larger than or equal to $\nu$, thus
  $F^{-1/2}P \leqslant \nu^{-1/2}P$ and hence
  $\norm[\big]{(F^{-1/2}-\Ident) P\ket{\psi}} \leqslant \nu^{-1/2}-1 \leqslant
  \nu^{-1/2}$.  Hence, we have
  \begin{align}
    \text{\eqref{eq:fodisafodisahif}}
    \geqslant - \sqrt{\frac{2\epsilon}{\nu}}
  \end{align}
  Following the inequalities from~\eqref{eq:fidnsofda3e},
  invoking~\eqref{eq:fiodbsnfois} and with the above, we finally obtain
  \begin{align}
    \Re`*{ \matrixel{\psi'}{U}{\psi} }
    \geqslant 1 - \epsilon - 2\sqrt{2\epsilon+\nu} - \sqrt{\frac{2\epsilon}{\nu}}\ .
  \end{align}
  Choosing $\nu=2\epsilon^{1/2}$, we obtain, using $\epsilon\leqslant\sqrt{\epsilon}$,
  \begin{align}
    1 - \Re`*{ \matrixel{\psi'}{U}{\psi} }
    \leqslant
    \epsilon + 2\sqrt{2\epsilon+2\epsilon^{1/2}} + \sqrt{\frac{2\epsilon}{2\epsilon^{1/2}}}
    \leqslant (1 + 4 + 1)\, \epsilon^{1/4} = 6\,\epsilon^{1/4}
  \end{align}
  as claimed.
\end{proof}

\section{Robust counterexample against extensions of Construction \#1}
\label{app:smooth}

In this appendix we show that the counterexample
of~\cref{sec:entropic-proof-approach-nontrivial-Hamiltonians} is robust to small
errors on the process. The process is
$\mathcal{E}_{X\to X'}(\cdot) = \tr[\cdot]\proj+$, where
$\ket+ = [\ket0+\ket1]/\sqrt{2}$ with $\ket0,\ket1$ energy eigenstates of
respective energies $E_0=0$, $E_1>0$; we write $H_X = \sum_{j=0,1} E_j \proj{j}$
and $\Gamma_X = \ee^{-\beta H_X}$.  The initial state on $X$ and a reference
system $R_X\simeq X$ is the maximally entangled state
$\ket\sigma_{XR_X} = [\ket{00} + \ket{11}]/\sqrt{2} =
\ket{\Phi}_{X:R_X}/\sqrt{2}$.

We seek a map $\mathcal{T}_{X\to X'}$ such that
\begin{align}
  P(\mathcal{T}_{X\to X'}(\sigma_{XR_X}), \mathcal{E}_{X\to X'}(\sigma_{XR_X}))
  \leqslant \epsilon
  \qquad\text{and}\qquad
  \mathcal{T}_{X\to X}(\Gamma_X)\leqslant\alpha\Gamma_{X'}\ ,
  \label{eq:ofiuoepjdijklndls}
\end{align}
for a $\alpha$ that is independent of $E_0,E_1$.  Here we have $X\simeq X'$ and
$\Gamma_X=\Gamma_{X'}$.

Let $\rho_{X'R_X} = \mathcal{E}_{X\to
  X'}(\sigma_{XR_X})$. From~\eqref{eq:ofiuoepjdijklndls} we find
$\frac12 \norm{\mathcal{T}_{X\to X'}(\sigma_{XR_X}) - \rho_{X'R_X}}_1
\leqslant\epsilon$, which in turn implies that
$(1/4)\bigl\lVert\mathcal{T}_{X\to X'}(\Phi_{X:R_X}) -
\proj+_{X'}\otimes\Ident_{R_X}\bigr\rVert_1 \leqslant \epsilon$, and hence that
$\mathcal{T}_{X\to X'}(\cdot) = \tr[\cdot]\,\proj{+}_{X'} + \Delta(\cdot)$ for
some Hermiticity preserving map $\Delta(\cdot)$ satisfying
$\frac12\norm{\Delta(\Phi_{XR_X})}_1\leqslant 2\epsilon$.

Let $\Delta_\pm\geqslant 0$ be the positive and negative parts of
$\Delta(\Gamma) = \Delta_+ - \Delta_-$, noting that
$\tr(\Delta_-) \leqslant \tr(\Delta_-) + \tr(\Delta_+) = \norm{\Delta(\Gamma)}_1
= \norm{\tr_{R_X}`\big(
  \Gamma_{R_X}^{1/2}\,\Delta(\Phi_{X:R_X})\,\Gamma_{R_X}^{1/2} ) }_1$, defining
$\Gamma_{R_X}$ as the transpose of $\Gamma_X$ onto the system $R_X$, and
continuing the computation we obtain
$\tr(\Delta_-) \leqslant
\norm{\Gamma_{R_X}^{1/2}\,\Delta(\Phi_{X:R_X})\,\Gamma_{R_X}^{1/2} }_1 \leqslant
\norm{\Gamma_{R_X}}_\infty \norm{ \Delta(\Phi_{X:R_X}) }_1 \leqslant 4\epsilon$,
using the fact that $\norm{\Gamma_X}_\infty = \max_j `{ \ee^{-\beta E_j} } = 1$.

The condition $\mathcal{T}_{X\to X'}(\Gamma) \leqslant \alpha \Gamma$ implies
that
$\alpha\Gamma \geqslant \tr[\Gamma]\proj{+} + \Delta(\Gamma) \geqslant \proj{+}
- \Delta_-$. Hence, we have that
$\alpha^{-1}\,\proj{+} \leqslant \Gamma + \Delta_-/\alpha$.  Hence, for any
$0<\eta\leqslant 1$ to be fixed later, $\mu=\alpha^{-1}$ is feasible for the
dual problem~\eqref{eq:defn-hypothesis-testing-entropy--dual} defining the
hypothesis testing entropy $\DHyp[\eta]{\proj{+}}{\Gamma}$, and
$\ee^{-\DHyp[\eta]{\proj{+}}{\Gamma}} \geqslant \alpha^{-1} -
\tr[\Delta_-/\alpha]/\eta \geqslant \alpha^{-1} `\big( 1 -
4\epsilon/\eta)$. Thus, we have
$\ln(\alpha) \geqslant \DHyp[\eta]{\proj+}{\Gamma} + \ln(1 - 4\epsilon/\eta)$.
Choosing $\eta=8\epsilon$ yields $\ln(1 - 4\epsilon/\eta) = -\ln(2)$.

On the other hand, by definition we have
$\ee^{-\DHyp[\eta]{\proj+}{\Gamma}} \leqslant \tr[Q\Gamma]/\eta$ for any
$0\leqslant Q\leqslant \Ident$ satisfying $\tr[Q\proj+]\geqslant\eta$; with
$Q=2\eta\proj1$ we obtain
$\ee^{-\DHyp[\eta]{\proj+}{\Gamma}} \leqslant 2\ee^{-\beta E_1}$ and thus
$\DHyp[\eta]{\proj+}{\Gamma} \geqslant \beta E_1 -\ln(2)$.

Then,
$\ln(\alpha) \geqslant -\ln(2) + \beta E_1 -\ln(2) = -2\ln(2) + \beta E_1$.  Now
let $\alpha$ be the optimal candidate in the coherent relative entropy
$\DCohz[\epsilon]{\rho}{X}{X'}{\Gamma}{\Gamma} = -\ln(\alpha)$.  We finally see
that the transformation $\Ident/2\to\ket+$ may require arbitrarily much energy
if $E_1\to\infty$, even for a small $\epsilon>0$, since
\begin{align}
  \text{energy cost} =
  - \beta^{-1}\DCohz[\epsilon]{\rho}{X}{X'}{\Gamma}{\Gamma}
  = \beta^{-1}\ln(\alpha)
  \geqslant E_1 - 2 \beta^{-1} \ln(2)\ .
\end{align}

\section{Universal conditional typical projector for trivial Hamiltonians}
\label{appx:universal-conditional-typical-projector}

In the case of trivial Hamiltonians,
\cref{defn:universal-relative-conditional-typical-smoother} can be
simplified. We call the corresponding object a universal conditional typical
projector

\begin{definition}
  \label{defn:univ-cond-typical-proj}
  Consider two systems with Hilbert spaces $\Hs_A,\Hs_B$ and let
  $s\in\mathbb{R}$.  We define a \emph{universal conditional typical projector}
  $P^{s,\delta}_{A^nB^n}$ with parameter $\delta>0$ as a projector acting on
  $(\Hs_A\otimes\Hs_B)^{\otimes n}$ such that:
  \begin{enumerate}[label=(\roman*)]
  \item \label{item:prop-univ-quant-cond-typical-subspace-high-weight} There
    exists $\eta>0$ independent of $n$ such that for any quantum state
    $\rho_{AB}$ with $\HH[\rho]{A}[B]\leqslant s$, we have
    \begin{align}
      \tr`\big[P^{s,\delta}_{A^nB^n}\,\rho_{AB}^{\otimes n}]
      \geqslant 1 - \poly(n)\exp(-n\eta)\ ;
    \end{align}
  \item \label{item:prop-univ-quant-cond-typical-subspace-conditional-size}
    $\displaystyle \tr_{A^n}`\big[ P^{s,\delta}_{A^nB^n} ] \leqslant
    \poly(n)\,\ee^{n (s + 2\delta)}\,\Ident_{B^n}$.
  \end{enumerate}
\end{definition}

Observe that we choose to define the object in
\cref{defn:univ-cond-typical-proj} as a projector whereas we only require the
object in \cref{defn:universal-relative-conditional-typical-smoother} to be an
operator of norm at most $1$.  The reason is that while we can prove that a
projector satisfying the conditions of \cref{defn:univ-cond-typical-proj}
exists, we are currently not able to guarantee the existence of a projector
satisfying the criteria of
\cref{defn:universal-relative-conditional-typical-smoother}.

\begin{proposition}\label{prop:univ-quant-cond-typical-subspace}
  Consider two systems $A,B$ and let $s\in\mathbb{R}$.  For any $\delta>0$ and
  $n\in\mathbb{N}$ there exists a universal conditional typical projector
  $P_{A^nB^n}^{s,\delta}$ that is permutation-invariant.
\end{proposition}

The proof of \cref{prop:univ-quant-cond-typical-subspace} is developed in the
rest of this appendix. To understand why the projector of
\cref{defn:univ-cond-typical-proj} is conditional\,---\,as well as for a simple
illustration of its use\,---\,consider the smooth R\'enyi-zero conditional
max-entropy, also known as the smooth alternative
max-entropy~\cite{BookTomamichel2016_Finite}. It is defined for a bipartite
state $\rho_{AB}$ as
\begin{align}
  \Hzero[\rho][\epsilon]{A}[B] = \min_{\hat\rho\approx_\epsilon\rho} \,
  \ln\,\norm[\big]{ \tr_A`\big[ \Pi^{\hat\rho_{AB}}_{AB} ] }_\infty\ ,
\end{align}
where $\Pi^{\hat\rho_{AB}}_{AB}$ is the projector onto the support of
$\hat\rho_{AB}$, and where the optimization ranges over sub-normalized states
$\hat\rho_{AB}$ which are $\epsilon$-close to $\rho_{AB}$ in purified distance.
We may understand the i.i.d.\@{} behaviour of this quantity as follows. For
$\delta>0$ and $n\in\mathbb{N}$ let $P^{s,\delta}_{A^nB^n}$ be a universal
conditional typical projector with $s=\HH[\rho]{A}[B]$. We define
$\hat\rho_{A^nB^n} = P^{s,\delta}\,\rho_{AB}^{\otimes n}\,P^{s,\delta}$. Then,
we have $\hat\rho_{A^nB^n} \approx_\epsilon \rho_{AB}^{\otimes n}$ for
$n\in\mathbb{N}$ large enough, thanks to
Property~\ref{item:prop-univ-quant-cond-typical-subspace-high-weight} and the
gentle measurement lemma (\cref{x:gentle-measurement-lemma}).  On the other hand, using
Property~\ref{item:prop-univ-quant-cond-typical-subspace-conditional-size} we
have
\begin{align}
  \frac1n \Hzero[\rho^{\otimes n}][\epsilon]{A^n}[B^n]
  \leqslant \frac1n \ln\,\norm[\big]{ \tr_{A^n} `\big[ P^{s,\delta} ] }_\infty
  \leqslant \HH[\rho]{A}[B] + 2\delta + \frac1n\ln(\poly(n))
\end{align}
such that taking the limits $n\to\infty$ and $\delta\to0$, we get that the
smooth R\'enyi-zero conditional entropy is asymptotically upper bounded by the
von Neumann conditional entropy in the i.i.d.\@{} regime.

We proceed to construct a universal conditional typical projector based on ideas
from Schur-Weyl duality.  The construction presented here is similar to, and
inspired by, techniques put forward in earlier
work~\cite{Bjelakovic2003arXiv_revisted,PhDHarrow2005,Noetzel2012arXiv_two,Bennett2014_reverse,Berta2015QIC_monotonicity,Haah2017IEEETIT_sampleoptimal}.

\begin{proof}[*prop:univ-quant-cond-typical-subspace]
  Let
  \begin{align}
    P^{s,\delta}_{A^nB^n}
    = \sum_{\substack{ \lambda,\lambda'\;:\\
    \bar{H}(\lambda) - \bar{H}(\lambda') \leqslant s + 2\delta }}
    (\Ident_{A^n}\otimes\Pi^{\lambda'}_{B^n})\,\Pi^{\lambda}_{A^nB^n}\ ,
  \end{align}
  where the respective projectors $\Pi^{\lambda'}_{B^n}$,
  $\Pi^{\lambda}_{A^nB^n}$ refer to Schur-Weyl decompositions of
  $\Hs_{B}^{\otimes n}$ and of $(\Hs_{A}\otimes\Hs_B)^{\otimes n}$,
  respectively, $\lambda\in\Young{d_Ad_B}{n}$ and
  $\lambda'\in\Young{d_{B}}{n}$.  Observe that $P^{s,\delta}_{A^nB^n}$ is a
  projector: Each term in the sum is a projector as a product of two commuting
  projectors (\cref{lemma:SchurWeylTensProdDecompCommute}), and each term of
  the sum acts on a different subspace of $(\Hs_A\otimes\Hs_B)^{\otimes n}$. The projector $P^{s,\delta}_{A^nB^n}$ corresponds to the measurement of the
  two commuting POVMs $`\big{ \Pi^{\lambda}_{A^nB^n} }$ and
  $`\big{ \Pi^{\lambda'}_{B^n} }$, and testing whether or not the event
  $\bar{H}(\lambda) - \bar{H}(\lambda') \leqslant s + 2\delta$ is satisfied.
  Also by construction $P^{s,\delta}_{A^nB^n}$ is permutation-invariant.

  For any $\rho_{AB}$ with $\HH[\rho]{A}[B] \leqslant s$, the probability that the
  measurement of $P^{s,\delta}_{A^nB^n}$ fails on $\rho_{AB}^{\otimes n}$ can be
  upper bounded as follows.  The passing event
  $\bar{H}(\lambda) - \bar{H}(\lambda') \leqslant s + 2\delta$ is implied in
  particular by the two events (a)
  $\bar{H}(\lambda) \leqslant H(AB)_\rho + \delta$ and (b)
  $\bar{H}(\lambda') \geqslant H(B)_\rho - \delta$ happening simultaneously,
  recalling that $H(AB)_\rho - H(B)_\rho = \HH[\rho]{A}[B] \leqslant s$.  The
  probability of event (a) failing is
  \begin{align}
  \Pr`*[ \bar{H}(\lambda) > H(AB)_\rho + \delta ]\leqslant \poly(n) \exp`*(-n\eta)
  \end{align}
  as given by \cref{prop:entropy-measurement-POVM-n-systems}, and similarly for
  event (b)
  \begin{align}
    \Pr`*[ \bar{H}(\lambda') < H(B)_\rho - \delta ]
    \leqslant \poly(n) \exp`*(-n\eta)\ .
  \end{align}
  We can use the same $\eta$ in both cases by picking the lesser of the two
  values given
  by \cref{prop:entropy-measurement-POVM-n-systems}, if
  necessary.  Note furthermore that $\eta>0$ does not depend on $\rho$. Hence
  with this $\eta$, for any $\rho_{AB}$ we have
  \begin{align}
    \tr`\big[P^{s,\delta}_{A^nB^n}\, \rho_{AB}^{\otimes n}]
    \geqslant 1 - \poly(n)\exp(-n\eta)
  \end{align}
  as required.
  
  For the second property, we use
  \cref{lemma:Schur-Weyl-bipartite-trAn-PiLambdaPiLambdaprime} to write
  \begin{align}
    \tr_{A^n}`\big[ P^{s,\delta}_{A^nB^n} ]
    &= \sum_{\substack{ \lambda,\lambda'\;:\\
    \bar{H}(\lambda) - \bar{H}(\lambda') \leqslant s + 2\delta }}
    \Pi^{\lambda'}_{B^n}\, \tr_{A^n}`*[\Pi^{\lambda}_{A^nB^n}] \,
    \Pi^{\lambda'}_{B^n}
    \nonumber\\
    &\leqslant \sum_{\substack{ \lambda,\lambda'\;:\\
    \bar{H}(\lambda) - \bar{H}(\lambda') \leqslant s + 2\delta }}
    \poly(n)\,\ee^{n(\bar{H}(\lambda) - \bar{H}(\lambda'))} \Ident_{B^n}
    \nonumber\\
    &\leqslant 
    \poly(n)\,\ee^{n(s+2\delta)}\, \Ident_{B^n}
  \end{align}
  recalling that there are only $\poly(n)$ many possible Young diagrams and
  hence at most so many terms in the sum.
\end{proof}

\section{Universal conditional erasure for $n$ copies and trivial Hamiltonians}
\label{app:conditional-erasure}
\label{appx:last} %

\begin{corollary}[Thermodynamic protocol for universal conditional erasure for $n$ copies]%
  \label{thm:thermo-universal-cond-erasure-ncopies-trivH}
  Let $S,M$ be systems, let $\sigma_{S}$ be the maximally mixed state on $S$.
  Let $s < \ln(d_S)$, where $d_S$ is the dimension of $S$, and let $\delta>0$
  small enough.  Let $n\in\mathbb{N}$ be large enough. Let $J$ be a large enough information battery and let any
  $m \leqslant n(\ln(d_S) - s - 3\delta)$ such that $e^m$ is integer.
  
  Then, there exists $\eta'>0$ and a thermal operation
  $\mathcal{R}_{S^nM^nJ\to S^nM^nJ}$ acting on the systems $S^nM^nJ$, such that
  the effective work process $\mathcal{T}_{S^nM^n\to S^nM^n}$ of
  $\mathcal{R}_{S^nM^nJ\to S^nM^nJ}$ with respect to the battery states
  $(\tau_J^m,\ket0_J)$ is a universal conditional
  $(\poly(n)\,\ee^{-n\eta'})$-erasure process resetting $S^n$ to the state
  $\sigma_{S}^{\otimes n}$ with respect to the set of states
  $\mathcal{S}_{S^nM^n}'$, where $\mathcal{S}_{S^nM^n}'$ is the convex hull of
  $\mathcal{S}_{S^nM^n} = `\big{ \rho_{SM}^{\otimes n} : \HH[\rho]{S}[M]
    \leqslant s }$.
\end{corollary}

The case where $s = \ln(d_S)$ is uninteresting as we cannot hope to extract any
work. In such cases one can simply set $m=0$ and take $\mathcal{R}_{S^nM^nJ}$
to be the thermal operation that completely thermalizes $S^n$.

\begin{proof}[**thm:thermo-universal-cond-erasure-ncopies-trivH]
  This is in fact a relatively straightforward application of
  \cref{thm:univ-conditional-erasure-PBD-as-thermalop} over $n$ copies of $SM$.
  Let $P_{S^nM^n}^{s,\delta}$ be given by
  \cref{prop:univ-quant-cond-typical-subspace}.  We seek $\kappa,\kappa'$ that
  satisfy \eqref{eq:univ-isom-cond-eras-conditions-hypo-test}.  We can choose
  $\kappa = \poly(n)\exp`*{-n\eta(\delta)}$ thanks to
  \cref{defn:univ-cond-typical-proj}.  Furthermore for any
  $\rho_{SM}^{\otimes n}\in\mathcal{S}_{S^nM^n}$ we have
  \begin{align}
    \tr`*[ P_{S^nM^n} `*(\frac{\Ident_S}{d_S}\otimes\rho_M)^{\otimes n}]
    \leqslant \poly(n)\,\ee^{n(s+2\delta)}\,d_S^{-n}\,\tr[\rho_M^{\otimes n}]
    &= \poly(n)\,\ee^{-n(\ln(d_S) - s - 2\delta)}\notag
    \\
    &\leqslant \frac{\poly(n)\,\ee^{-n\delta}}{{e^m}}
    \label{eq:fiodbafdosakf}
  \end{align}
  and thus we may take $\kappa' = \poly(n)\,\ee^{-n\delta}$.
  Finally, $\eta'$ is given as $\eta'=\min\{\delta,\eta(\delta)\}$.
\end{proof}

\bibliographysetup


\begin{thebibliography}{77}%
\makeatletter
\providecommand \@ifxundefined [1]{%
 \@ifx{#1\undefined}
}%
\providecommand \@ifnum [1]{%
 \ifnum #1\expandafter \@firstoftwo
 \else \expandafter \@secondoftwo
 \fi
}%
\providecommand \@ifx [1]{%
 \ifx #1\expandafter \@firstoftwo
 \else \expandafter \@secondoftwo
 \fi
}%
\providecommand \natexlab [1]{#1}%
\providecommand \enquote  [1]{``#1''}%
\providecommand \bibnamefont  [1]{#1}%
\providecommand \bibfnamefont [1]{#1}%
\providecommand \citenamefont [1]{#1}%
\providecommand \href@noop [0]{\@secondoftwo}%
\providecommand \href [0]{\begingroup \@sanitize@url \@href}%
\providecommand \@href[1]{\@@startlink{#1}\@@href}%
\providecommand \@@href[1]{\endgroup#1\@@endlink}%
\providecommand \@sanitize@url [0]{\catcode `\\12\catcode `\$12\catcode
  `\&12\catcode `\#12\catcode `\^12\catcode `\_12\catcode `\%12\relax}%
\providecommand \@@startlink[1]{}%
\providecommand \@@endlink[0]{}%
\providecommand \url  [0]{\begingroup\@sanitize@url \@url }%
\providecommand \@url [1]{\endgroup\@href {#1}{\urlprefix }}%
\providecommand \urlprefix  [0]{URL }%
\providecommand \Eprint [0]{\href }%
\providecommand \doibase [0]{https://doi.org/}%
\providecommand \selectlanguage [0]{\@gobble}%
\providecommand \bibinfo  [0]{\@secondoftwo}%
\providecommand \bibfield  [0]{\@secondoftwo}%
\providecommand \translation [1]{[#1]}%
\providecommand \BibitemOpen [0]{}%
\providecommand \bibitemStop [0]{}%
\providecommand \bibitemNoStop [0]{.\EOS\space}%
\providecommand \EOS [0]{\spacefactor3000\relax}%
\providecommand \BibitemShut  [1]{\csname bibitem#1\endcsname}%
\let\auto@bib@innerbib\@empty
\bibitem [{\citenamefont {Goold}\ \emph {et~al.}(2016)\citenamefont {Goold},
  \citenamefont {Huber}, \citenamefont {Riera}, \citenamefont {Rio},\ and\
  \citenamefont {Skrzypczyk}}]{Goold2016JPA_review}%
  \BibitemOpen
  \bibfield  {author} {\bibinfo {author} {\bibfnamefont {J.}~\bibnamefont
  {Goold}}, \bibinfo {author} {\bibfnamefont {M.}~\bibnamefont {Huber}},
  \bibinfo {author} {\bibfnamefont {A.}~\bibnamefont {Riera}}, \bibinfo
  {author} {\bibfnamefont {L.~d.}\ \bibnamefont {Rio}},\ and\ \bibinfo {author}
  {\bibfnamefont {P.}~\bibnamefont {Skrzypczyk}},\ }\bibfield  {title}
  {\bibinfo {title} {{T}he role of quantum information in
  thermodynamics{\textemdash}a topical review},\ }\href
  {https://doi.org/10.1088/1751-8113/49/14/143001} {\bibfield  {journal}
  {\bibinfo  {journal} {Journal of Physics A: Mathematical and Theoretical}\
  }\textbf {\bibinfo {volume} {49}},\ \bibinfo {pages} {143001} (\bibinfo
  {year} {2016})}\BibitemShut {NoStop}%
\bibitem [{\citenamefont {Brand{\~a}o}\ \emph {et~al.}(2013)\citenamefont
  {Brand{\~a}o}, \citenamefont {Horodecki}, \citenamefont {Oppenheim},
  \citenamefont {Renes},\ and\ \citenamefont
  {Spekkens}}]{Brandao2013_resource}%
  \BibitemOpen
  \bibfield  {author} {\bibinfo {author} {\bibfnamefont {F.~G. S.~L.}\
  \bibnamefont {Brand{\~a}o}}, \bibinfo {author} {\bibfnamefont
  {M.}~\bibnamefont {Horodecki}}, \bibinfo {author} {\bibfnamefont
  {J.}~\bibnamefont {Oppenheim}}, \bibinfo {author} {\bibfnamefont {J.~M.}\
  \bibnamefont {Renes}},\ and\ \bibinfo {author} {\bibfnamefont {R.~W.}\
  \bibnamefont {Spekkens}},\ }\bibfield  {title} {\bibinfo {title} {{R}esource
  theory of quantum states out of thermal equilibrium},\ }\href
  {https://doi.org/10.1103/PhysRevLett.111.250404} {\bibfield  {journal}
  {\bibinfo  {journal} {Physical Review Letters}\ }\textbf {\bibinfo {volume}
  {111}},\ \bibinfo {pages} {250404} (\bibinfo {year} {2013})}\BibitemShut
  {NoStop}%
\bibitem [{\citenamefont {Brand{\~a}o}\ \emph {et~al.}(2015)\citenamefont
  {Brand{\~a}o}, \citenamefont {Horodecki}, \citenamefont {Ng}, \citenamefont
  {Oppenheim},\ and\ \citenamefont {Wehner}}]{Brandao2015PNAS_secondlaws}%
  \BibitemOpen
  \bibfield  {author} {\bibinfo {author} {\bibfnamefont {F.}~\bibnamefont
  {Brand{\~a}o}}, \bibinfo {author} {\bibfnamefont {M.}~\bibnamefont
  {Horodecki}}, \bibinfo {author} {\bibfnamefont {N.}~\bibnamefont {Ng}},
  \bibinfo {author} {\bibfnamefont {J.}~\bibnamefont {Oppenheim}},\ and\
  \bibinfo {author} {\bibfnamefont {S.}~\bibnamefont {Wehner}},\ }\bibfield
  {title} {\bibinfo {title} {{T}he second laws of quantum thermodynamics},\
  }\href {https://doi.org/10.1073/pnas.1411728112} {\bibfield  {journal}
  {\bibinfo  {journal} {Proceedings of the National Academy of Sciences}\
  }\textbf {\bibinfo {volume} {112}},\ \bibinfo {pages} {3275--3279} (\bibinfo
  {year} {2015})}\BibitemShut {NoStop}%
\bibitem [{\citenamefont {Chitambar}\ and\ \citenamefont
  {Gour}(2019)}]{Chitambar2019RMP_resource}%
  \BibitemOpen
  \bibfield  {author} {\bibinfo {author} {\bibfnamefont {E.}~\bibnamefont
  {Chitambar}}\ and\ \bibinfo {author} {\bibfnamefont {G.}~\bibnamefont
  {Gour}},\ }\bibfield  {title} {\bibinfo {title} {{Q}uantum resource
  theories},\ }\href {https://doi.org/10.1103/RevModPhys.91.025001} {\bibfield
  {journal} {\bibinfo  {journal} {Reviews of Modern Physics}\ }\textbf
  {\bibinfo {volume} {91}},\ \bibinfo {pages} {025001} (\bibinfo {year}
  {2019})}\BibitemShut {NoStop}%
\bibitem [{\citenamefont {Janzing}\ \emph {et~al.}(2000)\citenamefont
  {Janzing}, \citenamefont {Wocjan}, \citenamefont {Zeier}, \citenamefont
  {Geiss},\ and\ \citenamefont {Beth}}]{Janzing2000_cost}%
  \BibitemOpen
  \bibfield  {author} {\bibinfo {author} {\bibfnamefont {D.}~\bibnamefont
  {Janzing}}, \bibinfo {author} {\bibfnamefont {P.}~\bibnamefont {Wocjan}},
  \bibinfo {author} {\bibfnamefont {R.}~\bibnamefont {Zeier}}, \bibinfo
  {author} {\bibfnamefont {R.}~\bibnamefont {Geiss}},\ and\ \bibinfo {author}
  {\bibfnamefont {T.}~\bibnamefont {Beth}},\ }\bibfield  {title} {\bibinfo
  {title} {{T}hermodynamic cost of reliability and low temperatures: Tightening
  {Landauer}'s principle and the second law},\ }\href
  {https://doi.org/10.1023/A:1026422630734} {\bibfield  {journal} {\bibinfo
  {journal} {International Journal of Theoretical Physics}\ }\textbf {\bibinfo
  {volume} {39}},\ \bibinfo {pages} {2717--2753} (\bibinfo {year}
  {2000})}\BibitemShut {NoStop}%
\bibitem [{\citenamefont {Faist}\ \emph
  {et~al.}(2015{\natexlab{a}})\citenamefont {Faist}, \citenamefont
  {Oppenheim},\ and\ \citenamefont {Renner}}]{Faist2015NJP_Gibbs}%
  \BibitemOpen
  \bibfield  {author} {\bibinfo {author} {\bibfnamefont {P.}~\bibnamefont
  {Faist}}, \bibinfo {author} {\bibfnamefont {J.}~\bibnamefont {Oppenheim}},\
  and\ \bibinfo {author} {\bibfnamefont {R.}~\bibnamefont {Renner}},\
  }\bibfield  {title} {\bibinfo {title} {{Gibbs}-preserving maps outperform
  thermal operations in the quantum regime},\ }\href
  {https://doi.org/10.1088/1367-2630/17/4/043003} {\bibfield  {journal}
  {\bibinfo  {journal} {New Journal of Physics}\ }\textbf {\bibinfo {volume}
  {17}},\ \bibinfo {pages} {043003} (\bibinfo {year}
  {2015}{\natexlab{a}})}\BibitemShut {NoStop}%
\bibitem [{\citenamefont {{\r{A}}berg}(2013)}]{Aberg2013_worklike}%
  \BibitemOpen
  \bibfield  {author} {\bibinfo {author} {\bibfnamefont {J.}~\bibnamefont
  {{\r{A}}berg}},\ }\bibfield  {title} {\bibinfo {title} {{T}ruly work-like
  work extraction via a single-shot analysis},\ }\href
  {https://doi.org/10.1038/ncomms2712} {\bibfield  {journal} {\bibinfo
  {journal} {Nature Communications}\ }\textbf {\bibinfo {volume} {4}},\
  \bibinfo {pages} {1925} (\bibinfo {year} {2013})}\BibitemShut {NoStop}%
\bibitem [{\citenamefont {Horodecki}\ and\ \citenamefont
  {Oppenheim}(2013)}]{Horodecki2013_ThermoMaj}%
  \BibitemOpen
  \bibfield  {author} {\bibinfo {author} {\bibfnamefont {M.}~\bibnamefont
  {Horodecki}}\ and\ \bibinfo {author} {\bibfnamefont {J.}~\bibnamefont
  {Oppenheim}},\ }\bibfield  {title} {\bibinfo {title} {{F}undamental
  limitations for quantum and nanoscale thermodynamics},\ }\href
  {https://doi.org/10.1038/ncomms3059} {\bibfield  {journal} {\bibinfo
  {journal} {Nature Communications}\ }\textbf {\bibinfo {volume} {4}},\
  \bibinfo {pages} {2059} (\bibinfo {year} {2013})}\BibitemShut {NoStop}%
\bibitem [{\citenamefont {Renner}(2005)}]{PhDRenner2005_SQKD}%
  \BibitemOpen
  \bibfield  {author} {\bibinfo {author} {\bibfnamefont {R.}~\bibnamefont
  {Renner}},\ }\emph {\bibinfo {title} {{S}ecurity of Quantum Key
  Distribution}},\ \href {https://doi.org/10.3929/ethz-a-005115027} {Ph.D.
  thesis},\ \bibinfo  {school} {ETH Z{\"u}rich} (\bibinfo {year} {2005}),\
  \Eprint {https://arxiv.org/abs/quant-ph/0512258} {arXiv:quant-ph/0512258}
  \BibitemShut {NoStop}%
\bibitem [{\citenamefont {Tomamichel}(2012)}]{PhDTomamichel2012}%
  \BibitemOpen
  \bibfield  {author} {\bibinfo {author} {\bibfnamefont {M.}~\bibnamefont
  {Tomamichel}},\ }\emph {\bibinfo {title} {{A} Framework for Non-Asymptotic
  Quantum Information Theory}},\ \href {https://doi.org/10.3929/ethz-a-7356080}
  {\bibinfo {type} {Quantum physics; mathematical physics; mathematical
  physics}},\ \bibinfo  {school} {ETH Zurich} (\bibinfo {year} {2012}),\
  \Eprint {https://arxiv.org/abs/1203.2142} {arXiv:1203.2142} \BibitemShut
  {NoStop}%
\bibitem [{\citenamefont {Tomamichel}(2016)}]{BookTomamichel2016_Finite}%
  \BibitemOpen
  \bibfield  {author} {\bibinfo {author} {\bibfnamefont {M.}~\bibnamefont
  {Tomamichel}},\ }\href {https://doi.org/10.1007/978-3-319-21891-5} {\emph
  {\bibinfo {title} {{Q}uantum Information Processing with Finite
  Resources}}},\ \bibinfo {series} {SpringerBriefs in Mathematical Physics},
  Vol.~\bibinfo {volume} {5}\ (\bibinfo  {publisher} {Springer International
  Publishing},\ \bibinfo {address} {Cham},\ \bibinfo {year} {2016})\ \Eprint
  {https://arxiv.org/abs/1504.00233} {arXiv:1504.00233} \BibitemShut {NoStop}%
\bibitem [{\citenamefont {Chubb}\ \emph {et~al.}(2018)\citenamefont {Chubb},
  \citenamefont {Tomamichel},\ and\ \citenamefont
  {Korzekwa}}]{Chubb2018Qu_beyond}%
  \BibitemOpen
  \bibfield  {author} {\bibinfo {author} {\bibfnamefont {C.~T.}\ \bibnamefont
  {Chubb}}, \bibinfo {author} {\bibfnamefont {M.}~\bibnamefont {Tomamichel}},\
  and\ \bibinfo {author} {\bibfnamefont {K.}~\bibnamefont {Korzekwa}},\
  }\bibfield  {title} {\bibinfo {title} {{B}eyond the thermodynamic limit:
  finite-size corrections to state interconversion rates},\ }\href
  {https://doi.org/10.22331/q-2018-11-27-108} {\bibfield  {journal} {\bibinfo
  {journal} {Quantum}\ }\textbf {\bibinfo {volume} {2}},\ \bibinfo {pages}
  {108} (\bibinfo {year} {2018})}\BibitemShut {NoStop}%
\bibitem [{\citenamefont {Faist}\ \emph
  {et~al.}(2015{\natexlab{b}})\citenamefont {Faist}, \citenamefont {Dupuis},
  \citenamefont {Oppenheim},\ and\ \citenamefont {Renner}}]{Faist2015NatComm}%
  \BibitemOpen
  \bibfield  {author} {\bibinfo {author} {\bibfnamefont {P.}~\bibnamefont
  {Faist}}, \bibinfo {author} {\bibfnamefont {F.}~\bibnamefont {Dupuis}},
  \bibinfo {author} {\bibfnamefont {J.}~\bibnamefont {Oppenheim}},\ and\
  \bibinfo {author} {\bibfnamefont {R.}~\bibnamefont {Renner}},\ }\bibfield
  {title} {\bibinfo {title} {{T}he minimal work cost of information
  processing},\ }\href {https://doi.org/10.1038/ncomms8669} {\bibfield
  {journal} {\bibinfo  {journal} {Nature Communications}\ }\textbf {\bibinfo
  {volume} {6}},\ \bibinfo {pages} {7669} (\bibinfo {year}
  {2015}{\natexlab{b}})}\BibitemShut {NoStop}%
\bibitem [{\citenamefont {C{\^\i}rstoiu}\ and\ \citenamefont
  {Jennings}(2017)}]{Cirstoiu2017arXiv_gauge}%
  \BibitemOpen
  \bibfield  {author} {\bibinfo {author} {\bibfnamefont {C.}~\bibnamefont
  {C{\^\i}rstoiu}}\ and\ \bibinfo {author} {\bibfnamefont {D.}~\bibnamefont
  {Jennings}},\ }\bibfield  {title} {\bibinfo {title} {{G}lobal and local gauge
  symmetries beyond lagrangian formulations}} (\bibinfo {year} {2017}),\
  \bibinfo {note}
  {\href{https://arxiv.org/abs/1707.09826}{arXiv:1707.09826}}\BibitemShut
  {NoStop}%
\bibitem [{\citenamefont {Ben~Dana}\ \emph {et~al.}(2017)\citenamefont
  {Ben~Dana}, \citenamefont {Garc{\'\i}a~D{\'\i}az}, \citenamefont {Mejatty},\
  and\ \citenamefont {Winter}}]{Dana2017PRA_beyond}%
  \BibitemOpen
  \bibfield  {author} {\bibinfo {author} {\bibfnamefont {K.}~\bibnamefont
  {Ben~Dana}}, \bibinfo {author} {\bibfnamefont {M.}~\bibnamefont
  {Garc{\'\i}a~D{\'\i}az}}, \bibinfo {author} {\bibfnamefont {M.}~\bibnamefont
  {Mejatty}},\ and\ \bibinfo {author} {\bibfnamefont {A.}~\bibnamefont
  {Winter}},\ }\bibfield  {title} {\bibinfo {title} {{R}esource theory of
  coherence: Beyond states},\ }\href
  {https://doi.org/10.1103/PhysRevA.95.062327} {\bibfield  {journal} {\bibinfo
  {journal} {Physical Review A}\ }\textbf {\bibinfo {volume} {95}},\ \bibinfo
  {pages} {062327} (\bibinfo {year} {2017})}\BibitemShut {NoStop}%
\bibitem [{\citenamefont {Faist}\ and\ \citenamefont
  {Renner}(2018)}]{Faist2018PRX_workcost}%
  \BibitemOpen
  \bibfield  {author} {\bibinfo {author} {\bibfnamefont {P.}~\bibnamefont
  {Faist}}\ and\ \bibinfo {author} {\bibfnamefont {R.}~\bibnamefont {Renner}},\
  }\bibfield  {title} {\bibinfo {title} {{F}undamental work cost of quantum
  processes},\ }\href {https://doi.org/10.1103/PhysRevX.8.021011} {\bibfield
  {journal} {\bibinfo  {journal} {Physical Review X}\ }\textbf {\bibinfo
  {volume} {8}},\ \bibinfo {pages} {021011} (\bibinfo {year}
  {2018})}\BibitemShut {NoStop}%
\bibitem [{\citenamefont {Smith}(2010)}]{Smith10}%
  \BibitemOpen
  \bibfield  {author} {\bibinfo {author} {\bibfnamefont {G.}~\bibnamefont
  {Smith}},\ }\bibfield  {title} {\bibinfo {title} {{Q}uantum channel
  capacities},\ }in\ \href {https://doi.org/10.1109/CIG.2010.5592851} {\emph
  {\bibinfo {booktitle} {IEEE Information Theory Workshop}}}\ (\bibinfo {year}
  {2010})\ pp.\ \bibinfo {pages} {1--5}\BibitemShut {NoStop}%
\bibitem [{\citenamefont {Christandl}\ \emph {et~al.}(2009)\citenamefont
  {Christandl}, \citenamefont {K{\"o}nig},\ and\ \citenamefont
  {Renner}}]{Christandl2009PRL_Postselection}%
  \BibitemOpen
  \bibfield  {author} {\bibinfo {author} {\bibfnamefont {M.}~\bibnamefont
  {Christandl}}, \bibinfo {author} {\bibfnamefont {R.}~\bibnamefont
  {K{\"o}nig}},\ and\ \bibinfo {author} {\bibfnamefont {R.}~\bibnamefont
  {Renner}},\ }\bibfield  {title} {\bibinfo {title} {{P}ostselection technique
  for quantum channels with applications to quantum cryptography},\ }\href
  {https://doi.org/10.1103/PhysRevLett.102.020504} {\bibfield  {journal}
  {\bibinfo  {journal} {Physical Review Letters}\ }\textbf {\bibinfo {volume}
  {102}},\ \bibinfo {pages} {20504} (\bibinfo {year} {2009})}\BibitemShut
  {NoStop}%
\bibitem [{\citenamefont {Anshu}\ \emph
  {et~al.}(2019{\natexlab{a}})\citenamefont {Anshu}, \citenamefont {Jain},\
  and\ \citenamefont {Warsi}}]{Anshu2019IEEETIT_oneshot}%
  \BibitemOpen
  \bibfield  {author} {\bibinfo {author} {\bibfnamefont {A.}~\bibnamefont
  {Anshu}}, \bibinfo {author} {\bibfnamefont {R.}~\bibnamefont {Jain}},\ and\
  \bibinfo {author} {\bibfnamefont {N.~A.}\ \bibnamefont {Warsi}},\ }\bibfield
  {title} {\bibinfo {title} {{B}uilding blocks for communication over noisy
  quantum networks},\ }\href {https://doi.org/10.1109/TIT.2018.2851297}
  {\bibfield  {journal} {\bibinfo  {journal} {IEEE Transactions on Information
  Theory}\ }\textbf {\bibinfo {volume} {65}},\ \bibinfo {pages} {1287--1306}
  (\bibinfo {year} {2019}{\natexlab{a}})}\BibitemShut {NoStop}%
\bibitem [{\citenamefont {Faist}\ \emph {et~al.}(2019)\citenamefont {Faist},
  \citenamefont {Berta},\ and\ \citenamefont
  {Brand{\~a}o}}]{OurShortPaperToAppear}%
  \BibitemOpen
  \bibfield  {author} {\bibinfo {author} {\bibfnamefont {P.}~\bibnamefont
  {Faist}}, \bibinfo {author} {\bibfnamefont {M.}~\bibnamefont {Berta}},\ and\
  \bibinfo {author} {\bibfnamefont {F.}~\bibnamefont {Brand{\~a}o}},\
  }\bibfield  {title} {\bibinfo {title} {{T}hermodynamic capacity of quantum
  processes},\ }\href {https://doi.org/10.1103/PhysRevLett.122.200601}
  {\bibfield  {journal} {\bibinfo  {journal} {Physical Review Letters}\
  }\textbf {\bibinfo {volume} {122}},\ \bibinfo {pages} {200601} (\bibinfo
  {year} {2019})}\BibitemShut {NoStop}%
\bibitem [{\citenamefont {Navascu{\'e}s}\ and\ \citenamefont
  {Garc{\'\i}a-Pintos}(2015)}]{Navascues2015PRL_nonthermal}%
  \BibitemOpen
  \bibfield  {author} {\bibinfo {author} {\bibfnamefont {M.}~\bibnamefont
  {Navascu{\'e}s}}\ and\ \bibinfo {author} {\bibfnamefont {L.~P.}\ \bibnamefont
  {Garc{\'\i}a-Pintos}},\ }\bibfield  {title} {\bibinfo {title} {{N}onthermal
  quantum channels as a thermodynamical resource},\ }\href
  {https://doi.org/10.1103/PhysRevLett.115.010405} {\bibfield  {journal}
  {\bibinfo  {journal} {Physical Review Letters}\ }\textbf {\bibinfo {volume}
  {115}},\ \bibinfo {pages} {010405} (\bibinfo {year} {2015})}\BibitemShut
  {NoStop}%
\bibitem [{\citenamefont {Bennett}\ \emph {et~al.}(2014)\citenamefont
  {Bennett}, \citenamefont {Devetak}, \citenamefont {Harrow}, \citenamefont
  {Shor},\ and\ \citenamefont {Winter}}]{Bennett2014_reverse}%
  \BibitemOpen
  \bibfield  {author} {\bibinfo {author} {\bibfnamefont {C.~H.}\ \bibnamefont
  {Bennett}}, \bibinfo {author} {\bibfnamefont {I.}~\bibnamefont {Devetak}},
  \bibinfo {author} {\bibfnamefont {A.~W.}\ \bibnamefont {Harrow}}, \bibinfo
  {author} {\bibfnamefont {P.~W.}\ \bibnamefont {Shor}},\ and\ \bibinfo
  {author} {\bibfnamefont {A.}~\bibnamefont {Winter}},\ }\bibfield  {title}
  {\bibinfo {title} {{T}he quantum reverse {Shannon} theorem and resource
  tradeoffs for simulating quantum channels},\ }\href
  {https://doi.org/10.1109/TIT.2014.2309968} {\bibfield  {journal} {\bibinfo
  {journal} {IEEE Transactions on Information Theory}\ }\textbf {\bibinfo
  {volume} {60}},\ \bibinfo {pages} {2926--2959} (\bibinfo {year}
  {2014})}\BibitemShut {NoStop}%
\bibitem [{\citenamefont {Berta}\ \emph {et~al.}(2011)\citenamefont {Berta},
  \citenamefont {Christandl},\ and\ \citenamefont
  {Renner}}]{Berta2011_reverse}%
  \BibitemOpen
  \bibfield  {author} {\bibinfo {author} {\bibfnamefont {M.}~\bibnamefont
  {Berta}}, \bibinfo {author} {\bibfnamefont {M.}~\bibnamefont {Christandl}},\
  and\ \bibinfo {author} {\bibfnamefont {R.}~\bibnamefont {Renner}},\
  }\bibfield  {title} {\bibinfo {title} {{T}he quantum reverse {Shannon}
  theorem based on one-shot information theory},\ }\href
  {https://doi.org/10.1007/s00220-011-1309-7} {\bibfield  {journal} {\bibinfo
  {journal} {Communications in Mathematical Physics}\ }\textbf {\bibinfo
  {volume} {306}},\ \bibinfo {pages} {579--615} (\bibinfo {year}
  {2011})}\BibitemShut {NoStop}%
\bibitem [{\citenamefont {Harrow}(2005)}]{PhDHarrow2005}%
  \BibitemOpen
  \bibfield  {author} {\bibinfo {author} {\bibfnamefont {A.~W.}\ \bibnamefont
  {Harrow}},\ }\emph {\bibinfo {title} {{A}pplications of coherent classical
  communication and the Schur transform to quantum information theory}},\
  \href@noop {} {Ph.D. thesis},\ \bibinfo  {school} {Massachusetts Institute of
  Technology} (\bibinfo {year} {2005}),\ \Eprint
  {https://arxiv.org/abs/quant-ph/0512255} {arXiv:quant-ph/0512255}
  \BibitemShut {NoStop}%
\bibitem [{\citenamefont {Haah}\ \emph {et~al.}(2017)\citenamefont {Haah},
  \citenamefont {Harrow}, \citenamefont {Ji}, \citenamefont {Wu},\ and\
  \citenamefont {Yu}}]{Haah2017IEEETIT_sampleoptimal}%
  \BibitemOpen
  \bibfield  {author} {\bibinfo {author} {\bibfnamefont {J.}~\bibnamefont
  {Haah}}, \bibinfo {author} {\bibfnamefont {A.~W.}\ \bibnamefont {Harrow}},
  \bibinfo {author} {\bibfnamefont {Z.}~\bibnamefont {Ji}}, \bibinfo {author}
  {\bibfnamefont {X.}~\bibnamefont {Wu}},\ and\ \bibinfo {author}
  {\bibfnamefont {N.}~\bibnamefont {Yu}},\ }\bibfield  {title} {\bibinfo
  {title} {{S}ample-optimal tomography of quantum states},\ }\href
  {https://doi.org/10.1109/TIT.2017.2719044} {\bibfield  {journal} {\bibinfo
  {journal} {IEEE Transactions on Information Theory}\ }\textbf {\bibinfo
  {volume} {63}},\ \bibinfo {pages} {5628--5641} (\bibinfo {year}
  {2017})}\BibitemShut {NoStop}%
\bibitem [{\citenamefont {N{\"o}tzel}(2012)}]{Noetzel2012arXiv_two}%
  \BibitemOpen
  \bibfield  {author} {\bibinfo {author} {\bibfnamefont {J.}~\bibnamefont
  {N{\"o}tzel}},\ }\bibfield  {title} {\bibinfo {title} {{A} solution to two
  party typicality using representation theory of the symmetric group}}
  (\bibinfo {year} {2012}),\ \bibinfo {note}
  {\href{https://arxiv.org/abs/1209.5094}{arXiv:1209.5094}}\BibitemShut
  {NoStop}%
\bibitem [{\citenamefont {Tomamichel}\ \emph {et~al.}(2010)\citenamefont
  {Tomamichel}, \citenamefont {Colbeck},\ and\ \citenamefont
  {Renner}}]{Tomamichel2010IEEE_Duality}%
  \BibitemOpen
  \bibfield  {author} {\bibinfo {author} {\bibfnamefont {M.}~\bibnamefont
  {Tomamichel}}, \bibinfo {author} {\bibfnamefont {R.}~\bibnamefont
  {Colbeck}},\ and\ \bibinfo {author} {\bibfnamefont {R.}~\bibnamefont
  {Renner}},\ }\bibfield  {title} {\bibinfo {title} {{D}uality between smooth
  min- and max-entropies},\ }\href {https://doi.org/10.1109/TIT.2010.2054130}
  {\bibfield  {journal} {\bibinfo  {journal} {IEEE Transactions on Information
  Theory}\ }\textbf {\bibinfo {volume} {56}},\ \bibinfo {pages} {4674--4681}
  (\bibinfo {year} {2010})}\BibitemShut {NoStop}%
\bibitem [{\citenamefont {Nielsen}\ and\ \citenamefont
  {Chuang}(2000)}]{BookNielsenChuang2000}%
  \BibitemOpen
  \bibfield  {author} {\bibinfo {author} {\bibfnamefont {M.~A.}\ \bibnamefont
  {Nielsen}}\ and\ \bibinfo {author} {\bibfnamefont {I.~L.}\ \bibnamefont
  {Chuang}},\ }\href@noop {} {\emph {\bibinfo {title} {{Q}uantum Computation
  and Quantum Information}}}\ (\bibinfo  {publisher} {Cambridge University
  Press},\ \bibinfo {year} {2000})\BibitemShut {NoStop}%
\bibitem [{\citenamefont {Wang}\ and\ \citenamefont
  {Renner}(2012)}]{Wang2012PRL_oneshot}%
  \BibitemOpen
  \bibfield  {author} {\bibinfo {author} {\bibfnamefont {L.}~\bibnamefont
  {Wang}}\ and\ \bibinfo {author} {\bibfnamefont {R.}~\bibnamefont {Renner}},\
  }\bibfield  {title} {\bibinfo {title} {{O}ne-shot classical-quantum capacity
  and hypothesis testing},\ }\href
  {https://doi.org/10.1103/PhysRevLett.108.200501} {\bibfield  {journal}
  {\bibinfo  {journal} {Physical Review Letters}\ }\textbf {\bibinfo {volume}
  {108}},\ \bibinfo {pages} {200501} (\bibinfo {year} {2012})}\BibitemShut
  {NoStop}%
\bibitem [{\citenamefont {Tomamichel}\ and\ \citenamefont
  {Hayashi}(2013)}]{Tomamichel2013_hierarchy}%
  \BibitemOpen
  \bibfield  {author} {\bibinfo {author} {\bibfnamefont {M.}~\bibnamefont
  {Tomamichel}}\ and\ \bibinfo {author} {\bibfnamefont {M.}~\bibnamefont
  {Hayashi}},\ }\bibfield  {title} {\bibinfo {title} {{A} hierarchy of
  information quantities for finite block length analysis of quantum tasks},\
  }\href {https://doi.org/10.1109/TIT.2013.2276628} {\bibfield  {journal}
  {\bibinfo  {journal} {IEEE Transactions on Information Theory}\ }\textbf
  {\bibinfo {volume} {59}},\ \bibinfo {pages} {7693--7710} (\bibinfo {year}
  {2013})}\BibitemShut {NoStop}%
\bibitem [{\citenamefont {Matthews}\ and\ \citenamefont
  {Wehner}(2014)}]{Matthews2014IEEETIT_blocklength}%
  \BibitemOpen
  \bibfield  {author} {\bibinfo {author} {\bibfnamefont {W.}~\bibnamefont
  {Matthews}}\ and\ \bibinfo {author} {\bibfnamefont {S.}~\bibnamefont
  {Wehner}},\ }\bibfield  {title} {\bibinfo {title} {{F}inite blocklength
  converse bounds for quantum channels},\ }\href
  {https://doi.org/10.1109/TIT.2014.2353614} {\bibfield  {journal} {\bibinfo
  {journal} {IEEE Transactions on Information Theory}\ }\textbf {\bibinfo
  {volume} {60}},\ \bibinfo {pages} {7317--7329} (\bibinfo {year}
  {2014})}\BibitemShut {NoStop}%
\bibitem [{\citenamefont {Buscemi}\ and\ \citenamefont
  {Datta}(2010)}]{Buscemi2010IEEETIT_capacity}%
  \BibitemOpen
  \bibfield  {author} {\bibinfo {author} {\bibfnamefont {F.}~\bibnamefont
  {Buscemi}}\ and\ \bibinfo {author} {\bibfnamefont {N.}~\bibnamefont
  {Datta}},\ }\bibfield  {title} {\bibinfo {title} {{T}he quantum capacity of
  channels with arbitrarily correlated noise},\ }\href
  {https://doi.org/10.1109/TIT.2009.2039166} {\bibfield  {journal} {\bibinfo
  {journal} {IEEE Transactions on Information Theory}\ }\textbf {\bibinfo
  {volume} {56}},\ \bibinfo {pages} {1447--1460} (\bibinfo {year}
  {2010})}\BibitemShut {NoStop}%
\bibitem [{\citenamefont {Brand{\~a}o}\ and\ \citenamefont
  {Datta}(2011)}]{Brandao2011IEEETIT_oneshot}%
  \BibitemOpen
  \bibfield  {author} {\bibinfo {author} {\bibfnamefont {F.~G. S.~L.}\
  \bibnamefont {Brand{\~a}o}}\ and\ \bibinfo {author} {\bibfnamefont
  {N.}~\bibnamefont {Datta}},\ }\bibfield  {title} {\bibinfo {title}
  {{O}ne-shot rates for entanglement manipulation under non-entangling maps},\
  }\href {https://doi.org/10.1109/TIT.2011.2104531} {\bibfield  {journal}
  {\bibinfo  {journal} {IEEE Transactions on Information Theory}\ }\textbf
  {\bibinfo {volume} {57}},\ \bibinfo {pages} {1754--1760} (\bibinfo {year}
  {2011})}\BibitemShut {NoStop}%
\bibitem [{\citenamefont {Dupuis}\ \emph {et~al.}(2013)\citenamefont {Dupuis},
  \citenamefont {Kraemer}, \citenamefont {Faist}, \citenamefont {Renes},\ and\
  \citenamefont {Renner}}]{Dupuis2013_DH}%
  \BibitemOpen
  \bibfield  {author} {\bibinfo {author} {\bibfnamefont {F.}~\bibnamefont
  {Dupuis}}, \bibinfo {author} {\bibfnamefont {L.}~\bibnamefont {Kraemer}},
  \bibinfo {author} {\bibfnamefont {P.}~\bibnamefont {Faist}}, \bibinfo
  {author} {\bibfnamefont {J.~M.}\ \bibnamefont {Renes}},\ and\ \bibinfo
  {author} {\bibfnamefont {R.}~\bibnamefont {Renner}},\ }\bibfield  {title}
  {\bibinfo {title} {{G}eneralized entropies},\ }in\ \href
  {https://doi.org/10.1142/9789814449243_0008} {\emph {\bibinfo {booktitle}
  {XVIIth International Congress on Mathematical Physics}}}\ (\bibinfo
  {publisher} {World Scientific},\ \bibinfo {address} {Singapore},\ \bibinfo
  {year} {2013})\ pp.\ \bibinfo {pages} {134--153},\ \Eprint
  {https://arxiv.org/abs/1211.3141} {arXiv:1211.3141} \BibitemShut {NoStop}%
\bibitem [{\citenamefont {Watrous}(2009)}]{Watrous2009_sdps}%
  \BibitemOpen
  \bibfield  {author} {\bibinfo {author} {\bibfnamefont {J.}~\bibnamefont
  {Watrous}},\ }\bibfield  {title} {\bibinfo {title} {{S}emidefinite programs
  for completely bounded norms},\ }\href
  {https://doi.org/10.4086/toc.2009.v005a011} {\bibfield  {journal} {\bibinfo
  {journal} {Theory of Computing}\ }\textbf {\bibinfo {volume} {5}},\ \bibinfo
  {pages} {217--238} (\bibinfo {year} {2009})}\BibitemShut {NoStop}%
\bibitem [{\citenamefont {Szilard}(1929)}]{Szilard1929ZeitschriftFuerPhysik}%
  \BibitemOpen
  \bibfield  {author} {\bibinfo {author} {\bibfnamefont {L.}~\bibnamefont
  {Szilard}},\ }\bibfield  {title} {{\selectlanguage {german}\bibinfo {title}
  {{\"U}ber die entropieverminderung in einem thermodynamischen system bei
  eingriffen intelligenter wesen}},\ }\href
  {https://doi.org/10.1007/BF01341281} {\bibfield  {journal} {\bibinfo
  {journal} {Zeitschrift f{\"u}r Physik}\ }\textbf {\bibinfo {volume} {53}},\
  \bibinfo {pages} {840--856} (\bibinfo {year} {1929})}\BibitemShut {NoStop}%
\bibitem [{\citenamefont {Boyd}\ and\ \citenamefont
  {Vandenberghe}(2004)}]{BookBoyd2004ConvexOptimization}%
  \BibitemOpen
  \bibfield  {author} {\bibinfo {author} {\bibfnamefont {S.~P.}\ \bibnamefont
  {Boyd}}\ and\ \bibinfo {author} {\bibfnamefont {L.}~\bibnamefont
  {Vandenberghe}},\ }\href@noop {} {\emph {\bibinfo {title} {{C}onvex
  Optimization}}}\ (\bibinfo  {publisher} {Cambridge University Press},\
  \bibinfo {address} {Cambridge, UK},\ \bibinfo {year} {2004})\BibitemShut
  {NoStop}%
\bibitem [{\citenamefont {Pitchford}\ \emph {et~al.}(2016)\citenamefont
  {Pitchford}, \citenamefont {Granade}, \citenamefont {Nation},\ and\
  \citenamefont {Johansson}}]{QuTipPy_4_1_0}%
  \BibitemOpen
  \bibfield  {author} {\bibinfo {author} {\bibfnamefont {A.}~\bibnamefont
  {Pitchford}}, \bibinfo {author} {\bibfnamefont {C.}~\bibnamefont {Granade}},
  \bibinfo {author} {\bibfnamefont {P.~D.}\ \bibnamefont {Nation}},\ and\
  \bibinfo {author} {\bibfnamefont {R.~J.}\ \bibnamefont {Johansson}},\ }\href
  {http://qutip.org} {\bibinfo {title} {{QuTip} 4.1.0}} (\bibinfo {year}
  {2016})\BibitemShut {NoStop}%
\bibitem [{\citenamefont {Johansson}\ \emph {et~al.}(2013)\citenamefont
  {Johansson}, \citenamefont {Nation},\ and\ \citenamefont
  {Nori}}]{Johansson2013CPC_QuTip2}%
  \BibitemOpen
  \bibfield  {author} {\bibinfo {author} {\bibfnamefont {J.}~\bibnamefont
  {Johansson}}, \bibinfo {author} {\bibfnamefont {P.}~\bibnamefont {Nation}},\
  and\ \bibinfo {author} {\bibfnamefont {F.}~\bibnamefont {Nori}},\ }\bibfield
  {title} {\bibinfo {title} {{QuTip} 2: A {Python} framework for the dynamics
  of open quantum systems},\ }\href {https://doi.org/10.1016/j.cpc.2012.11.019}
  {\bibfield  {journal} {\bibinfo  {journal} {Computer Physics Communications}\
  }\textbf {\bibinfo {volume} {184}},\ \bibinfo {pages} {1234--1240} (\bibinfo
  {year} {2013})}\BibitemShut {NoStop}%
\bibitem [{\citenamefont {Andersen}\ \emph {et~al.}(2016)\citenamefont
  {Andersen}, \citenamefont {Dahl},\ and\ \citenamefont
  {Vandenberghe}}]{CVXOPTPy_1_1_9}%
  \BibitemOpen
  \bibfield  {author} {\bibinfo {author} {\bibfnamefont {M.~S.}\ \bibnamefont
  {Andersen}}, \bibinfo {author} {\bibfnamefont {J.}~\bibnamefont {Dahl}},\
  and\ \bibinfo {author} {\bibfnamefont {L.}~\bibnamefont {Vandenberghe}},\
  }\href {https://cvxopt.org/} {\bibinfo {title} {{CVXOPT} 1.1.9}} (\bibinfo
  {year} {2016})\BibitemShut {NoStop}%
\bibitem [{\citenamefont {Ramakrishnan}\ \emph {et~al.}(2021)\citenamefont
  {Ramakrishnan}, \citenamefont {Iten}, \citenamefont {Scholz},\ and\
  \citenamefont {Berta}}]{Berta21}%
  \BibitemOpen
  \bibfield  {author} {\bibinfo {author} {\bibfnamefont {N.}~\bibnamefont
  {Ramakrishnan}}, \bibinfo {author} {\bibfnamefont {R.}~\bibnamefont {Iten}},
  \bibinfo {author} {\bibfnamefont {V.~B.}\ \bibnamefont {Scholz}},\ and\
  \bibinfo {author} {\bibfnamefont {M.}~\bibnamefont {Berta}},\ }\bibfield
  {title} {\bibinfo {title} {{C}omputing quantum channel capacities},\ }\href
  {https://doi.org/10.1109/TIT.2020.3034471} {\bibfield  {journal} {\bibinfo
  {journal} {IEEE Transactions on Information Theory}\ }\textbf {\bibinfo
  {volume} {67}},\ \bibinfo {pages} {946--960} (\bibinfo {year}
  {2021})}\BibitemShut {NoStop}%
\bibitem [{\citenamefont {Alicki}(2004)}]{Alicki2004arXiv_isotropic}%
  \BibitemOpen
  \bibfield  {author} {\bibinfo {author} {\bibfnamefont {R.}~\bibnamefont
  {Alicki}},\ }\bibfield  {title} {\bibinfo {title} {{I}sotropic quantum spin
  channels and additivity questions}} (\bibinfo {year} {2004}),\ \bibinfo
  {note}
  {\href{https://arxiv.org/abs/quant-ph/0402080}{arXiv:quant-ph/0402080}}\BibitemShut
  {NoStop}%
\bibitem [{\citenamefont {Devetak}\ \emph {et~al.}(2006)\citenamefont
  {Devetak}, \citenamefont {Junge}, \citenamefont {King},\ and\ \citenamefont
  {Ruskai}}]{Devetak2006CMP_multiplicativity}%
  \BibitemOpen
  \bibfield  {author} {\bibinfo {author} {\bibfnamefont {I.}~\bibnamefont
  {Devetak}}, \bibinfo {author} {\bibfnamefont {M.}~\bibnamefont {Junge}},
  \bibinfo {author} {\bibfnamefont {C.}~\bibnamefont {King}},\ and\ \bibinfo
  {author} {\bibfnamefont {M.~B.}\ \bibnamefont {Ruskai}},\ }\bibfield  {title}
  {\bibinfo {title} {{M}ultiplicativity of completely bounded p-norms implies a
  new additivity result},\ }\href {https://doi.org/10.1007/s00220-006-0034-0}
  {\bibfield  {journal} {\bibinfo  {journal} {Communications in Mathematical
  Physics}\ }\textbf {\bibinfo {volume} {266}},\ \bibinfo {pages} {37--63}
  (\bibinfo {year} {2006})}\BibitemShut {NoStop}%
\bibitem [{\citenamefont
  {Holevo}(2011{\natexlab{a}})}]{Holevo2011ISIT_entropygain}%
  \BibitemOpen
  \bibfield  {author} {\bibinfo {author} {\bibfnamefont {A.~S.}\ \bibnamefont
  {Holevo}},\ }\bibfield  {title} {\bibinfo {title} {{T}he entropy gain of
  quantum channels},\ }in\ \href {https://doi.org/10.1109/ISIT.2011.6034107}
  {\emph {\bibinfo {booktitle} {Proceedings of the 2011 IEEE International
  Symposium on Information Theory}}},\ Vol.~\bibinfo {volume} {82}\ (\bibinfo
  {publisher} {IEEE},\ \bibinfo {year} {2011})\ pp.\ \bibinfo {pages}
  {289--292}\BibitemShut {NoStop}%
\bibitem [{\citenamefont {Holevo}(2010)}]{Holevo2010DM_infinitedim}%
  \BibitemOpen
  \bibfield  {author} {\bibinfo {author} {\bibfnamefont {A.~S.}\ \bibnamefont
  {Holevo}},\ }\bibfield  {title} {\bibinfo {title} {{T}he entropy gain of
  infinite-dimensional quantum evolutions},\ }\href
  {https://doi.org/10.1134/S1064562410050133} {\bibfield  {journal} {\bibinfo
  {journal} {Doklady Mathematics}\ }\textbf {\bibinfo {volume} {82}},\ \bibinfo
  {pages} {730--731} (\bibinfo {year} {2010})}\BibitemShut {NoStop}%
\bibitem [{\citenamefont {Holevo}(2011{\natexlab{b}})}]{Holevo2011TMP_CJ}%
  \BibitemOpen
  \bibfield  {author} {\bibinfo {author} {\bibfnamefont {A.~S.}\ \bibnamefont
  {Holevo}},\ }\bibfield  {title} {\bibinfo {title} {{E}ntropy gain and the
  {Choi}-{Jamiolkowski} correspondence for infinite-dimensional quantum
  evolutions},\ }\href {https://doi.org/10.1007/s11232-011-0010-5} {\bibfield
  {journal} {\bibinfo  {journal} {Theoretical and Mathematical Physics}\
  }\textbf {\bibinfo {volume} {166}},\ \bibinfo {pages} {123--138} (\bibinfo
  {year} {2011}{\natexlab{b}})}\BibitemShut {NoStop}%
\bibitem [{\citenamefont
  {Holevo}(2012)}]{BookHolevo2012_QuSystemsChannelsInformation}%
  \BibitemOpen
  \bibfield  {author} {\bibinfo {author} {\bibfnamefont {A.~S.}\ \bibnamefont
  {Holevo}},\ }\href {https://doi.org/10.1515/9783110273403} {\emph {\bibinfo
  {title} {{Q}uantum Systems, Channels, Information}}}\ (\bibinfo  {publisher}
  {De Gruyter},\ \bibinfo {address} {Berlin, Boston},\ \bibinfo {year}
  {2012})\BibitemShut {NoStop}%
\bibitem [{\citenamefont {Buscemi}\ \emph {et~al.}(2016)\citenamefont
  {Buscemi}, \citenamefont {Das},\ and\ \citenamefont
  {Wilde}}]{Buscemi2016PRA_reversibility}%
  \BibitemOpen
  \bibfield  {author} {\bibinfo {author} {\bibfnamefont {F.}~\bibnamefont
  {Buscemi}}, \bibinfo {author} {\bibfnamefont {S.}~\bibnamefont {Das}},\ and\
  \bibinfo {author} {\bibfnamefont {M.~M.}\ \bibnamefont {Wilde}},\ }\bibfield
  {title} {\bibinfo {title} {{A}pproximate reversibility in the context of
  entropy gain, information gain, and complete positivity},\ }\href
  {https://doi.org/10.1103/PhysRevA.93.062314} {\bibfield  {journal} {\bibinfo
  {journal} {Physical Review A}\ }\textbf {\bibinfo {volume} {93}},\ \bibinfo
  {pages} {062314} (\bibinfo {year} {2016})}\BibitemShut {NoStop}%
\bibitem [{\citenamefont {Gour}\ and\ \citenamefont
  {Wilde}(2021)}]{Gour2021PRR_entropychannel}%
  \BibitemOpen
  \bibfield  {author} {\bibinfo {author} {\bibfnamefont {G.}~\bibnamefont
  {Gour}}\ and\ \bibinfo {author} {\bibfnamefont {M.~M.}\ \bibnamefont
  {Wilde}},\ }\bibfield  {title} {\bibinfo {title} {{E}ntropy of a quantum
  channel},\ }\href {https://doi.org/10.1103/PhysRevResearch.3.023096}
  {\bibfield  {journal} {\bibinfo  {journal} {Physical Review Research}\
  }\textbf {\bibinfo {volume} {3}},\ \bibinfo {pages} {023096} (\bibinfo {year}
  {2021})}\BibitemShut {NoStop}%
\bibitem [{\citenamefont {Berta}\ \emph {et~al.}(2014)\citenamefont {Berta},
  \citenamefont {Renes},\ and\ \citenamefont
  {Wilde}}]{Berta2014IEEETIT_InfGainMeas}%
  \BibitemOpen
  \bibfield  {author} {\bibinfo {author} {\bibfnamefont {M.}~\bibnamefont
  {Berta}}, \bibinfo {author} {\bibfnamefont {J.~M.}\ \bibnamefont {Renes}},\
  and\ \bibinfo {author} {\bibfnamefont {M.~M.}\ \bibnamefont {Wilde}},\
  }\bibfield  {title} {\bibinfo {title} {{I}dentifying the information gain of
  a quantum measurement},\ }\href {https://doi.org/10.1109/TIT.2014.2365207}
  {\bibfield  {journal} {\bibinfo  {journal} {IEEE Transactions on Information
  Theory}\ }\textbf {\bibinfo {volume} {60}},\ \bibinfo {pages} {7987--8006}
  (\bibinfo {year} {2014})}\BibitemShut {NoStop}%
\bibitem [{\citenamefont {Faist}(2016)}]{Faist16}%
  \BibitemOpen
  \bibfield  {author} {\bibinfo {author} {\bibfnamefont {P.}~\bibnamefont
  {Faist}},\ }\emph {\bibinfo {title} {{Q}uantum Coarse-Graining: An
  Information-Theoretic Approach to Thermodynamics}},\ \href
  {https://doi.org/10.3929/ethz-a-010695790} {Ph.D. thesis},\ \bibinfo
  {school} {ETH Z{\"u}rich} (\bibinfo {year} {2016}),\ \Eprint
  {https://arxiv.org/abs/1607.03104} {arXiv:1607.03104} \BibitemShut {NoStop}%
\bibitem [{\citenamefont {Morgan}\ and\ \citenamefont
  {Winter}(2014)}]{Morgan2014IEEETIT_prettystrong}%
  \BibitemOpen
  \bibfield  {author} {\bibinfo {author} {\bibfnamefont {C.}~\bibnamefont
  {Morgan}}\ and\ \bibinfo {author} {\bibfnamefont {A.}~\bibnamefont
  {Winter}},\ }\bibfield  {title} {\bibinfo {title}
  {{\textquotedblleft}{P}retty strong{\textquotedblright} converse for the
  quantum capacity of degradable channels},\ }\href
  {https://doi.org/10.1109/TIT.2013.2288971} {\bibfield  {journal} {\bibinfo
  {journal} {IEEE Transactions on Information Theory}\ }\textbf {\bibinfo
  {volume} {60}},\ \bibinfo {pages} {317--333} (\bibinfo {year}
  {2014})}\BibitemShut {NoStop}%
\bibitem [{\citenamefont {Tomamichel}\ \emph {et~al.}(2009)\citenamefont
  {Tomamichel}, \citenamefont {Colbeck},\ and\ \citenamefont
  {Renner}}]{Tomamichel2009IEEE_AEP}%
  \BibitemOpen
  \bibfield  {author} {\bibinfo {author} {\bibfnamefont {M.}~\bibnamefont
  {Tomamichel}}, \bibinfo {author} {\bibfnamefont {R.}~\bibnamefont
  {Colbeck}},\ and\ \bibinfo {author} {\bibfnamefont {R.}~\bibnamefont
  {Renner}},\ }\bibfield  {title} {\bibinfo {title} {{A} fully quantum
  asymptotic equipartition property},\ }\href
  {https://doi.org/10.1109/TIT.2009.2032797} {\bibfield  {journal} {\bibinfo
  {journal} {IEEE Transactions on Information Theory}\ }\textbf {\bibinfo
  {volume} {55}},\ \bibinfo {pages} {5840--5847} (\bibinfo {year}
  {2009})}\BibitemShut {NoStop}%
\bibitem [{\citenamefont {Bjelakovic}\ and\ \citenamefont
  {Siegmund-Schultze}(2003)}]{Bjelakovic2003arXiv_revisted}%
  \BibitemOpen
  \bibfield  {author} {\bibinfo {author} {\bibfnamefont {I.}~\bibnamefont
  {Bjelakovic}}\ and\ \bibinfo {author} {\bibfnamefont {R.}~\bibnamefont
  {Siegmund-Schultze}},\ }\bibfield  {title} {\bibinfo {title} {{Q}uantum
  {Stein}'s lemma revisited, inequalities for quantum entropies, and a
  concavity theorem of {Lieb}}} (\bibinfo {year} {2003}),\ \bibinfo {note}
  {\href{https://arxiv.org/abs/quant-ph/0307170}{arXiv:quant-ph/0307170}}\BibitemShut
  {NoStop}%
\bibitem [{\citenamefont {Berta}\ \emph {et~al.}(2015)\citenamefont {Berta},
  \citenamefont {Lemm},\ and\ \citenamefont
  {Wilde}}]{Berta2015QIC_monotonicity}%
  \BibitemOpen
  \bibfield  {author} {\bibinfo {author} {\bibfnamefont {M.}~\bibnamefont
  {Berta}}, \bibinfo {author} {\bibfnamefont {M.}~\bibnamefont {Lemm}},\ and\
  \bibinfo {author} {\bibfnamefont {M.~M.}\ \bibnamefont {Wilde}},\ }\bibfield
  {title} {\bibinfo {title} {{M}onotonicity of quantum relative entropy and
  recoverability},\ }\href@noop {} {\bibfield  {journal} {\bibinfo  {journal}
  {Quantum Information and Computation}\ }\textbf {\bibinfo {volume} {15}},\
  \bibinfo {pages} {1333--1354} (\bibinfo {year} {2015})}\BibitemShut {NoStop}%
\bibitem [{\citenamefont {Anshu}\ \emph {et~al.}(2017)\citenamefont {Anshu},
  \citenamefont {Devabathini},\ and\ \citenamefont
  {Jain}}]{Anshu2017PRL_convexsplit}%
  \BibitemOpen
  \bibfield  {author} {\bibinfo {author} {\bibfnamefont {A.}~\bibnamefont
  {Anshu}}, \bibinfo {author} {\bibfnamefont {V.~K.}\ \bibnamefont
  {Devabathini}},\ and\ \bibinfo {author} {\bibfnamefont {R.}~\bibnamefont
  {Jain}},\ }\bibfield  {title} {\bibinfo {title} {{Q}uantum communication
  using coherent rejection sampling},\ }\href
  {https://doi.org/10.1103/PhysRevLett.119.120506} {\bibfield  {journal}
  {\bibinfo  {journal} {Physical Review Letters}\ }\textbf {\bibinfo {volume}
  {119}},\ \bibinfo {pages} {120506} (\bibinfo {year} {2017})}\BibitemShut
  {NoStop}%
\bibitem [{\citenamefont {Anshu}\ \emph
  {et~al.}(2018{\natexlab{a}})\citenamefont {Anshu}, \citenamefont {Jain},\
  and\ \citenamefont {Warsi}}]{Anshu2018IEEETIT_redistribution}%
  \BibitemOpen
  \bibfield  {author} {\bibinfo {author} {\bibfnamefont {A.}~\bibnamefont
  {Anshu}}, \bibinfo {author} {\bibfnamefont {R.}~\bibnamefont {Jain}},\ and\
  \bibinfo {author} {\bibfnamefont {N.~A.}\ \bibnamefont {Warsi}},\ }\bibfield
  {title} {\bibinfo {title} {{A} one-shot achievability result for quantum
  state redistribution},\ }\href {https://doi.org/10.1109/TIT.2017.2776112}
  {\bibfield  {journal} {\bibinfo  {journal} {IEEE Transactions on Information
  Theory}\ }\textbf {\bibinfo {volume} {64}},\ \bibinfo {pages} {1425--1435}
  (\bibinfo {year} {2018}{\natexlab{a}})}\BibitemShut {NoStop}%
\bibitem [{\citenamefont {Anshu}\ \emph
  {et~al.}(2018{\natexlab{b}})\citenamefont {Anshu}, \citenamefont {Jain},\
  and\ \citenamefont {Warsi}}]{Anshu2018IEEETIT_SlepianWolf}%
  \BibitemOpen
  \bibfield  {author} {\bibinfo {author} {\bibfnamefont {A.}~\bibnamefont
  {Anshu}}, \bibinfo {author} {\bibfnamefont {R.}~\bibnamefont {Jain}},\ and\
  \bibinfo {author} {\bibfnamefont {N.~A.}\ \bibnamefont {Warsi}},\ }\bibfield
  {title} {\bibinfo {title} {{A} generalized quantum
  {Slepian}{\textendash}{Wolf}},\ }\href
  {https://doi.org/10.1109/TIT.2017.2786348} {\bibfield  {journal} {\bibinfo
  {journal} {IEEE Transactions on Information Theory}\ }\textbf {\bibinfo
  {volume} {64}},\ \bibinfo {pages} {1436--1453} (\bibinfo {year}
  {2018}{\natexlab{b}})}\BibitemShut {NoStop}%
\bibitem [{\citenamefont {Anshu}\ \emph
  {et~al.}(2019{\natexlab{b}})\citenamefont {Anshu}, \citenamefont {Jain},\
  and\ \citenamefont {Warsi}}]{Anshu2019IEEETIT_compression}%
  \BibitemOpen
  \bibfield  {author} {\bibinfo {author} {\bibfnamefont {A.}~\bibnamefont
  {Anshu}}, \bibinfo {author} {\bibfnamefont {R.}~\bibnamefont {Jain}},\ and\
  \bibinfo {author} {\bibfnamefont {N.~A.}\ \bibnamefont {Warsi}},\ }\bibfield
  {title} {\bibinfo {title} {{C}onvex-split and hypothesis testing approach to
  one-shot quantum measurement compression and randomness extraction},\ }\href
  {https://doi.org/10.1109/TIT.2019.2915242} {\bibfield  {journal} {\bibinfo
  {journal} {IEEE Transactions on Information Theory}\ }\textbf {\bibinfo
  {volume} {65}},\ \bibinfo {pages} {5905--5924} (\bibinfo {year}
  {2019}{\natexlab{b}})}\BibitemShut {NoStop}%
\bibitem [{\citenamefont {Majenz}\ \emph {et~al.}(2017)\citenamefont {Majenz},
  \citenamefont {Berta}, \citenamefont {Dupuis}, \citenamefont {Renner},\ and\
  \citenamefont {Christandl}}]{Majenz2017PRL_catdecoupling}%
  \BibitemOpen
  \bibfield  {author} {\bibinfo {author} {\bibfnamefont {C.}~\bibnamefont
  {Majenz}}, \bibinfo {author} {\bibfnamefont {M.}~\bibnamefont {Berta}},
  \bibinfo {author} {\bibfnamefont {F.}~\bibnamefont {Dupuis}}, \bibinfo
  {author} {\bibfnamefont {R.}~\bibnamefont {Renner}},\ and\ \bibinfo {author}
  {\bibfnamefont {M.}~\bibnamefont {Christandl}},\ }\bibfield  {title}
  {\bibinfo {title} {{C}atalytic decoupling of quantum information},\ }\href
  {https://doi.org/10.1103/PhysRevLett.118.080503} {\bibfield  {journal}
  {\bibinfo  {journal} {Physical Review Letters}\ }\textbf {\bibinfo {volume}
  {118}},\ \bibinfo {pages} {080503} (\bibinfo {year} {2017})}\BibitemShut
  {NoStop}%
\bibitem [{\citenamefont {Anshu}\ \emph {et~al.}(2020)\citenamefont {Anshu},
  \citenamefont {Berta}, \citenamefont {Jain},\ and\ \citenamefont
  {Tomamichel}}]{Anshu2020ITIT_partially}%
  \BibitemOpen
  \bibfield  {author} {\bibinfo {author} {\bibfnamefont {A.}~\bibnamefont
  {Anshu}}, \bibinfo {author} {\bibfnamefont {M.}~\bibnamefont {Berta}},
  \bibinfo {author} {\bibfnamefont {R.}~\bibnamefont {Jain}},\ and\ \bibinfo
  {author} {\bibfnamefont {M.}~\bibnamefont {Tomamichel}},\ }\bibfield  {title}
  {\bibinfo {title} {{P}artially smoothed information measures},\ }\href
  {https://doi.org/10.1109/TIT.2020.2981573} {\bibfield  {journal} {\bibinfo
  {journal} {IEEE Transactions on Information Theory}\ }\textbf {\bibinfo
  {volume} {66}},\ \bibinfo {pages} {5022--5036} (\bibinfo {year}
  {2020})}\BibitemShut {NoStop}%
\bibitem [{\citenamefont {Berta}\ and\ \citenamefont {Majenz}(2018)}]{Berta18}%
  \BibitemOpen
  \bibfield  {author} {\bibinfo {author} {\bibfnamefont {M.}~\bibnamefont
  {Berta}}\ and\ \bibinfo {author} {\bibfnamefont {C.}~\bibnamefont {Majenz}},\
  }\bibfield  {title} {\bibinfo {title} {{D}isentanglement cost of quantum
  states},\ }\href {https://doi.org/10.1103/PhysRevLett.121.190503} {\bibfield
  {journal} {\bibinfo  {journal} {Physical Review Letters}\ }\textbf {\bibinfo
  {volume} {121}},\ \bibinfo {pages} {190503} (\bibinfo {year}
  {2018})}\BibitemShut {NoStop}%
\bibitem [{\citenamefont {del Rio}\ \emph {et~al.}(2011)\citenamefont {del
  Rio}, \citenamefont {{\r{A}}berg}, \citenamefont {Renner}, \citenamefont
  {Dahlsten},\ and\ \citenamefont {Vedral}}]{delRio2011Nature}%
  \BibitemOpen
  \bibfield  {author} {\bibinfo {author} {\bibfnamefont {L.}~\bibnamefont {del
  Rio}}, \bibinfo {author} {\bibfnamefont {J.}~\bibnamefont {{\r{A}}berg}},
  \bibinfo {author} {\bibfnamefont {R.}~\bibnamefont {Renner}}, \bibinfo
  {author} {\bibfnamefont {O.}~\bibnamefont {Dahlsten}},\ and\ \bibinfo
  {author} {\bibfnamefont {V.}~\bibnamefont {Vedral}},\ }\bibfield  {title}
  {\bibinfo {title} {{T}he thermodynamic meaning of negative entropy},\ }\href
  {https://doi.org/10.1038/nature10123} {\bibfield  {journal} {\bibinfo
  {journal} {Nature}\ }\textbf {\bibinfo {volume} {474}},\ \bibinfo {pages}
  {61--63} (\bibinfo {year} {2011})}\BibitemShut {NoStop}%
\bibitem [{\citenamefont {Hayashi}\ and\ \citenamefont
  {Nagaoka}(2003)}]{Hayashi2003IEEETIT_formulas}%
  \BibitemOpen
  \bibfield  {author} {\bibinfo {author} {\bibfnamefont {M.}~\bibnamefont
  {Hayashi}}\ and\ \bibinfo {author} {\bibfnamefont {H.}~\bibnamefont
  {Nagaoka}},\ }\bibfield  {title} {\bibinfo {title} {{G}eneral formulas for
  capacity of classical-quantum channels},\ }\href
  {https://doi.org/10.1109/TIT.2003.813556} {\bibfield  {journal} {\bibinfo
  {journal} {IEEE Transactions on Information Theory}\ }\textbf {\bibinfo
  {volume} {49}},\ \bibinfo {pages} {1753--1768} (\bibinfo {year}
  {2003})}\BibitemShut {NoStop}%
\bibitem [{\citenamefont {Scutaru}(1979)}]{Scutaru1979RMP_covariant}%
  \BibitemOpen
  \bibfield  {author} {\bibinfo {author} {\bibfnamefont {H.}~\bibnamefont
  {Scutaru}},\ }\bibfield  {title} {\bibinfo {title} {{S}ome remarks on
  covariant completely positive linear maps on
  {C$^{\ensuremath{*}}$}-algebras},\ }\href
  {https://doi.org/10.1016/0034-4877(79)90040-5} {\bibfield  {journal}
  {\bibinfo  {journal} {Reports on Mathematical Physics}\ }\textbf {\bibinfo
  {volume} {16}},\ \bibinfo {pages} {79--87} (\bibinfo {year}
  {1979})}\BibitemShut {NoStop}%
\bibitem [{\citenamefont {Keyl}\ and\ \citenamefont
  {Werner}(1999)}]{Keyl1999JMP_cloning}%
  \BibitemOpen
  \bibfield  {author} {\bibinfo {author} {\bibfnamefont {M.}~\bibnamefont
  {Keyl}}\ and\ \bibinfo {author} {\bibfnamefont {R.~F.}\ \bibnamefont
  {Werner}},\ }\bibfield  {title} {\bibinfo {title} {{O}ptimal cloning of pure
  states, testing single clones},\ }\href {https://doi.org/10.1063/1.532887}
  {\bibfield  {journal} {\bibinfo  {journal} {Journal of Mathematical Physics}\
  }\textbf {\bibinfo {volume} {40}},\ \bibinfo {pages} {3283--3299} (\bibinfo
  {year} {1999})}\BibitemShut {NoStop}%
\bibitem [{\citenamefont {Marvian~Mashhad}(2012)}]{PhDMarvian2012_symmetry}%
  \BibitemOpen
  \bibfield  {author} {\bibinfo {author} {\bibfnamefont {I.}~\bibnamefont
  {Marvian~Mashhad}},\ }\emph {\bibinfo {title} {{S}ymmetry, Asymmetry and
  Quantum Information}},\ \href {https://hdl.handle.net/10012/7088} {Ph.D.
  thesis},\ \bibinfo  {school} {University of Waterloo} (\bibinfo {year}
  {2012})\BibitemShut {NoStop}%
\bibitem [{\citenamefont {{\r{A}}berg}(2014)}]{Aberg2014PRL_catalytic}%
  \BibitemOpen
  \bibfield  {author} {\bibinfo {author} {\bibfnamefont {J.}~\bibnamefont
  {{\r{A}}berg}},\ }\bibfield  {title} {\bibinfo {title} {{C}atalytic
  coherence},\ }\href {https://doi.org/10.1103/PhysRevLett.113.150402}
  {\bibfield  {journal} {\bibinfo  {journal} {Physical Review Letters}\
  }\textbf {\bibinfo {volume} {113}},\ \bibinfo {pages} {150402} (\bibinfo
  {year} {2014})}\BibitemShut {NoStop}%
\bibitem [{\citenamefont {Fang}\ \emph {et~al.}(2020)\citenamefont {Fang},
  \citenamefont {Wang}, \citenamefont {Tomamichel},\ and\ \citenamefont
  {Berta}}]{Fang2019IEEETIT_simulation}%
  \BibitemOpen
  \bibfield  {author} {\bibinfo {author} {\bibfnamefont {K.}~\bibnamefont
  {Fang}}, \bibinfo {author} {\bibfnamefont {X.}~\bibnamefont {Wang}}, \bibinfo
  {author} {\bibfnamefont {M.}~\bibnamefont {Tomamichel}},\ and\ \bibinfo
  {author} {\bibfnamefont {M.}~\bibnamefont {Berta}},\ }\bibfield  {title}
  {\bibinfo {title} {{Q}uantum channel simulation and the
  channel{\textquoteright}s smooth max-information},\ }\href
  {https://doi.org/10.1109/TIT.2019.2943858} {\bibfield  {journal} {\bibinfo
  {journal} {IEEE Transactions on Information Theory}\ }\textbf {\bibinfo
  {volume} {66}},\ \bibinfo {pages} {2129 -- 2140} (\bibinfo {year}
  {2020})}\BibitemShut {NoStop}%
\bibitem [{\citenamefont {Gour}\ and\ \citenamefont {Winter}(2019)}]{Gour19}%
  \BibitemOpen
  \bibfield  {author} {\bibinfo {author} {\bibfnamefont {G.}~\bibnamefont
  {Gour}}\ and\ \bibinfo {author} {\bibfnamefont {A.}~\bibnamefont {Winter}},\
  }\bibfield  {title} {\bibinfo {title} {{H}ow to quantify a dynamical
  resource?},\ }\href {https://doi.org/10.1103/PhysRevLett.123.150401}
  {\bibfield  {journal} {\bibinfo  {journal} {Physical Review Letters}\
  }\textbf {\bibinfo {volume} {123}},\ \bibinfo {pages} {150401} (\bibinfo
  {year} {2019})}\BibitemShut {NoStop}%
\bibitem [{\citenamefont {Dutil}(2011)}]{PhDDutil2011_multiparty}%
  \BibitemOpen
  \bibfield  {author} {\bibinfo {author} {\bibfnamefont {N.}~\bibnamefont
  {Dutil}},\ }\emph {\bibinfo {title} {{M}ultiparty quantum protocols for
  assisted entanglement distillation}},\ \href@noop {} {Ph.D. thesis},\
  \bibinfo  {school} {McGill University, Montr{\'e}al} (\bibinfo {year}
  {2011}),\ \Eprint {https://arxiv.org/abs/1105.4657} {arXiv:1105.4657}
  \BibitemShut {NoStop}%
\bibitem [{\citenamefont {Drescher}\ and\ \citenamefont
  {Fawzi}(2013)}]{Drescher2013ISIT_simultaneous}%
  \BibitemOpen
  \bibfield  {author} {\bibinfo {author} {\bibfnamefont {L.}~\bibnamefont
  {Drescher}}\ and\ \bibinfo {author} {\bibfnamefont {O.}~\bibnamefont
  {Fawzi}},\ }\bibfield  {title} {\bibinfo {title} {{O}n simultaneous
  min-entropy smoothing},\ }in\ \href
  {https://doi.org/10.1109/ISIT.2013.6620208} {\emph {\bibinfo {booktitle}
  {2013 IEEE International Symposium on Information Theory}}}\ (\bibinfo
  {publisher} {IEEE},\ \bibinfo {year} {2013})\ pp.\ \bibinfo {pages}
  {161--165},\ \Eprint {https://arxiv.org/abs/1312.7642} {arXiv:1312.7642}
  \BibitemShut {NoStop}%
\bibitem [{\citenamefont {Sen}(2018)}]{Sen2018arXiv_joint}%
  \BibitemOpen
  \bibfield  {author} {\bibinfo {author} {\bibfnamefont {P.}~\bibnamefont
  {Sen}},\ }\bibfield  {title} {\bibinfo {title} {{A} one-shot quantum joint
  typicality lemma}} (\bibinfo {year} {2018}),\ \bibinfo {note}
  {\href{https://arxiv.org/abs/1806.07278}{arXiv:1806.07278}}\BibitemShut
  {NoStop}%
\bibitem [{\citenamefont {Anshu}\ \emph
  {et~al.}(2019{\natexlab{c}})\citenamefont {Anshu}, \citenamefont {Berta},
  \citenamefont {Jain},\ and\ \citenamefont {Tomamichel}}]{Anshu18-2}%
  \BibitemOpen
  \bibfield  {author} {\bibinfo {author} {\bibfnamefont {A.}~\bibnamefont
  {Anshu}}, \bibinfo {author} {\bibfnamefont {M.}~\bibnamefont {Berta}},
  \bibinfo {author} {\bibfnamefont {R.}~\bibnamefont {Jain}},\ and\ \bibinfo
  {author} {\bibfnamefont {M.}~\bibnamefont {Tomamichel}},\ }\bibfield  {title}
  {\bibinfo {title} {{A} minimax approach to one-shot entropy inequalities},\
  }\bibfield  {journal} {\bibinfo  {journal} {Journal of Mathematical Physics}\
  }\textbf {\bibinfo {volume} {60}},\ \href {https://doi.org/10.1063/1.5126723}
  {10.1063/1.5126723} (\bibinfo {year} {2019}{\natexlab{c}})\BibitemShut
  {NoStop}%
\bibitem [{\citenamefont {Fannes}(1973)}]{Fannes1973CMP_continuity}%
  \BibitemOpen
  \bibfield  {author} {\bibinfo {author} {\bibfnamefont {M.}~\bibnamefont
  {Fannes}},\ }\bibfield  {title} {\bibinfo {title} {{A} continuity property of
  the entropy density for spin lattice systems},\ }\href
  {https://doi.org/10.1007/BF01646490} {\bibfield  {journal} {\bibinfo
  {journal} {Communications in Mathematical Physics}\ }\textbf {\bibinfo
  {volume} {31}},\ \bibinfo {pages} {291--294} (\bibinfo {year}
  {1973})}\BibitemShut {NoStop}%
\bibitem [{\citenamefont {Audenaert}(2007)}]{Audenaert2007JPA_sharp}%
  \BibitemOpen
  \bibfield  {author} {\bibinfo {author} {\bibfnamefont {K.~M.~R.}\
  \bibnamefont {Audenaert}},\ }\bibfield  {title} {\bibinfo {title} {{A} sharp
  continuity estimate for the von {Neumann} entropy},\ }\href
  {https://doi.org/10.1088/1751-8113/40/28/S18} {\bibfield  {journal} {\bibinfo
   {journal} {Journal of Physics A: Mathematical and Theoretical}\ }\textbf
  {\bibinfo {volume} {40}},\ \bibinfo {pages} {8127--8136} (\bibinfo {year}
  {2007})}\BibitemShut {NoStop}%
\bibitem [{\citenamefont {Berta}\ \emph {et~al.}(2010)\citenamefont {Berta},
  \citenamefont {Christandl}, \citenamefont {Colbeck}, \citenamefont {Renes},\
  and\ \citenamefont {Renner}}]{Berta2010_uncertainty}%
  \BibitemOpen
  \bibfield  {author} {\bibinfo {author} {\bibfnamefont {M.}~\bibnamefont
  {Berta}}, \bibinfo {author} {\bibfnamefont {M.}~\bibnamefont {Christandl}},
  \bibinfo {author} {\bibfnamefont {R.}~\bibnamefont {Colbeck}}, \bibinfo
  {author} {\bibfnamefont {J.~M.}\ \bibnamefont {Renes}},\ and\ \bibinfo
  {author} {\bibfnamefont {R.}~\bibnamefont {Renner}},\ }\bibfield  {title}
  {\bibinfo {title} {{T}he uncertainty principle in the presence of quantum
  memory},\ }\href {https://doi.org/10.1038/nphys1734} {\bibfield  {journal}
  {\bibinfo  {journal} {Nature Physics}\ }\textbf {\bibinfo {volume} {6}},\
  \bibinfo {pages} {659--662} (\bibinfo {year} {2010})}\BibitemShut {NoStop}%
\end{thebibliography}
\end{document}